\documentclass[a4paper,UKenglish]{lipics-v2018}
\newcolumntype{L}[1]{>{\raggedright\let\newline\\\arraybackslash\hspace{0pt}}m{#1}}
\newcolumntype{C}[1]{>{\centering\let\newline\\\arraybackslash\hspace{0pt}}m{#1}}
\newcolumntype{R}[1]{>{\raggedleft\let\newline\\\arraybackslash\hspace{0pt}}m{#1}}

\usepackage{microtype}%if unwanted, comment out or use option "draft"
\nolinenumbers
\usepackage{complexity}
\title{A $\frac{3}{2}$-approximation algorithm for the Student-Project Allocation problem}

\titlerunning{A $\frac{3}{2}$-approximation algorithm for the Student-Project Allocation problem}%optional, please use if title is longer than one line
\usepackage{placeins}
\author{Frances Cooper}{School of Computing Science, University of Glasgow\\{[Glasgow, Scotland, UK]}}{f.cooper.1@research.gla.ac.uk}{[https://orcid.org/0000-0001-6363-9002]}{[Supported by an Engineering and Physical Sciences Research Council Doctoral Training Account]}%mandatory, please use full name; only 1 author per \author macro; first two parameters are mandatory, other parameters can be empty.

\author{David Manlove}{School of Computing Science, University of Glasgow\\{[Glasgow, Scotland, UK]}}{david.manlove@glasgow.ac.uk}{[https://orcid.org/0000-0001-6754-7308]}{[Supported by Engineering and Physical Sciences Research Council grant EP/P028306/01]}

\authorrunning{F. Cooper and D. Manlove}%mandatory. First: Use abbreviated first/middle names. Second (only in severe cases): Use first author plus 'et. al.'

\Copyright{Frances Cooper and David Manlove}%mandatory, please use full first names. LIPIcs license is "CC-BY";  http://creativecommons.org/licenses/by/3.0/

\subjclass{Theory of computation $\rightarrow$ Design and analysis of algorithms}% mandatory: Please choose ACM 2012 classifications from https://www.acm.org/publications/class-2012 or https://dl.acm.org/ccs/ccs_flat.cfm . E.g., cite as "General and reference $\rightarrow$ General literature" or \ccsdesc[100]{General and reference~General literature}. 

%mandatory

\category{}%optional, e.g. invited paper

\relatedversion{}%optional, e.g. full version hosted on arXiv, HAL, or other respository/website

\supplement{}%optional, e.g. related research data, source code, ... hosted on a repository like zenodo, figshare, GitHub, ...

\funding{}%optional, to capture a funding statement, which applies to all authors. Please enter author specific funding statements as fifth argument of the \author macro.

\acknowledgements{}%optional

\EventEditors{Gianlorenzo D'Angelo}
\EventNoEds{1}
\EventLongTitle{17th International Symposium on Experimental Algorithms 
(SEA 2018)}
\EventShortTitle{}
\EventAcronym{SEA}
\EventYear{2018}
\EventDate{June 27--29, 2017}
\EventLocation{L'Aquila, Italy}
\EventLogo{}
\SeriesVolume{103}
\ArticleNo{8} % “New number” (=<article-no>) goes here!
\hideLIPIcs  %uncomment to remove references to LIPIcs series (logo, DOI, ...), e.g. when preparing a pre-final version to be uploaded to arXiv or another public repository
\usepackage{algorithm}
\usepackage{algpseudocode}
\usepackage[nomain, acronym]{glossaries}
\newacronym{spa}{\sc spa}{Student-Project Allocation problem}
\newacronym{spa-st}{\sc spa-st}{Student-Project Allocation problem with lecturer preferences over Students including Ties}
\newacronym{max-spa-st}{\sc max spa-st}{finding a maximum stable matching in \acrshort{spa-st}}
\newacronym{spa-s}{\sc spa-s}{Student-Project Allocation problem with lecturer preferences over Students}
\newacronym{sm}{\sc sm}{Stable Marriage problem}
\newacronym{ip}{IP}{Integer Programming}
\newacronym{ha}{\sc ha}{House Allocation problem}
\newacronym{smti}{\sc smti}{Stable Marriage problem with Ties and Incomplete lists}
\newacronym{max-smti}{\sc max smti}{finding a maximum stable matching in \acrshort{smti}}
\newacronym{spa-p}{\sc spa-p}{Student-Project Allocation problem with lecturer preferences over Projects}
\newacronym{min-spa-st}{\sc min spa-st}{The problem of finding a minimum stable matching in \acrshort{spa-st}}
\newacronym{hrt}{\sc hrt}{Hospitals-Residents problem with Ties}
\newacronym{hr}{\sc hr}{Hospitals-Residents problem}

\usepackage[colorinlistoftodos]{todonotes}

\usepackage{mdframed}

\usepackage{pdflscape}
\theoremstyle{plain}
\newtheorem{prop}[theorem]{Proposition}
\usepackage{upgreek}

\tikzset{
simpleNodes/.style={
  circle,
  inner sep=0pt,
  text width=5mm,
  text height=3mm,
  align=center,
  draw=black,
  fill=black!5
  }, 
  simpleNodesBigger/.style={
  circle,
  inner sep=0pt,
  text width=6mm,
  align=center,
  draw=black,
  fill=black!5
  },
  thinLine/.style={black, thin},
  dashedLine/.style={black, dashed, thin},
thickLine/.style={black, ultra thick}}
\newcommand{\graphShiftStd}{0.5mm}
\newcommand{\graphShiftStdL}{0.6mm}
\newcommand{\graphShiftStdXL}{1.2mm}

\newcommand{\graphxStd}{1.5}
\newcommand{\graphyStd}{1.5}

\usepackage{enumerate}

\begin{document}

\maketitle

\begin{abstract}
The \emph{Student-Project Allocation problem with lecturer preferences over Students} (\acrshort{spa-s}) comprises three sets of agents, namely students, projects and lecturers, where students have preferences over projects and lecturers have preferences over students.  In this scenario we seek a \emph{stable matching}, that is, an assignment of students to projects such that there is no student and lecturer who have an incentive to deviate from their assignee/s.  We study \acrshort{spa-st}, the extension of \acrshort{spa-s} in which the preference lists of students and lecturers need not be strictly ordered, and may contain ties. In this scenario, stable matchings may be of different sizes, and it is known that \acrshort{max-spa-st}, the problem of finding a maximum stable matching in \acrshort{spa-st}, is $\NP$-hard.  We present a linear-time $\frac{3}{2}$-approximation algorithm for \acrshort{max-spa-st} and an Integer Programming (IP) model to solve \acrshort{max-spa-st} optimally.  We compare the approximation algorithm with the IP model experimentally using randomly-generated data.  We find that the performance of the approximation algorithm easily surpassed the $\frac{3}{2}$ bound, constructing a stable matching within 92\% of optimal in all cases, with the percentage being far higher for many instances.
\end{abstract}

\section{Introduction}

\subsection{Background and motivation}
In universities all over the world, students need to be assigned to projects as part of their degree programmes.  Lecturers typically offer a range of projects, and students may rank a subset of the available projects in preference order.  Lecturers may have preferences over students, or over the projects they offer, or they may not have explicit preferences at all.  There may also be capacity constraints on the maximum numbers of students that can be allocated to each project and lecturer.  The problem of allocating students to projects subject to these preference and capacity constraints is called the \emph{\acrlong{spa}} (\acrshort{spa}) \cite[Section 5.5]{Man13}\cite{CGKM17,CFG17}.  Variants of this problem can be defined for the cases that lecturers have preferences over the students that rank their projects \cite{AIM07}, or over the projects they offer \cite{MO08}, or not at all \cite{KIMS15}.  In this paper we focus on the first of these cases, where lecturers have preferences over students -- the so-called \emph{\acrlong{spa-s}} (\acrshort{spa-s}).

Finding an optimal allocation of students to projects manually is time-consuming and error-prone.  Consequently many universities automate the allocation process using a centralised algorithm.  Given the typical sizes of problem instances (e.g., $152$ students at the School of Computing Science, University of Glasgow in $2019$), the efficiency of the matching algorithm is of paramount importance.  In the case of \acrshort{spa-s}, the desired matching must be \emph{stable} with respect to the given preference lists, meaning that no student and lecturer have an incentive to deviate from the given allocation and form an assignment with one another \cite{Rot84}.

Abraham et al.\ \cite{AIM07} described a linear-time algorithm to find a stable matching in an instance $I$ of \acrshort{spa-s} when all preference lists in $I$ are strictly ordered.  They also showed that, under this condition, all stable matchings in $I$ are of the same size.  In this paper we focus on the variant of \acrshort{spa-s} in which preference lists of students and lecturers can contain ties, which we refer to as the \emph{\acrlong{spa-st}} (\acrshort{spa-st}).  Ties allow both students and lecturers to express indifference in their preference lists (in practice, for example, lecturers may be unable to distinguish between certain groups of students).  A stable matching in an instance of \acrshort{spa-st} can be found in linear time by breaking the ties arbitrarily and using the algorithm of Abraham et al.\ \cite{AIM07}.

The \emph{\acrlong{smti}} (\acrshort{smti}) is a special case of \acrshort{spa-st} in which each project and lecturer has capacity 1, and each lecturer offers one project.  Given an instance of \acrshort{smti}, it is known that stable matchings can have different sizes \cite{MIIMM02}, and thus the same is true for \acrshort{spa-st}.  Yet in practical applications it is desirable to match as many students to projects as possible. %, to minimise the number of students who are without a project after the algorithm runs.
This motivates \acrshort{max-spa-st}, the problem of finding a maximum (cardinality) stable matching in an instance of \acrshort{spa-st}.  This problem is $\NP$-hard, since the corresponding optimisation problem restricted to \acrshort{smti}, which we refer to as \acrshort{max-smti}, is $\NP$-hard \cite{MIIMM02}.  Kir\'{a}ly \cite{Kir12} described a $\frac{3}{2}$-approximation algorithm for \acrshort{max-smti}. He also showed how to extend this algorithm to the case of the \emph{\acrlong{hrt}} (\acrshort{hrt}), where \acrshort{hrt} is the special case of \acrshort{spa-st} in which each lecturer $l$ offers one project $p$, and the capacities of $l$ and $p$ are equal.  Yanagisawa \cite{Yan07} showed that \acrshort{max-smti} is not approximable within a factor of $\frac{33}{29}$ unless $\P=\NP$; the same bound applies to \acrshort{max-spa-st}.

\subsection{Our contribution}

In this paper we describe a linear-time $\frac{3}{2}$-approximation algorithm for \acrshort{max-spa-st}. This algorithm is a non-trivial extension of Kir\'aly's approximation algorithm for \acrshort{hrt} as mentioned above.  We also describe an \acrlong{ip} (\acrshort{ip}) model to solve \acrshort{max-spa-st} optimally.  Through a series of experiments on randomly-generated data, we then compare the sizes of stable matchings output by our approximation algorithm with the sizes of optimal solutions obtained from our IP model.  Our main finding is that the performance of the approximation algorithm easily surpassed the $\frac{3}{2}$ bound on the generated instances, constructing a stable matching within 92\% of optimal in all cases, with the percentage being far higher for many instances.

A natural ``cloning'' technique, involving transforming an instance $I$ of \acrshort{spa-st} into an instance $I'$ of \acrshort{smti}, and then using Kir\'{a}ly's $\frac{3}{2}$-approximation algorithm for \acrshort{smti} \cite{Kir12} in order to obtain a similar approximation in \acrshort{spa-st}, does not work in general, as we show in Section \ref{cloning_spast_smti}. This motivates the need for a bespoke algorithm for the \acrshort{spa-st} case. Note that we use Kir\'{a}ly's \acrshort{smti} algorithm rather than his \acrshort{hrt} algorithm as the former was more precisely described in \cite{Kir12}.

\subsection{Structure of this paper}
Section \ref{sec:spa-st_stab_defns} gives a formal definition of \acrshort{spa-st}. Section \ref{cloning_spast_smti} describes a technique to convert an \acrshort{spa-st} instance to an \acrshort{smti} instance, and gives an example where using this technique with Kir\'aly's $\frac{3}{2}$-approximation algorithm for \acrshort{smti} does not, in general, allow for the retention of the $\frac{3}{2}$ bound in the original \acrshort{spa-st} instance. Section \ref{sec_app_alg} describes the $\frac{3}{2}$-approximation algorithm with associated correctness proofs in Section \ref{32approx_corr_proofs}. The \acrshort{ip} model for \acrshort{max-spa-st} is given in Section \ref{sec_ip}.  The experimental evaluation is described in Section \ref{sec_spa-st_more_sec_exp}, and Section \ref{sec:conc} discusses future work.

\section{Formal definition of {\sc spa-st}}
\label{sec:spa-st_stab_defns}
%A formal definition of the \emph{\acrlong{spa-st}} (\acrshort{spa-st}) is now given. 
%
An instance $I$ of \acrshort{spa-st} comprises a set $S=\{s_1,s_2,...,s_{n_1}\}$ of \emph{students}, a set $P=\{p_1,p_2,...,p_{n_2}\}$ of \emph{projects}, and a set $L=\{l_1,l_2,...,l_{n_3}\}$ of \emph{lecturers}, where $|S| = n_1$, $|P| = n_2$ and $|L| = n_3$.  Each project is \emph{offered} by one lecturer, and each lecturer $l_k$ \emph{offers} a set of projects $P_k\subseteq P$, where $P_1,\dots,P_k$ partitions $P$.  Each project $p_j\in P$ has a \emph{capacity} $c_j\in \mathbb Z_0^+$, and similarly each lecturer $l_k\in L$ has a \emph{capacity} $d_k\in \mathbb Z_0^+$. We assume without loss of generality that $d_k$ is no greater than the sum of the capacities of the projects in $P_k$. Each student $s_i\in S$ has a set $A_i\subseteq P$ of \emph{acceptable} projects that they rank in order of preference. \emph{Ties} are allowed in preference lists, where a tie $t$ in a student $s_i$'s list indicates that $s_i$ is indifferent between all projects in $t$. Each lecturer $l_k\in L$ has a preference list over the students $s_i$ for which $A_i\cap P_k\neq \emptyset$.  Ties may also exist in lecturer preference lists. The \emph{rank} of project $p_j$ on student $s_i$'s list, denoted $\text{rank}(s_i,p_j)$, is defined as $1$ plus the number of projects that $s_i$ strictly prefers to $p_j$. An analogous definition exists for the rank of a student on a lecturer's list, denoted $\text{rank}(l_k,s_i)$. Let $m$ denote the total length of student preference lists.

An \emph{assignment} $M$ in $I$ is a subset of $S\times P$ such that, for each pair $(s_i,p_j)\in M$, $p_j \in A_i$, that is, $s_i$ finds $p_j$ acceptable. Let $M(s_i)$ denote the set of projects assigned to a student $s_i\in S$, let $M(p_j)$ denote the set of students assigned to a project $p_j\in P$, and let $M(l_k)$ denote the set of students assigned to projects in $P_k$ for a given lecturer $l_k\in L$.  A \emph{matching} $M$ is an assignment such that $|M(s_i)|\leq 1$ for all $s_i\in S$, $|M(p_j)|\leq c_j$ for all $p_j\in P$ and $|M(l_k)|\leq d_k$ for all $l_k\in L$. If $s_i\in S$ is assigned in a matching $M$, we let $M(s_i)$ denote $s_i$'s assigned project, otherwise $M(s_i)$ is empty. We describe a project $p_j$ as \emph{undersubscribed} or \emph{full} if $|M(p_j)|< c_j$ or $|M(p_j)| = c_j$, respectively. Similarly, we describe a lecturer $l_k$ as \emph{undersubscribed} or \emph{full} if $|M(l_k)|< d_k$ or $|M(l_k)| = d_k$, respectively.
%\subsection{Stability in \acrshort{spa-st}}

Given a matching $M$ in $I$, let $(s_i,p_j) \in (S \times P)\backslash M$ be a student-project pair, where $p_j$ is offered by lecturer $l_k$. Then $(s_i,p_j)$ is a \emph{blocking pair} of $M$ \cite{AIM07} if 1, 2 and 3 hold as follows:
\begin{enumerate}
	\item $s_i$ finds $p_j$ acceptable;
	\item $s_i$ either prefers $p_j$ to $M(s_i)$ or is unassigned in $M$;
	\item Either a, b or c holds as follows:
	\begin{enumerate}
		\item $p_j$ is \emph{undersubscribed} (i.e., $|M(p_j)|< c_j$) and $l_k$ is  \emph{undersubscribed} (i.e., $|M(l_k)|< d_k$);
		\item $p_j$ is undersubscribed, $l_k$ is full and either $s_i \in M(l_k)$ or $l_k$ prefers $s_i$ to the worst student in $M(l_k)$;
		\item $p_j$ is full and $l_k$ prefers $s_i$ to the worst student in $M(p_j)$.
	\end{enumerate}
\end{enumerate}

Let $(s_i,p_j)$ be a blocking pair of $M$. Then we say that $(s_i,p_j)$ is of \emph{type} $(3x)$ if $1$, $2$ and $3x$ are true in the above definition, where $x\in \{a,b,c\}$.  In order to more easily describe certain stages of the approximation algorithm, blocking pairs of type $(3b)$ are split into two subtypes as follows. $(3bi)$ defines a blocking pair of type $(3b)$ where $s_i$ is already assigned to another project of $l_k$'s. $(3bii)$ defines a blocking pair of type $(3b)$ where this is not the case.

A matching $M$ in an instance $I$ of \acrshort{spa-st} is \emph{stable} if it admits no blocking pair. Define \acrshort{max-spa-st} to be the problem of finding a maximum stable matching in \acrshort{spa-st} and let $M_{opt}$ denote a maximum stable matching for a given instance. Similarly, let \acrshort{min-spa-st} be the problem of finding a minimum stable matching in \acrshort{spa-st}.

%%%%%%%%%%%%%%%%%%%%%%%%%%%%%%%%%%%%%%%%%%%%%%%%%%%%%%%%%
%%%%%%%%%%%%%%%%%%%%%%%%%%%%%%%%%%%%%%%%%%%%%%%%%%%%%%%%%
\section{Cloning from {\sc spa-st} to {\sc smti}}
\label{cloning_spast_smti}
Manlove \cite[Theorem 3.11]{Man13} describes a polynomial transformation to convert a stable matching in an instance of \acrshort{hrt} to a stable matching in an instance of \acrshort{smti}, and vice versa, where the sizes of the matchings is conserved. An obvious question, which we address in this section, relates to whether a similar transformation could be used for \acrshort{spa-st}, which could allow existing approximation algorithms for \acrshort{hrt} and \acrshort{smti} to be used for the problem of finding a maximum stable matching in \acrshort{spa-st}.

A natural cloning method to convert instances of \acrshort{spa-st} to instances of \acrshort{hrt} is given as Algorithm \ref{Alg:APPROX_SPA-S_clonetohrt}. This algorithm involves converting students into residents and projects into hospitals. Hospitals inherit their capacity from projects. Residents inherit their preference lists naturally from students. Hospitals inherit their preference lists from the lecturer who offers their associated project; a resident entry $r_i$ is ranked only if $r_i$ also ranks this hospital. In order to translate lecturer capacities into the \acrshort{hrt} instance, a number of \emph{dummy residents} $R_d^k$ are created for each lecturer $l_k$. The number of dummy residents created for lecturer $l_k$, denoted $f_k$, is equal to the sum of capacities of their offered projects $P_k$ minus the capacity of $l_k$. We will ensure that all dummy residents are assigned in any stable matching. To this end, each dummy resident has a first position tie of all hospitals associated with projects of $l_k$, and each hospital $h_j$ in this set has a first position tie of all dummy residents associated with $l_k$. In this way, as all dummy residents must be assigned in any stable matching by Proposition \ref{th:stable_SPA_ST_clone_alldummyassigned}, lecturer capacities are automatically adhered to.

% THEOREM BEGIN
\begin{prop} 
	\label{th:stable_SPA_ST_clone_alldummyassigned}
	Let $I'$ be an instance of \acrshort{hrt} created from an instance $I$ of \acrshort{spa-st} using Algorithm \ref{Alg:APPROX_SPA-S_clonetohrt}. All dummy residents must be assigned in any stable matching in $I'$.
\end{prop}

% THEOREM END
\begin{proof}
In $I'$, for each lecturer $l_k$, the number of dummy residents created is equal to $f_k = \sum_{p_r \in P_k}{c_r} - d_k$. Assume for contradiction that one of the dummy residents $r_{d_1}^k$ is unassigned in some stable matching $M'$ of $I'$.

Let $H_k$ denote the set of hospitals associated with projects of $l_k$. Since $r_{d_1}^k$ is a dummy resident, it must have all hospitals in $H_k$ tied in first position. Also, each hospital in $H_k$ must rank all $f_k$ dummy residents (associated with $l_k$) in tied first position. Since $r_{d_1}^k$ is unassigned in $M'$, $r_{d_1}^k$ would prefer to be assigned to any hospital in $H_k$. Also, since there is at least one dummy resident unassigned, there must be at least one hospital $h_{d_2}^k$ in $H_k$ that has fewer first-choice assignees than its capacity. Hospital $h_{d_2}^k$ must exist since if it did not, then all dummy residents would be assigned. But then $(r_{d_1}^k,h_{d_2}^k)$ would be a blocking pair of $M'$, a contradiction.
\end{proof}

% ALGORITHM 
\begin{algorithm} [h]
  \caption[{\sf Clone-SPA-ST}, converts an {\sc spa-st} instance into an {\sc hrt} instance.]{{\sf Clone-SPA-ST}, converts an \acrshort{spa-st} instance into an \acrshort{hrt} instance.}
  \label{Alg:APPROX_SPA-S_clonetohrt}
	\begin{algorithmic}[1]
          \Require An instance $I$ of \acrshort{spa-st}
          \Ensure Return an instance $I'$ of \acrshort{hrt}
          
          \ForAll {student $s_i$ in $S$}
          \State Create a resident $r_i$ 
          \State $r_i$ inherits their preference list from $s_i$'s list, ranking hospitals rather than projects
          \EndFor
          \ForAll {project $p_j$ in $P$}
          \State Create a hospital $h_j$ 
          \State $h_j$'s capacity is given by $c'_j = c_j$
          \State Let $l_k$ be the lecturer offering project $p_j$
          \State $h_j$ inherits their preference list from $l_k$'s list, where a resident entry $r_i$ is retained only if $r_i$ also ranks $h_j$
          \EndFor
          
          \ForAll {lecturer $l_k$ in $L$}
          \If {$d_k < \sum_{p_j \in P_k}{c_j}$}
          \State let $f_k = \sum_{p_j \in P_k}{c_j} - d_k$
          \State Create $f_k$ new dummy residents $R_d^k = \{r^k_1,r^k_2,\dots ,r^k_{f_k}\}$
          \State Let $H_k$ denote the set of all hospitals in $I'$ associated with the projects of $P_k$ in $I$ 
          \State The preference list of each dummy resident is given by a first position tie of all hospitals in $H_k$
          \State A first position tie of all residents in $R_d^k$ is added to the start of the preference list of each hospital in $H_k$
          \EndIf
          \EndFor

          \State Let \acrshort{hrt} instance $I'$ be formed from all residents (including dummy residents) and hospitals
          \State \Return $I'$
	\end{algorithmic}
\end{algorithm}
% ALGORITHM END

% THEOREM BEGIN
\begin{theorem} 
	\label{th:stable_SPA_ST_clone_maxcardconserved}	
	Given an instance $I$ of \acrshort{spa-st} we can construct an instance $I'$ of \acrshort{hrt} in $O(n_1 + D n_2 + m)$ time with the property that a stable matching $M$ in $I$ can be converted to a stable matching $M'$ in $I'$ in $O(D n_2 + m)$ time, where $|M'| = |M| + \sum_{l_k \in L} \sum_{p_r \in P_k}(c_r) - d_k$. Here, $n_1$ denotes the number of students, $n_2$ the number of projects, $D$ the total capacities of lecturers and $m$ the total length of student preference lists. 
	\end{theorem}
% THEOREM END

\begin{proof}
Suppose $M$ is a stable matching in $I$. We construct an instance $I'$ of \acrshort{hrt} using Algorithm \ref{Alg:APPROX_SPA-S_clonetohrt}. The time complexity of $O(n_1 + D n_2 + m)$ for the reduction carried out by the algorithm is achieved by noting that $I'$ has a maximum of $n_1 + n_2 + D$ agents and that there are a maximum of $D n_2 + m$ acceptable resident-hospital pairs.

Initially let $M' = M$ (such that residents take the place of students and hospitals take the place of projects). By Proposition \ref{th:stable_SPA_ST_clone_alldummyassigned}, all dummy residents of $I'$ must be assigned in $M'$ and so we let the set of dummy residents $R_d$ form a resident-complete matching with the set of all hospitals and add these pairs to $M'$. This is possible because for each lecturer $l_k$, each dummy resident associated with $l_k$, denoted $r_d^k$, finds all hospitals in $H_k$ acceptable, and moreover the total number of remaining positions of the hospitals is equal to $\sum_{p_r \in P_k}(c_r) - |M(l_k)| = \delta_k$ and the number of dummy residents $f_k$ satisfies $f_k = \sum_{p_r \in P_k}(c_r) - d_k \leq \delta_k$, since $|M(l_k)| \leq d_k$. 

We claim that $M'$ is stable in $I'$. Suppose for contradiction that $(r_i,h_j)$ blocks $M'$ in $I'$. 

\begin{itemize}
	\item All dummy residents must be assigned in $M'$ to their first-choice hospital by above, hence $r_i$ corresponds to a student $s_i$ in $I$. Resident $r_i$ inherited their preference list from $s_i$ hence we know that $s_i$ finds $p_j$ acceptable. Therefore by the definition given in Section \ref{sec:spa-st_stab_defns}, Condition $1$ of a blocking pair of $M$ in $I$ is satisfied.
	\item Resident $r_i$ is either unassigned in $M'$ or prefers $h_j$ to $M'(r_i)$. Student $s_i$ is therefore in an equivalent position and Condition $2$ of a blocking pair of $M$ in $I$ is satisfied.
	\item Hospital $h_j$ is either undersubscribed or prefers $r_i$ to their worst assignee in $M'$.
	\begin{itemize}
		\item If $h_j$ is undersubscribed, then $p_j$ must also be undersubscribed. If $l_k$ were full in $M$ then $|M(l_k)| = d_k$ and so the number of remaining positions of hospitals in $H_k$ before dummy residents are added, $\delta_k$, is equal to the number of dummy residents $f_k$ in this scenario. But then all hospitals in $H_k$ (including $h_j$) would be full in $M'$, contradicting the fact that $h_j$ is undersubscribed. Therefore $l_k$ must be undersubscribed, but then this satisfies Condition $3(a)$ of a blocking pair. 
		\item If $h_j$ prefers $r_i$ to their worst assignee in $M'$, then $l_k$ must prefer $s_i$ to their worst assignee in $M(p_j)$. This satisfies Condition $3(c)$ of a blocking pair.
	\end{itemize}
\end{itemize}

Therefore by the definition in Section \ref{sec:spa-st_stab_defns}, $(s_i,p_j)$ is a blocking pair of $M$ in $I$, a contradiction.

Since dummy residents are added in the algorithm's execution it is clear that in general $|M| \neq |M'|$. However, since all dummy residents must be assigned by Proposition \ref{th:stable_SPA_ST_clone_alldummyassigned}, it is trivial to calculate the difference 
$$|M'| = |M| + |R_d| = |M| + \sum_{l_k \in L} \sum_{p_r \in P_k}(c_r) - d_k.$$
\end{proof}

The converse of Theorem \ref{th:stable_SPA_ST_clone_maxcardconserved} is not true in general, as shown in the example in Figure \ref{spastapprox:ex_conversefail}. Here, a stable matching $M'$ in an instance $I_0'$ of \acrshort{hrt} does not convert into a stable matching $M$ of the associated instance $I_0$ of \acrshort{spa-st}.

\begin{figure}[]
\centering

\begin{subfigure}[t]{0.4\textwidth}
Student preferences:\\
$s_1$: $\boldsymbol{p_1}$ $p_2$\\
$s_2$: $p_2$ $\boldsymbol{p_3}$\\

Project details:\\
$p_1$: lecturer $l_1$, $c_1=1$\\
$p_2$: lecturer $l_1$, $c_2=1$\\
$p_3$: lecturer $l_2$, $c_3=1$\\
  
Lecturer preferences:\\
\begin{minipage}[b]{0.6\textwidth}
$l_1$: $s_2$ $\boldsymbol{s_1}$\\
$l_2$: $\boldsymbol{s_2}$\\
\end{minipage}
\begin{minipage}[b]{0.3\textwidth}
$d_1=1$\\
$d_2=1$\\
\end{minipage}

    \subcaption{Example \acrshort{spa-st} instance $I_0$. Non-stable matching $M=\{(s_1,p_1),(s_2,p_3)\}$ derived from $M'$ is shown in bold.}
    \label{spastapprox:ex_conversefail2}
  \end{subfigure}
\hspace*{1cm}
  \begin{subfigure}[t]{0.4\textwidth}
  Resident preferences:\\
$r_1$: $\boldsymbol{h_1}$ $h_2$ \\
$r_2$: $h_2$ $\boldsymbol{h_3}$ \\
$r_3$: ($h_1$ $\boldsymbol{h_2}$) \\

  Hospital preferences:\\
  \begin{minipage}[b]{0.6\textwidth}
$h_1$: $r_3$ $\boldsymbol{r_1}$ \\
$h_2$: $\boldsymbol{r_3}$ $r_2$ $r_1$\\
$h_3$: $\boldsymbol{r_2}$\newline\newline\newline\newline
\end{minipage}
\begin{minipage}[b]{0.3\textwidth}
$c'_1=1$\\
$c'_2=1$\\
$c'_3=1$\newline\newline\newline\newline
\end{minipage}
   \subcaption{\acrshort{hrt} instance $I_0'$ created from the \acrshort{spa-st} instance in Figure \ref{spastapprox:ex_conversefail2} Stable matching $M' = \{(r_1,h_1),(r_2,h_3),(r_3,h_2)\}$ is shown in bold.} 
    \label{spastapprox:ex_conversefail1}
  \end{subfigure}

  \caption[Conversion of a stable matching in {\sc hrt} into a matching in {\sc spa-st}.]{Conversion of a stable matching $M'$ in \acrshort{hrt} into matching $M$ in \acrshort{spa-st}.}
  \label{spastapprox:ex_conversefail}
    \end{figure}

A natural question arises as to whether, using a cloning process, we may retain the $\frac{3}{2}$ bound in specific cases where the converted matching $M$ of our original instance of \acrshort{spa-st} does in fact turn out to be stable. This might occur if, for example, a specific stable matching returned by Kir\'aly's algorithm as applied to the cloned instance turns out to be stable in the original \acrshort{spa-st} instance. The cloning process in question would be as follows. For instance $I$ of \acrshort{spa-st}, we use the cloning process described in Algorithm \ref{Alg:APPROX_SPA-S_clonetohrt} to convert to instance $I'$ of \acrshort{hrt} then further convert to instance $I''$ of \acrshort{smti} using the process described by Manlove \cite[Theorem 3.11]{Man13}. Next, Kir\'aly's $\frac{3}{2}$-approximation algorithm is used on $I''$ generating stable matching $M''$. Finally, $M''$ is converted to matching $M$ of $I$. Then, assuming $M$ is stable, the question is whether it is always the case that $M$ is a $\frac{3}{2}$-approximation to a maximum stable matching of $I$.

The following example demonstrates Algorithm \ref{Alg:APPROX_SPA-S_clonetohrt} in use and shows that the process described above is \emph{not} sufficient to retain the $\frac{3}{2}$-approximation in an \acrshort{spa-st} instance $I$ even if the constructed matching $M$ is stable in $I$.

Algorithm \ref{Alg:APPROX_SPA-S_clonetohrt} is used to convert the \acrshort{spa-st} instance $I_1$ in Figure \ref{spastapprox:ex_before_clone} to an instance $I_1'$ of \acrshort{hrt} in Figure \ref{spastapprox:ex_after_clone1} which is then itself converted to the instance $I_1''$ of \acrshort{smti} in Figure \ref{spastapprox:ex_after_clone2} using the process described by Manlove \cite[Theorem 3.11]{Man13}. In this process men correspond to hospitals in $I_1'$ (projects in $I_1$) and women correspond to residents in $I_1'$ (students in $I_1$). Executing Kir\'aly's \cite{Kir12} $\frac{3}{2}$-approximation algorithm on the \acrshort{smti} instance $I_1''$ could (depending on order of proposals) yield the matching 
$M''=\{(m_2, w_5), (m_3, w_4), (m_4, w_2), (m_6, w_6), (m_7, w_7)\}.$
A trace of how this matching is created is given in Table \ref{tab:tracekiraly}. As $w_5$, $w_6$ and $w_7$ were created from dummy residents in Algorithm \ref{Alg:APPROX_SPA-S_clonetohrt}, $M''$ (stable in $I_1''$) converts into the stable matching $M = \{(s_2,p_4), (s_4,p_3)\}$ of size $2$ in $I_1$. But a maximum stable matching in $I_1$ is of size $4$, given by $M_{opt} = \{(s_1,p_3),(s_2,p_1),(s_3,p_3),(s_4,p_2)\}$. Therefore using the cloning method described above and Kir\'aly's algorithm does not result in a $\frac{3}{2}$-approximation to the maximum stable matching for instances of \acrshort{spa-st}, even when the resultant \acrshort{spa-st} matching is stable. This motivates the development of a $\frac{3}{2}$-approximation algorithm to the maximum stable matching specifically for instances of \acrshort{spa-st}.

\begin{figure}[]
\centering
\begin{subfigure}[t]{0.29\textwidth}
Student preferences:\\
$s_1$: $p_3$\\
$s_2$: $p_4$ $p_1$ $p_2$\\
$s_3$: $p_3$\\
$s_4$: ($p_2$ $p_3$) $p_4$ $p_1$\\

Project details:\\
$p_1$: lecturer $l_1$, $c_1=2$\\
$p_2$: lecturer $l_1$, $c_2=2$\\
$p_3$: lecturer $l_2$, $c_3=2$\\
$p_4$: lecturer $l_2$, $c_4=1$\\
  
Lecturer preferences:\\
\begin{minipage}[b]{0.6\textwidth}
$l_1$: $s_2$ $s_4$\\
$l_2$: $s_4$ ($s_1$ $s_2$ $s_3$)\newline\newline\newline
\end{minipage}
\begin{minipage}[b]{0.3\textwidth}
$d_1=2$\\
$d_2=2$\newline\newline\newline
\end{minipage}

    \subcaption{Example \acrshort{spa-st} instance $I_1$.}
    \label{spastapprox:ex_before_clone}
  \end{subfigure}
  \hspace*{\fill}
%  \hspace{0.5cm}
  \begin{subfigure}[t]{0.3\textwidth}
  Resident preferences:\\
$r_1$: $h_3$ \\
$r_2$: $h_4$ $h_1$ $h_2$ \\
$r_3$: $h_3$ \\
$r_4$: ($h_2$ $h_3$) $h_4$ $h_1$ \\
$r_5$: ($h_1$ $h_2$) \\
$r_6$: ($h_1$ $h_2$) \\
$r_7$: ($h_3$ $h_4$) \\

  Hospital preferences:\\
  \begin{minipage}[b]{0.6\textwidth}
$h_1$: ($r_5$ $r_6$) $r_2$ $r_4$ \\
$h_2$: ($r_5$ $r_6$) $r_2$ $r_4$ \\
$h_3$: $r_7$ $r_4$ ($r_1$ $r_3$) \\
$h_4$: $r_7$ $r_4$ $r_2$\newline\newline\newline\newline
\end{minipage}
\begin{minipage}[b]{0.3\textwidth}
$c'_1=2$\\
$c'_2=2$\\
$c'_3=2$\\
$c'_4=1$\newline\newline\newline\newline
\end{minipage}
   \subcaption{\acrshort{hrt} instance $I_1'$ converted from the \acrshort{spa-st} instance in Figure \ref{spastapprox:ex_before_clone}.}
    \label{spastapprox:ex_after_clone1}
  \end{subfigure}
  \hspace*{\fill}
   \begin{subfigure}[t]{0.37\textwidth}
  Women's preferences:\\
$w_1$: ($m_3$ $m_7$) \\
$w_2$: $m_4$ ($m_1$ $m_5$) ($m_2$ $m_6$) \\
$w_3$: ($m_3$ $m_7$) \\
$w_4$: ($m_2$ $m_3$ $m_6$ $m_7$) $m_4$ ($m_1$ $m_5$) \\
$w_5$: ($m_1$ $m_2$ $m_5$ $m_6$) \\
$w_6$: ($m_1$ $m_2$ $m_5$ $m_6$) \\
$w_7$: ($m_3$ $m_4$ $m_7$) \\

  Men's preferences:\\
$m_1$: ($w_5$ $w_6$) $w_2$ $w_4$\\
$m_2$: ($w_5$ $w_6$) $w_2$ $w_4$\\
$m_3$: $w_7$ $w_4$ ($w_1$ $w_3$)\\
$m_4$: $w_7$ $w_4$ $w_2$\\
$m_5$: ($w_5$ $w_6$) $w_2$ $w_4$\\
$m_6$: ($w_5$ $w_6$) $w_2$ $w_4$\\
$m_7$: $w_7$ $w_4$ ($w_1$ $w_3$)\\
   \subcaption{\acrshort{smti} instance $I_1''$ converted from the \acrshort{hrt} instance in Figure \ref{spastapprox:ex_after_clone1}.}
    \label{spastapprox:ex_after_clone2}
  \end{subfigure}
  \caption[Conversion of an {\sc spa-st}  instance to an {\sc smti}  instance.]{Conversion of an \acrshort{spa-st}  instance to an \acrshort{smti}  instance.}
  \label{spastapprox:ex_beforeandafter_clone}
    \end{figure}

\begin{table}[tbp]\hspace*{-1cm} \centering\begin{tabular}{ p{0.5cm}p{6.5cm} | p{0.7cm} p{0.6cm}p{0.6cm}p{0.6cm}p{0.6cm}p{0.6cm}p{0.6cm} }\hline\hline 
& Action & $m_1$ & $m_2$ & $m_3$ & $m_4$ & $m_5$ & $m_6$ & $m_7$  \\ 
\hline 
$1$ & $m_7$ applies to $w_7$, accepted &  &  &  &  &  &  & $w_7$\\
$2$ & $m_6$ applies to $w_5$, accepted &  &  &  &  &  & $w_5$ & $w_7$\\
$3$ & $m_5$ applies to $w_6$, accepted &  &  &  &  & $w_6$ & $w_5$ & $w_7$\\
$4$ & $m_4$ applies to $w_7$, rejected, $m_4$ removes $w_7$ &  &  &  &  & $w_6$ & $w_5$ & $w_7$\\
$5$ & $m_4$ applies to $w_4$, accepted &  &  &  & $w_4$ & $w_6$ & $w_5$ & $w_7$\\
$6$ & $m_3$ applies to $w_7$, rejected, $m_3$ removes $w_7$ &  &  &  & $w_4$ & $w_6$ & $w_5$ & $w_7$\\
$7$ & $m_3$ applies to $w_4$, accepted, $m_4$ removes $w_4$ &  &  & $w_4$ &  & $w_6$ & $w_5$ & $w_7$\\
$8$ & $m_4$ applies to $w_2$, accepted &  &  & $w_4$ & $w_2$ & $w_6$ & $w_5$ & $w_7$\\
$9$ & $m_2$ applies to $w_5$, rejected, $m_2$ removes $w_5$ &  &  & $w_4$ & $w_2$ & $w_6$ & $w_5$ & $w_7$\\
$10$ & $m_2$ applies to $w_6$, rejected, $m_2$ removes $w_6$ &  &  & $w_4$ & $w_2$ & $w_6$ & $w_5$ & $w_7$\\
$11$ & $m_2$ applies to $w_2$, rejected, $m_2$ removes $w_2$ &  &  & $w_4$ & $w_2$ & $w_6$ & $w_5$ & $w_7$\\
$12$ & $m_2$ applies to $w_4$, rejected, $m_2$ removes $w_4$ &  &  & $w_4$ & $w_2$ & $w_6$ & $w_5$ & $w_7$\\
$13$ & $m_2$ advantaged &  &  & $w_4$ & $w_2$ & $w_6$ & $w_5$ & $w_7$\\

$14$ & $m_2$ applies to $w_5$, accepted, $m_6$ removes $w_5$ &  & $w_5$ & $w_4$ & $w_2$ & $w_6$ &  & $w_7$\\
$15$ & $m_6$ applies to $w_6$, rejected, $m_6$ removes $w_6$ &  & $w_5$ & $w_4$ & $w_2$ & $w_6$ &  & $w_7$\\
$16$ & $m_6$ applies to $w_2$, rejected, $m_6$ removes $w_2$ &  & $w_5$ & $w_4$ & $w_2$ & $w_6$ &  & $w_7$\\
$17$ & $m_6$ applies to $w_4$, rejected, $m_6$ removes $w_4$ &  & $w_5$ & $w_4$ & $w_2$ & $w_6$ &  & $w_7$\\
$18$ & $m_6$ advantaged &  & $w_5$ & $w_4$ & $w_2$ & $w_6$ &  & $w_7$\\

$19$ & $m_6$ applies to $w_5$, rejected, $m_6$ removes $w_5$ &  & $w_5$ & $w_4$ & $w_2$ & $w_6$ &  & $w_7$\\
$20$ & $m_6$ applies to $w_6$, accepted, $m_5$ removes $w_6$ &  & $w_5$ & $w_4$ & $w_2$ &  & $w_6$ & $w_7$\\
$21$ & $m_5$ applies to $w_5$, rejected, $m_5$ removes $w_5$ &  & $w_5$ & $w_4$ & $w_2$ &  & $w_6$ & $w_7$\\
$22$ & $m_5$ applies to $w_2$, rejected, $m_5$ removes $w_2$ &  & $w_5$ & $w_4$ & $w_2$ &  & $w_6$ & $w_7$\\
$23$ & $m_5$ applies to $w_4$, rejected, $m_5$ removes $w_4$ &  & $w_5$ & $w_4$ & $w_2$ &  & $w_6$ & $w_7$\\
$24$ & $m_5$ advantaged &  & $w_5$ & $w_4$ & $w_2$ &  & $w_6$ & $w_7$\\

$25$ & $m_5$ applies to $w_5$, rejected, $m_5$ removes $w_5$ &  & $w_5$ & $w_4$ & $w_2$ &  & $w_6$ & $w_7$\\
$26$ & $m_5$ applies to $w_6$, rejected, $m_5$ removes $w_6$ &  & $w_5$ & $w_4$ & $w_2$ &  & $w_6$ & $w_7$\\
$27$ & $m_5$ applies to $w_2$, rejected, $m_5$ removes $w_2$ &  & $w_5$ & $w_4$ & $w_2$ &  & $w_6$ & $w_7$\\
$28$ & $m_5$ applies to $w_4$, rejected, $m_5$ removes $w_4$ &  & $w_5$ & $w_4$ & $w_2$ &  & $w_6$ & $w_7$\\
$29$ & $m_5$ inactive &  & $w_5$ & $w_4$ & $w_2$ & $-$ & $w_6$ & $w_7$\\

$30$ & $m_1$ applies to $w_5$, rejected, $m_1$ removes $w_5$ &  & $w_5$ & $w_4$ & $w_2$ & $-$ & $w_6$ & $w_7$\\
$31$ & $m_1$ applies to $w_6$, rejected, $m_1$ removes $w_6$ &  & $w_5$ & $w_4$ & $w_2$ & $-$ & $w_6$ & $w_7$\\
$32$ & $m_1$ applies to $w_2$, rejected, $m_1$ removes $w_2$ &  & $w_5$ & $w_4$ & $w_2$ & $-$ & $w_6$ & $w_7$\\
$33$ & $m_1$ applies to $w_4$, rejected, $m_1$ removes $w_4$ &  & $w_5$ & $w_4$ & $w_2$ & $-$ & $w_6$ & $w_7$\\
$34$ & $m_1$ advantaged &  & $w_5$ & $w_4$ & $w_2$ & $-$ & $w_6$ & $w_7$\\

$35$ & $m_1$ applies to $w_5$, rejected, $m_1$ removes $w_5$ &  & $w_5$ & $w_4$ & $w_2$ & $-$ & $w_6$ & $w_7$\\
$36$ & $m_1$ applies to $w_6$, rejected, $m_1$ removes $w_6$ &  & $w_5$ & $w_4$ & $w_2$ & $-$ & $w_6$ & $w_7$\\
$37$ & $m_1$ applies to $w_2$, rejected, $m_1$ removes $w_2$ &  & $w_5$ & $w_4$ & $w_2$ & $-$ & $w_6$ & $w_7$\\
$38$ & $m_1$ applies to $w_4$, rejected, $m_1$ removes $w_4$ &  & $w_5$ & $w_4$ & $w_2$ & $-$ & $w_6$ & $w_7$\\
$39$ & $m_1$ inactive & $-$ & $w_5$ & $w_4$ & $w_2$ & $-$ & $w_6$ & $w_7$\\

   \hline\hline \end{tabular} \caption[Trace of running Kir\'{a}ly's {\sc smti} $\frac{3}{2}$-approximation algorithm for instance $I_1''$.]{Trace of running Kir\'aly's \acrshort{smti} $\frac{3}{2}$-approximation algorithm for instance $I_1''$ in Figure \ref{spastapprox:ex_after_clone2}. In this table, the phrase ``$m_i$ removes $w_j$'' indicates that man $m_i$ removes woman $w_j$ from their preference list.} \label{tab:tracekiraly} \end{table}

In Appendix \ref{app_spa-st_further_discussions}, we introduce instance $I_2$, which is almost identical to $I_1$. However, applying the above process to $I_2$ does yield a stable matching $M$ in $I_2$ that is a $\frac{3}{2}$-approximation to a maximum stable matching $M_{opt}$. Comparing these two instances, we give an intuitive idea as to how the addition of dummy residents in the conversion of an \acrshort{spa-st} instance to an \acrshort{smti} instance can prevent the retention of the $\frac{3}{2}$ bound. See Appendix \ref{app_spa-st_further_discussions} for more details.

\newpage
%%%%%%%%%%%%%%%%%%%%%%%%%%%%%%%%%%%%%%%%%%%%%%%%%%%%%%%%%
%%%%%%%%%%%%%%%%%%%%%%%%%%%%%%%%%%%%%%%%%%%%%%%%%%%%%%%%%
\section{$\frac{3}{2}$-approximation algorithm}
\label{sec_app_alg}
\subsection{Introduction and preliminary definitions}
\label{spa_sec_defins}
We begin by defining key terminology before describing the approximation algorithm itself in Section \ref{sec:alg-desc}, which is a non-trivial extension of Kir\'aly's \acrshort{hrt} algorithm \cite{Kir12}.  

%\subsection{Definitions}
A student $s_i \in S$ is either in \emph{phase $1$, $2$} or \emph{$3$}. In \emph{phase $1$} there are still projects on $s_i$'s list that they have not applied to. In \emph{phase $2$}, $s_i$ has iterated once through their list and are doing so again whilst a priority is given to $s_i$ on each lecturer's preference list, compared to other students who tie with $s_i$. In \emph{phase $3$}, $s_i$ is considered unassigned and carries out no more applications. A project $p_j$ is \emph{fully available} if $p_j$ and $l_k$ are both undersubscribed, where lecturer $l_k$ offers $p_j$. A student $s_i$ \emph{meta-prefers} project $p_{j_1}$ to $p_{j_2}$ if either (i) rank$(s_i,p_{j_1}) <$ rank$(s_i,p_{j_2})$, or (ii) rank$(s_i,p_{j_1}) =$ rank$(s_i,p_{j_2})$ and $p_{j_1}$ is fully available, whereas $p_{j_2}$ is not. In phase $1$ or $2$, $s_i$ may be either \emph{available}, \emph{provisionally assigned} or \emph{held}. Student $s_i$ is \emph{available} if they are not assigned to a project. Student $s_i$ is \emph{provisionally assigned} to project $p_j$ if $s_i$ has been assigned in phase $1$ to $p_j$ and there is a project still on $s_i$'s list that they meta-prefer to $p_j$. Otherwise, $s_i$ is \emph{held}. 

If a student $s_i$ is a provisionally assigned to project $p_j$, then $(s_i,p_j)$ is said to be \emph{precarious}. If $s_i$ is held in their assignment to $p_j$, then $(s_i,p_j)$ is said to be \emph{non-precarious}. A project $p_j$ is \emph{precarious} if it is assigned a student $s_i$ such that $(s_i, p_j)$ is precarious, otherwise $p_j$ is \emph{non-precarious}. A lecturer is \emph{precarious} if they offer a project $p_j$ that is precarious, otherwise $l_k$ is \emph{non-precarious}. Lecturer $l_k$ \emph{meta-prefers} $s_{i_1}$ to $s_{i_2}$ if either (i) rank$(l_k,s_{i_1}) <$ rank$(l_k,s_{i_2})$, or (ii) rank$(l_k,s_{i_1}) =$ rank$(l_k,s_{i_2})$ and $s_{i_1}$ is in phase $2$, whereas $s_{i_2}$ is not. The \emph{favourite} projects $F_i$ of a student $s_i$ are defined as the set of projects on $s_i$'s preference list for which there is no other project on $s_i$'s list meta-preferred to any project in $F_i$. A \emph{worst assignee} of lecturer $l_k$ is defined to be a student in $M(l_k)$ of worst rank, with priority given to phase $1$ students over phase $2$ students. Similarly, a \emph{worst assignee of lecturer $l_k$ in $M(p_j)$} is defined to be a student in $M(p_j)$ of worst rank, prioritising phase $1$ over phase $2$ students, where $l_k$ offers $p_j$.

We remark that some of the above terms such as \emph{favourite} and \emph{precarious} have been defined for the \acrshort{spa-st} setting by extending the definitions of the corresponding terms as given by Kir\'aly in the \acrshort{hrt} context \cite{Kir12}. %(described in Section \ref{spa_sec_defins}).

%%%%%%%%%%%%%%%%%%%%%%%%%%%%%%%%%%%%%%%%%%%%%%%%%%%%%%%%%
%%%%%%%%%%%%%%%%%%%%%%%%%%%%%%%%%%%%%%%%%%%%%%%%%%%%%%%%%
\subsection{Description of the algorithm}
\label{sec:alg-desc}
Algorithm {\sf Max-SPA-ST-Approx} (Algorithm \ref{Alg:APPROX_SPA-S_stable}) begins with an empty matching $M$ which will be built up over the course of the algorithm's execution. All students are initially set to be available and in phase $1$. The algorithm proceeds as follows. While there are still available students in phase $1$ or $2$, choose some such student $s_i$. Student $s_i$ applies to a favourite project $p_j$ at the head of their list, that is, there is no project on $s_i$'s list that $s_i$ meta-prefers to $p_j$. Let $l_k$ be the lecturer who offers $p_j$. We consider the following cases.
\begin{itemize}
  \item If $p_j$ and $l_k$ are both undersubscribed then $(s_i,p_j)$ is added to $M$. %This is the easiest case to describe. 
  Clearly if $(s_i,p_j)$ were not added to $M$, it would potentially be a blocking pair of type $(3a)$.
  
  \item If $p_j$ is undersubscribed, $l_k$ is full and $l_k$ is precarious where precarious pair $(s_{i'},p_{j'}) \in M$ for some project $p_j'$ offered by $l_k$, then we remove $(s_{i'},p_{j'})$ from $M$ and add pair $(s_i,p_j)$. This notion of precariousness allows us to find a stable matching of sufficient size even when there are ties in student preference lists (there may also be ties in lecturer preference lists). Allowing a pair $(s_{i'},p_{j'}) \in M$ to be precarious means that we are noting that $s_{i'}$ has other fully available project options in their preference list at equal rank to $p_{j'}$. Hence, if another student applies to $p_{j'}$ when $p_{j'}$ is full, or to a project offered by $l_k$ where $l_k$ is full, we allow this assignment to happen removing $(s_{i'},p_{j'})$ from $M$, since there is a chance that the size of the resultant matching could be increased. Note that since $s_{i'}$ does not remove $p_{j'}$ from their preference list, $s_{i'}$ will get a chance to reapply to $p_{j'}$ if applications to other fully available projects at the same rank are unsuccessful.

  \item If on the other hand $p_j$ is undersubscribed, $l_k$ is full and $l_k$ meta-prefers $s_i$ to a worst assignee $s_{i'}$, where $(s_{i'},p_{j'}) \in M$ for some project $p_{j'}$ offered by $l_k$, then we remove $(s_{i'},p_{j'})$ from $M$ and add pair $(s_i,p_j)$. It makes intuitive sense that if $l_k$ is full and gets an offer to an undersubscribed project $p_j$ from a student $s_i$ that they meta-prefer to a worst assigned student $s_{i'}$, then $l_k$ would want to remove $s_{i'}$ from $p_{j'}$ and take on $s_i$ for $p_j$. Student $s_{i'}$ will subsequently remove $p_{j'}$ from their preference list as $l_k$ will not want to assign to them on re-application. This is done via Algorithm \ref{Alg:APPROX_SPA-S_stable_extra}.

  \item If $p_j$ is full and precarious then pair $(s_i, p_j)$ is added to $M$ while precarious pair $(s_{i'}, p_j)$ is removed. As before, this allows $s_{i'}$ to potentially assign to other fully available projects at the same rank as $p_j$ on their list. Since $s_{i'}$ does not remove $p_j$ from their preference list, $s_{i'}$ will get another chance to assign to $p_j$ if these other applications to fully available projects at the same rank are not successful.

  \item If $p_j$ is full and $l_k$ meta-prefers $s_i$ to a worst assignee $s_{i'}$ in $M(p_j)$, then pair $(s_i, p_j)$ is added to $M$ while $(s_{i'}, p_j)$ is removed. As this lecturer's project is full (and non-precarious) the only time they will want to add a student $s_i$ to this project (meaning the removal of another student) is if $s_i$ is meta-preferred to a worst student $s_{i'}$ assigned to that project. Similar to before, $s_{i'}$ will not subsequently be able to assign to this project and so removes it from their preference list via Algorithm \ref{Alg:APPROX_SPA-S_stable_extra}.
\end{itemize}

\begin{algorithm} []
  \caption[{\sf Max-SPA-ST-Approx}$(I)$, $\frac{3}{2}$-approximation algorithm for {\sc spa-st}.]{{\sf Max-SPA-ST-Approx}$(I)$, $\frac{3}{2}$-approximation algorithm for \acrshort{spa-st}.} %finding a stable matching $M$ in an instance $I$ of
  \label{Alg:APPROX_SPA-S_stable}
  \begin{algorithmic}[1]
          \Require An instance $I$ of \acrshort{spa-st}
          \Ensure Return a stable matching $M$ where $|M| \geq \frac{2}{3}|M_{opt}|$
          \State $M \gets \emptyset$
          \State All students are initially set to be available and in phase $1$

          \While {there exists an available student $s_i \in S$ who is in phase $1$ or $2$} \label{Alg:APPROX_SPA-S_stable_mainloop} 
          \State Let $l_k$ be the lecturer who offers $p_j$
          \State $s_i$ applies to a favourite project $p_j \in A(s_i)$
           
           \If {$p_j$ is fully available}
           \State $M \gets M \cup \{(s_i,p_j)\}$ \label{Alg:APPROX_SPA-S_stable:pairadded1}
           
           \ElsIf {$p_j$ is undersubscribed, $l_k$ is full \textbf{and} ($l_k$ is precarious \textbf{or} $l_k$ meta-prefers $s_i$ to a worst assignee)} \label{Alg:APPROX_SPA-S_stable_elseif1} \Comment according to the \emph{worst assignee} definition in Section \ref{spa_sec_defins}
           \If {$l_k$ is precarious} 
           \State Let $p_{j'}$ be a project in $P_k$ such that there exists $(s_{i'},p_{j'}) \in M$ that is precarious%
           \label{Alg:APPROX_SPA-S_stable:pjidentified_asprecarious}

           \Else \Comment{$l_k$ is non-precarious} 

           \State Let $s_{i'}$ be a worst assignee of $l_k$ such that $l_k$ meta-prefers $s_i$ to $s_{i'}$ and let $p_{j'}=M(s_{i'})$

           \State {\sf Remove-Pref}($s_{i'},p_{j'}$) \label{Alg:APPROX_SPA-S_stable:rempref1}

           \EndIf
           \State $M \gets M\backslash\{(s_{i'},p_{j'})\}$ \label{Alg:APPROX_SPA-S_stable:1}
           \State $M \gets M \cup \{(s_i,p_j)\}$ \label{Alg:APPROX_SPA-S_stable:add1}

           \ElsIf {$p_j$ is full \textbf{and} ($p_j$ is precarious \textbf{or} $l_k$ meta-prefers $s_i$ to a worst assignee in $M(p_j)$)} \label{Alg:APPROX_SPA-S_stable_elseif2}

           \If {$p_j$ is precarious}
           \State Identify a student $s_{i'}\in M(p_j)$ such that $(s_{i'},p_j)$ is precarious

           \Else \Comment{$p_j$ is non-precarious} 
           \State Let $s_{i'}$ be a worst assignee of $l_k$ in $M(p_j)$ such that $l_k$ meta-prefers $s_i$ to $s_{i'}$

           \State {\sf Remove-Pref}($s_{i'},p_j$) \label{Alg:APPROX_SPA-S_stable:rempref2}
           \EndIf
           \State $M \gets M\backslash\{(s_{i'},p_j)\}$ \label{Alg:APPROX_SPA-S_stable:2}
           \State $M \gets M \cup \{(s_i,p_j)\}$ \label{Alg:APPROX_SPA-S_stable:add2}

           \Else 
           \State {\sf Remove-Pref}($s_i,p_j$) \label{Alg:APPROX_SPA-S_stable:rempref3}
           \EndIf
    
          \EndWhile
          \State {\sf Promote-Students}($M$) \label{Alg:APPROX_SPA-S_stable_promotealgline}
      \State \Return $M$;
          
  \end{algorithmic}
\end{algorithm}
% ALGORITHM END

% ALGORITHM 
\begin{algorithm} [h]
  \caption[{\sf Remove-Pref}($s_i,p_j$), subroutine for Algorithm \ref{Alg:APPROX_SPA-S_stable}.]{{\sf Remove-Pref}($s_i,p_j$), subroutine for Algorithm \ref{Alg:APPROX_SPA-S_stable}. Removes a project $p_j$ from a student $s_i$'s preference list.}
  \label{Alg:APPROX_SPA-S_stable_extra}
  \begin{algorithmic}[1]
          \Require An instance $I$ of \acrshort{spa-st} and a student $s_i$ and project $p_j$
          \Ensure Return an instance $I$ where $p_j$ is removed from $s_i$'s preference list
          
          \State Remove $p_j$ from $s_i$'s preference list
          \If {$s_i$'s preference list is empty}
          \State Reinstate $s_i$'s preference list
          
          \If {$s_i$ is in phase $1$}
          \State Move $s_i$ to phase $2$
          \ElsIf {$s_i$ is in phase $2$}
          \State Move $s_i$ to phase $3$
          \EndIf

          \EndIf
          \State \Return $I$
  \end{algorithmic}
\end{algorithm}
% ALGORITHM END

% ALGORITHM 
\begin{algorithm} [h]
  \caption[{\sf Promote-Students}($M$), subroutine for Algorithm \ref{Alg:APPROX_SPA-S_stable}.]{{\sf Promote-Students}($M$), subroutine for Algorithm \ref{Alg:APPROX_SPA-S_stable}. Removes all blocking pairs of type $(3bi)$.}%, including any that are created during this process}
  \label{Alg:APPROX_SPA-S_stable_3bibpBrief}
  \begin{algorithmic}[1]
          \Require \acrshort{spa-st} Instance $I$ and matching $M$ that does not contain blocking pairs of type $(3a)$, $(3bii)$ or $(3c)$.
          \Ensure Return a stable matching $M$.

  \While {there are still blocking pairs of type $(3bi)$}
  \State Let $(s_i,p_{j'})$ be a blocking pair of type $(3bi)$ 
  \State $M \gets M\backslash\{(s_i,M(s_i))\}$ 
    \State $M \gets M \cup \{(s_i,p_{j'})\}$ 
  \EndWhile
  \State \Return $M$
  \end{algorithmic}
\end{algorithm}
% ALGORITHM END

When removing a project from a student $s_i$'s preference list (the {\sf Remove-Pref} operation), if $s_i$ has removed all projects from their preference list and is in phase $1$ then their preference list is reinstated and they are set to be in phase $2$. If on the other hand they were already in phase $2$, then they are set to be in phase $3$ and are hence inactive. The proof that Algorithm \ref{Alg:APPROX_SPA-S_stable} produces a stable matching (see Section \ref{32approx_corr_proofs}) relies only on the fact that a student iterates once through their preference list. Allowing students to iterate through their preference lists a second time when in phase $2$ allows us to find a stable matching of sufficient size when there are ties in lecturer preference lists (there may also be ties in student preference lists). This is due to the meta-prefers definition where a lecturer favours one student $s_i$ over another $s_{i'}$ if they are the same rank and $s_i$ is in phase $2$ whereas $s_{i'}$ is not. Similar to above, this then allows $s_i$ to steal a position from $s_{i'}$ with the chance that $s_{i'}$ may find another assignment and increase the size of the resultant matching.

After the main while loop has terminated, the final part of the algorithm begins where all blocking pairs of type $(3bi)$ are removed using the {\sf Promote-Students} operation (Algorithm \ref{Alg:APPROX_SPA-S_stable_3bibpBrief}).

\newpage
\clearpage
\newpage
\subsection{Example execution of the algorithm}
A detailed trace of Algorithm {\sf Max-SPA-ST-Approx} (Algorithm \ref{Alg:APPROX_SPA-S_stable}) over the course of its execution, as applied to the example instance $I_3$ of \acrshort{spa-st} shown in Figure \ref{spa-st-fig-example-alg-run}, is given in Table \ref{spa-st-fig-example-alg-run-trace}. In this trace, each application by a student to a project is recorded along with their effects on the instance and matching (such as adding or removing a pair from a matching, removing a preference list element and a student changing phase). The line numbers of Algorithm \ref{Alg:APPROX_SPA-S_stable} where these effects take place are also recorded. The state of the matching after each application may be seen in the final twelve columns, indicating the assignments of each of the twelve students. An asterisk next to a project $p_j$ for column $s_i$ indicates that $(s_i, p_j)$ is a precarious pair. From this trace, we can see that all parts of the algorithm are executed at some point (although not all line numbers are listed in the trace, this observation may be easily verified by examining the structure of Algorithm \ref{Alg:APPROX_SPA-S_stable}). 

For this particular instance, the algorithm outputs stable matching $M$ where 

\begin{equation*}
\begin{split}
 M = \{&(s_1, p_3), (s_2, p_2), (s_3, p_8), (s_5, p_3), (s_6, p_1), \\
 & (s_7, p_6), (s_9, p_5), (s_{10}, p_{12}), (s_{11}, p_9), (s_{12}, p_{11}) \}.
 \end{split}
\end{equation*}

\begin{figure}
\begin{minipage}[t]{0.3\textwidth}
Student preferences:\\
$s_1$: ($p_2$ $p_3$) $p_1$\\
$s_2$: $p_2$ $p_1$\\
$s_3$: ($p_7$ $p_8$)\\
$s_4$: $p_7$\\
$s_5$: $p_1$ ($p_2$ $p_3$)\\
$s_6$: $p_1$\\
$s_7$: ($p_4$ $p_6$)\\
$s_8$: $p_5$\\
$s_9$: $p_5$\\
$s_{10}$: $p_9$ $p_{12}$\\
$s_{11}$: $p_9$ $p_{10}$\\
$s_{12}$: $p_{11}$\\
\end{minipage}
\begin{minipage}[t]{0.3\textwidth}
Project details:\\
$p_1$: lecturer $l_1$, $c_1=1$\\
$p_2$: lecturer $l_1$, $c_2=2$\\
$p_3$: lecturer $l_2$, $c_3=2$\\
$p_4$: lecturer $l_3$, $c_4=1$\\
$p_5$: lecturer $l_3$, $c_5=1$\\
$p_6$: lecturer $l_4$, $c_6=1$\\
$p_7$: lecturer $l_4$, $c_7=1$\\
$p_8$: lecturer $l_5$, $c_8=1$\\
$p_9$: lecturer $l_6$, $c_9=1$\\
$p_{10}$: lecturer $l_6$, $c_{10}=1$\\
$p_{11}$: lecturer $l_6$, $c_{11}=1$\\
$p_{12}$: lecturer $l_7$, $c_{12}=1$\\
\end{minipage}
\begin{minipage}[t]{0.35\textwidth}
Lecturer preferences:\\
\begin{minipage}[t]{0.6\textwidth}
$l_1$: ($s_1$ $s_2$) $s_6$ $s_5$\\
$l_2$: $s_1$ $s_5$\\
$l_3$: ($s_8$ $s_9$) $s_7$\\
$l_4$: $s_7$ ($s_3$ $s_4$)\\
$l_5$: $s_3$\\
$l_6$: $s_{12}$ ($s_{10}$ $s_{11}$)\\
$l_7$: $s_{10}$\\
\end{minipage}
\begin{minipage}[t]{0.25\textwidth}
$d_1=2$\\
$d_2=2$\\
$d_3=1$\\
$d_4=1$\\
$d_5=1$\\
$d_6=2$\\
$d_7=2$\\
\end{minipage}
\end{minipage}
\caption[{\sc spa-st} instance $I_3$.]{\acrshort{spa-st} instance $I_3$.}
\label{spa-st-fig-example-alg-run}
\end{figure}

% \newgeometry{bottom=1.5cm, top=1.5cm}

\begin{landscape}
 \thispagestyle{empty}
\begin{table}[tbp] \centering\begin{tabular}{ p{0.5cm}p{9cm} |p{1.5cm}| p{0.4cm} p{0.4cm}p{0.4cm}p{0.4cm}p{0.4cm}p{0.4cm}p{0.4cm}p{0.4cm}p{0.4cm}p{0.4cm}p{0.4cm}p{0.4cm}}\hline\hline 
& Action & Line(s) & $s_1$ & $s_2$ & $s_3$ & $s_4$ & $s_5$ & $s_6$ & $s_7$ & $s_8$ & $s_9$ & $s_{10}$ & $s_{11}$ & $s_{12}$  \\ 
\hline 

$1$ & $s_1$ applies to $p_2$ &\\
& ($s_1$, $p_2$) added to $M$ & 7 & $p_2$*\\

$2$ & $s_2$ applies to $p_2$ &\\
& ($s_2$, $p_2$) added to $M$ & 7 & $p_2$* & $p_2$\\

$3$ & $s_3$ applies to $p_7$ &\\
& ($s_3$, $p_7$) added to $M$ & 7 & $p_2$* & $p_2$ & $p_7$*\\

$4$ & $s_4$ applies to $p_7$ &\\
& ($s_3$, $p_7$) removed from $M$, ($s_4$, $p_7$) added to $M$ & 24, 25 & $p_2$* & $p_2$ & & $p_7$\\

$5$ & $s_5$ applies to $p_1$ &\\
& ($s_1$, $p_2$) removed from $M$, ($s_5$, $p_1$) added to $M$ & 15, 16 & & $p_2$ & & $p_7$ & $p_1$\\

$6$ & $s_6$ applies to $p_1$ &\\
& $s_5$ removes $p_1$, ($s_5$, $p_1$) removed from $M$, ($s_6$, $p_1$) added to $M$ & 22, 24, 25 & & $p_2$ & & $p_7$ & & $p_1$\\

$7$ & $s_7$ applies to $p_4$ &\\
& ($s_7$, $p_4$) added to $M$ & 7 & & $p_2$ & & $p_7$ & & $p_1$ & $p_4$\\

$8$ & $s_8$ applies to $p_5$ &\\
& $s_7$ removes $p_4$, ($s_7$, $p_4$) removed from $M$, ($s_8$, $p_5$) added to $M$ & 13, 15, 16 & & $p_2$ & & $p_7$ & & $p_1$ & & $p_5$\\

$9$ & $s_9$ applies to $p_5$ &\\
& $s_9$ rejected, $s_9$ removes $p_5$, $s_9$ moves to phase $2$ & 27 & & $p_2$ & & $p_7$ & & $p_1$ & & $p_5$\\

$10$ & $s_9$ applies to $p_5$ &\\
& $s_8$ removes $p_5$, $s_8$ moves to phase $2$, ($s_8$, $p_5$) removed from $M$, ($s_9$, $p_5$) added to $M$ & 22, 24, 25 & & $p_2$ & & $p_7$ & & $p_1$ & & & $p_5$\\

$11$ & $s_{10}$ applies to $p_9$ &\\
& ($s_{10}$, $p_9$) added to $M$ & 7 & & $p_2$ & & $p_7$ & & $p_1$ & & & $p_5$ & $p_9$\\

$12$ & $s_{11}$ applies to $p_9$ &\\
& $s_{11}$ rejected, $s_{11}$ removes $p_9$ & 27 & & $p_2$ & & $p_7$ & & $p_1$ & & & $p_5$ & $p_9$\\

   \hline\hline 
   \multicolumn{2}{l}{Table continued on next page.}\end{tabular} \end{table} 
\end{landscape}

\begin{landscape}
 \thispagestyle{empty}

   \begin{table}[tbp] \centering\begin{tabular}{ p{0.5cm}p{9cm} |p{1.5cm}| p{0.4cm} p{0.4cm}p{0.4cm}p{0.4cm}p{0.4cm}p{0.4cm}p{0.4cm}p{0.4cm}p{0.4cm}p{0.4cm}p{0.4cm}p{0.4cm}}\hline\hline 
& Action & Line(s) & $s_1$ & $s_2$ & $s_3$ & $s_4$ & $s_5$ & $s_6$ & $s_7$ & $s_8$ & $s_9$ & $s_{10}$ & $s_{11}$ & $s_{12}$  \\ 
\hline 

$13$ & $s_{11}$ applies to $p_{10}$ &\\
& ($s_{11}$, $p_{10}$) added to $M$ & 7 & & $p_2$ & & $p_7$ & & $p_1$ & & & $p_5$ & $p_9$ & $p_{10}$\\

$14$ & $s_{12}$ applies to $p_{11}$ &\\
& $s_{10}$ removes $p_9$, ($s_{10}$, $p_9$) removed from $M$, ($s_{12}$, $p_{11}$) added to $M$ & 13, 15, 16 & & $p_2$ & & $p_7$ & & $p_1$ & & & $p_5$ & & $p_{10}$ & $p_{11}$\\

$15$ & $s_3$ applies to $p_8$ &\\
& ($s_3$, $p_8$) added to $M$ & 7 & & $p_2$ & $p_8$ & $p_7$ & & $p_1$ & & & $p_5$ & & $p_{10}$ & $p_{11}$\\

$16$ & $s_1$ applies to $p_3$ &\\
& ($s_1$, $p_3$) added to $M$ & 7 & $p_3$ & $p_2$ & $p_8$ & $p_7$ & & $p_1$ & & & $p_5$ & & $p_{10}$ & $p_{11}$\\

$17$ & $s_5$ applies to $p_3$ &\\
& ($s_5$, $p_3$) added to $M$ & 7 & $p_3$ & $p_2$ & $p_8$ & $p_7$ & $p_3$ & $p_1$ & & & $p_5$ & & $p_{10}$ & $p_{11}$\\

$18$ & $s_7$ applies to $p_6$ &\\
& $s_4$ removes $p_7$, $s_4$ moves to phase $2$, ($s_4$, $p_7$) removed from $M$, ($s_7$, $p_6$) added to $M$ & 13, 14, 15 & $p_3$ & $p_2$ & $p_8$ & & $p_3$ & $p_1$ & $p_6$ & & $p_5$ & & $p_{10}$ & $p_{11}$\\

$19$ & $s_8$ applies to $p_5$ &\\
& $s_8$ rejected, $s_8$ removes $p_5$, $s_8$ moves to phase $3$ & 27 & $p_3$ & $p_2$ & $p_8$ & & $p_3$ & $p_1$ & $p_6$ & - & $p_5$ & & $p_{10}$ & $p_{11}$\\

$20$ & $s_{10}$ applies to $p_{12}$ &\\
& ($s_{10}$, $p_{12}$) added to $M$ & 7 & $p_3$ & $p_2$ & $p_8$ & & $p_3$ & $p_1$ & $p_6$ & - & $p_5$ & $p_{12}$ & $p_{10}$ & $p_{11}$\\

$21$ & $s_4$ applies to $p_7$ &\\
& $s_4$ rejected, $s_4$ removes $p_7$, $s_4$ moves to phase $3$ & 27 & $p_3$ & $p_2$ & $p_8$ & - & $p_3$ & $p_1$ & $p_6$ & - & $p_5$ & $p_{12}$ & $p_{10}$ & $p_{11}$\\

$22$ & Identifying blocking pairs of type $(3bi)$ &\\
& ($s_{11}$, $p_{10}$) removed from $M$, ($s_{11}$, $p_9$) added to $M$ & 30 & $p_3$ & $p_2$ & $p_8$ & - & $p_3$ & $p_1$ & $p_6$ & - & $p_5$ & $p_{12}$ & $p_9$ & $p_{11}$\\

   \hline\hline \end{tabular} \caption[Detailed trace of running Algorithm {\sf Max-SPA-ST-Approx} for instance $I_3$.]{Detailed trace of running Algorithm {\sf Max-SPA-ST-Approx} for instance $I_3$ in Figure \ref{spa-st-fig-example-alg-run}. Projects that form part of a precarious pair with their associated student are marked with an ``*''. In this table, the phrase ``$s_i$ removes $p_j$'' indicates that student $s_i$ removes project $p_j$ from their preference list.} \label{spa-st-fig-example-alg-run-trace} \end{table} 
   \thispagestyle{empty}
\end{landscape}
% \restoregeometry

%%%%%%%%%%%%%%%%%%%%%%%%%%%%%%%%%%%%%%%%%%%%%%%%%%%%%%%%%
%%%%%%%%%%%%%%%%%%%%%%%%%%%%%%%%%%%%%%%%%%%%%%%%%%%%%%%%%

\section{$\frac{3}{2}$-approximation algorithm correctness proofs}
\label{32approx_corr_proofs}

\subsection{Introduction}

In this section we present proofs of correctness for Algorithm {\sf Max-SPA-ST-Approx} (Algorithm \ref{Alg:APPROX_SPA-S_stable}).  These involve showing firstly that the algorithm always produces a stable matching, secondly that the algorithm runs in linear time, and finally that the performance guarantee is $\frac{3}{2}$. The proofs required for this algorithm are naturally longer and more complex than those given by Kir\'aly \cite{Kir12} for \acrshort{smti} (formal proofs were not provided for his \acrshort{hrt} algorithm), as \acrshort{spa-st} generalises \acrshort{smti} to the case that lecturers can offer multiple projects, and projects and lecturers may have capacities greater than $1$. These extensions add extra components to the definition of a blocking pair (given in Section \ref{sec:spa-st_stab_defns}) which in turn adds complexity to the algorithm and its proof of correctness. 

The rest of this section is structured as follows. Section \ref{spa-st-proofs-sec-prelim} presents proofs of several preliminary results that are used throughout Section \ref{32approx_corr_proofs}. Proof that the algorithm produces a stable matching is given in Section \ref{spa-st-proofs-sec-stability}. Section \ref{spa-st-proofs-sec-time} shows that the algorithm runs in linear time with respect to the total length of student preference lists. Finally in Section \ref{spa-st-proofs-sec-performance}, we present proof of the $\frac{3}{2}$ performance guarantee.

\subsection{Proofs of preliminary results}
\label{spa-st-proofs-sec-prelim}

This section comprises several proofs of preliminary results that are used in the following sections. In general, they concern the issue of when in the algorithm's execution a project may be fully available, or when a project or lecturer may be precarious.

First we show in Proposition \ref{th:stable_SPA_S_propfullyavailable} (the proof of which uses Propositions \ref{th:stable_SPA_S_propfullyavailablePRE} and \ref{th:stable_SPA_S_bp_lk_full}), that if a project is not fully available at some point in the algorithm's execution, then it cannot subsequently become fully available.

% THEOREM BEGIN
\begin{prop} 
	\label{th:stable_SPA_S_propfullyavailablePRE} Let $T_0$ denote the point in Algorithm \ref{Alg:APPROX_SPA-S_stable}'s execution at the end of the main while loop. If a project $p_j$ is not fully available at some point before $T_0$, then it cannot subsequently become fully available before $T_0$.
	\end{prop}
% THEOREM END

\begin{proof}
Assume for contradiction that project $p_j$ is not fully available at some point before $T_0$, but then subsequently becomes fully available before $T_0$. At a point where $p_j$ is not fully available, either $p_j$ is full or $l_k$ is full (or both), where $l_k$ offers $p_j$. If $l_k$ is full, it is clear that $l_k$ must remain so, since students can only be removed from a project of $l_k$'s if they are immediately replaced by another student assigning to a project of $l_k$. Therefore assume that $p_j$ is full. Then they must somehow become undersubscribed in order to be classified as fully available. The only way this can happen before $T_0$ is if lecturer $l_k$ removes a student assigned to $p_j$ in order to replace them with a student becoming assigned to another project of $l_k$'s. But then this deletion can only occur if $l_k$ is full and as above $l_k$ remains full, so $p_j$ cannot become fully available before $T_0$, a contradiction. 
\end{proof}

% THEOREM BEGIN
\begin{prop} 
	\label{th:stable_SPA_S_bp_lk_full}
	 Suppose a blocking pair $(s_i,p_{j'})$ of type $(3bi)$ exists at the end of the main while loop of Algorithm \ref{Alg:APPROX_SPA-S_stable}, where $l_k$ offers $p_{j'}$, and denote this time by $T_0$. Then at time $T_0$, $l_k$ is full. 
	\end{prop}
% THEOREM END
\begin{proof} 
	Assume for contradiction that $l_k$ is undersubscribed at $T_0$. We know that once a lecturer is full they must remain full (since we can only remove a pair associated with a lecturer if we are immediately replacing it with an associated pair). Therefore $l_k$ must have always been undersubscribed. At $T_0$, $p_{j'}$ must be undersubscribed for $(s_i,p_{j'})$ to be a blocking pair of type $(3bi)$. Therefore at $T_0$, $p_{j'}$ is fully available and must always have previously been fully available by Proposition \ref{th:stable_SPA_S_propfullyavailablePRE}. But $s_i$ must have applied to $p_{j'}$ at least once before $T_0$ and as $p_{j'}$ was fully available this must have been accepted. Then since $(s_i,p_{j'})$ is not in the matching at $T_0$ it must have been removed before $T_0$. But in order for this to happen either $p_{j'}$ or $l_k$ would have to be full, contradicting the fact that $p_{j'}$ was fully available before $T_0$. Hence $l_k$ must be full at $T_0$.
\end{proof}

% {\color{blue}
% Proposition \ref{th:stable_SPA_S_propfullyavailable} may be seen as an extension of Proposition \ref{th:stable_SPA_S_propfullyavailablePRE}.
% }

% THEOREM BEGIN
\begin{prop} 
	\label{th:stable_SPA_S_propfullyavailable}
	During the execution of Algorithm \ref{Alg:APPROX_SPA-S_stable}, if a project $p_j$ is not fully available at some point, then it cannot subsequently become fully available.
	\end{prop}
% THEOREM END

\begin{proof} Let $T_0$ denote the point in the algorithm's execution at the end of the main while loop. We know from Proposition \ref{th:stable_SPA_S_propfullyavailablePRE} that if $p_j$ is not fully available before $T_0$ then it cannot subsequently become fully available before $T_0$.

Let lecturer $l_k$ offer project $p_j$ and assume $p_j$ is not fully available at $T_0$. If $l_k$ contains no blocking pairs of type $(3bi)$ then there can be no changes to allocations of $p_j$ after $T_0$. Therefore assume $l_k$ contains at least one blocking pair of type $(3bi)$ at $T_0$. Then by Proposition \ref{th:stable_SPA_S_bp_lk_full}, $l_k$ is full at $T_0$. But Algorithm \ref{Alg:APPROX_SPA-S_stable_3bibpBrief} does not change the student allocations for any lecturer and hence $l_k$ remains full and $p_j$ cannot subsequently become fully available.

It remains to show that if $p_j$ is fully available at $T_0$, that it cannot subsequently cease to be fully available and then return to be fully available before the end of the algorithms execution. Since $l_k$ is undersubscribed at $T_0$, $l_k$ cannot contain any blocking pairs by Proposition \ref{th:stable_SPA_S_bp_lk_full}. Since Algorithm \ref{Alg:APPROX_SPA-S_stable_3bibpBrief} can only affect allocations of projects offered by a full lecturer it is not possible for $p_j$ to change to being not fully available after $T_0$.
\end{proof}

In Proposition \ref{prop-spa-st-prom_subsequentprecarious} we show that after the main while loop, no student can be promoted to a fully available project and cannot create a precarious pair. This proposition is used as a stepping stone for other propositions in this section.

\begin{prop} 
\label{prop-spa-st-prom_subsequentprecarious}
Algorithm \ref{Alg:APPROX_SPA-S_stable_3bibpBrief} cannot promote a student to a fully available project, nor can it create a precarious pair.
\end{prop}

\begin{proof} Suppose in Algorithm \ref{Alg:APPROX_SPA-S_stable_3bibpBrief}, $s_i$ is being promoted from project $p_j$ to project $p_{j'}$ both offered by lecturer $l_k$. We know that in the main while loop $s_i$ must have iterated over their preference list at least to the position of $p_j$ (and perhaps further if $s_i$ has been previously promoted). Therefore, $s_i$ has either been removed from and / or rejected by all projects at the same rank as $p_{j'}$ in their preference list at least once. This can only occur if each of those projects was not fully available at the time and by Proposition \ref{th:stable_SPA_S_propfullyavailable} none of these projects could subsequently be fully available. Therefore when $s_i$ is promoted to $p_{j'}$, it can never form a precarious pair.
\end{proof}

Propositions \ref{prop-spa-st-subsequentprecarious} and \ref{th:stable_SPA_S_bp_lk_not_prec} describe conditions in the algorithm's execution under which it is not possible for a particular project or lecturer to be precarious.

\begin{prop} 
\label{prop-spa-st-subsequentprecarious}
Let $p_j$ be a project and let $l_k$ be the lecturer who offers $p_j$. If $l_k$ is full and non-precarious at some point, then they cannot subsequently become precarious. Similarly, if a project $p_j$ is full and non-precarious at some point, then it cannot subsequently become precarious. Further if $l_k$ is full and $p_j$ is non-precarious then $p_j$ cannot subsequently become precarious.
	
\end{prop}

\begin{proof}
We know that a precarious pair cannot be created in Algorithm \ref{Alg:APPROX_SPA-S_stable_3bibpBrief} by Proposition \ref{prop-spa-st-prom_subsequentprecarious}, therefore we focus on the main while loop of Algorithm \ref{Alg:APPROX_SPA-S_stable}. Let lecturer $l_k$ be full and non-precarious at some point during the main while loop Algorithm \ref{Alg:APPROX_SPA-S_stable}'s execution and assume that they later becomes precarious. Since $l_k$ is currently non-precarious, the only way they can become so is by a student $s_i$ forming a precarious assignment to a project of $l_k$'s. But recall that a student will first apply to fully available projects at the head of their preference list. Since $l_k$ is full, no project of $l_k$'s can be considered fully available. In order for $s_i$ to apply to a project of $l_k$'s they must first apply to all fully available projects at the head of their list, gain the assignment and then be removed from $M$. But if a pair $(s_i,p_j)$ is removed from $M$, $p_j$ cannot be fully available. Ultimately, $s_i$ will have exhausted all previously fully available projects at the head of their list before eventually applying to a project of $l_k$'s. But then at that point $s_i$ cannot create a precarious pair giving a contradiction. 

Now, let $p_j$ be full and non-precarious at some point during the main while loop of Algorithm \ref{Alg:APPROX_SPA-S_stable}'s execution and assume that it later becomes precarious. If $p_j$ remains full until the time at which it becomes precarious then using similar reasoning to above, any student assigning to $p_j$ cannot be precarious giving a contradiction. If $p_j$ becomes undersubscribed at any point then $l_k$ must be full for the remainder of the algorithm by Proposition \ref{th:stable_SPA_S_propfullyavailable}. It is possible at this point that $l_k$ is precarious (with precarious pairs that include projects other than $p_j$) but since $l_k$ is full, $p_j$ is not fully available. Therefore using similar reasoning to before, any student assigning to $p_j$ cannot create a precarious pair giving a contradiction. 

It also follows then that if $l_k$ is full and $p_j$ is non-precarious then $p_j$ cannot subsequently become precarious.
\end{proof}

% THEOREM BEGIN
\begin{prop} 
	\label{th:stable_SPA_S_bp_lk_not_prec}
	 Suppose a blocking pair $(s_i,p_{j'})$ of type $(3bi)$ exists at the end of the main while loop of Algorithm \ref{Alg:APPROX_SPA-S_stable}, where $l_k$ offers $p_{j'}$, and denote this time by $T_0$. Then at time $T_0$, $l_k$ is non-precarious. 
	\end{prop}
% THEOREM END
\begin{proof} Let $M_0$ be the matching being built at $T_0$ and let $(s_i,p_j)\in M_0$ with $l_k$ offering $p_j$. Suppose for contradiction that $l_k$ is precarious at $T_0$. 

As $(s_i,p_{j'})$ is a blocking pair of type $(3bi)$, $p_{j'}$ must be undersubscribed at $T_0$. Also, since $(s_i,p_{j'})$ is a blocking pair we know that $s_i$ prefers $p_{j'}$ to $p_j$. Therefore $s_i$ must have removed $p_{j'}$ from their preference list. Denote this time as $T_1$. The removal at $T_1$ occurred either because $(s_i,p_{j'})$ was removed as a non-precarious pair, or because $s_i$ was rejected on application to $p_{j'}$.

\begin{itemize}
	\item If $(s_i,p_{j'})$ was removed as a non-precarious pair at $T_1$ then either $l_k$ was full, non-precarious and $s_i$ was a worst student assigned to $l_k$, or $p_{j'}$ was full, non-precarious and $s_i$ was a worst student assigned to $p_{j'}$.
	\item If on the other hand, $s_i$ was rejected on application to $p_{j'}$ at $T_1$, we know that either $l_k$ was full, non-precarious and $l_k$ did not meta-prefer $s_i$ to a worst student in $M(l_k)$, or $p_{j'}$ was full, non-precarious and $l_k$ did not meta-prefer $s_i$ to a worst student in $M(p_j)$.
\end{itemize}

Whichever of these possibilities occurred we know that at $T_1$, either $l_k$ was full and non-precarious or $p_j$ was full and non-precarious.

\begin{itemize}
	\item Firstly suppose $l_k$ was full and non-precarious at $T_1$. In this case by Proposition \ref{prop-spa-st-subsequentprecarious}, $l_k$ cannot subsequently become precarious, a contradiction to the fact that $l_k$ is precarious at $T_0$. 
	
	\item Therefore, $p_{j'}$ must have been full and non-precarious at $T_1$. By Proposition \ref{prop-spa-st-subsequentprecarious}, $p_{j'}$ cannot subsequently become precarious. We also know that $p_{j'}$ must go from being full to being undersubscribed since $p_{j'}$ is undersubscribed at $T_0$. Denote this point in the algorithm's execution as $T_2$. At $T_2$, $p_{j'}$ must be non-precarious by above and so a non-precarious pair involved with $p_{j'}$ is removed and replaced with a pair involved with some other project of $l_k$'s. This could only happen if $l_k$ was full and non-precarious and so as before cannot again become precarious, a contradiction. 
\end{itemize}

Therefore, $l_k$ must be non-precarious at $T_0$.
\end{proof}

In Proposition \ref{th:stable_SPA_S_bp_promalg_notchange} we show that after the main while loop, it is not possible for any project to change its fully available status or for any project or lecturer to change their precarious status.

% THEOREM BEGIN
\begin{prop} 
	\label{th:stable_SPA_S_bp_promalg_notchange}
	 Algorithm \ref{Alg:APPROX_SPA-S_stable_3bibpBrief} cannot change the fully available status of any project or the precarious status of any project or lecturer. 
	 % Additionally, Algorithm \ref{Alg:APPROX_SPA-S_stable_3bibpBrief} cannot change the precarious status of any pair. 
	\end{prop}
% THEOREM END
\begin{proof}
	By Proposition \ref{prop-spa-st-prom_subsequentprecarious}, in Algorithm \ref{Alg:APPROX_SPA-S_stable_3bibpBrief} it is not possible to assign a student to a fully available project. Therefore we cannot change a fully available project to be not fully available. Also, any promotions that take place will be to remove a blocking pair of type $(3bi)$, and so by definition the lecturer involved, say $l_k$, will be full and $|M(l_k)|$ will remain the same. Therefore, all of $l_k$'s projects are not fully available at the start of Algorithm \ref{Alg:APPROX_SPA-S_stable_3bibpBrief} and must remain so.
	
By Proposition \ref{prop-spa-st-prom_subsequentprecarious}, in Algorithm \ref{Alg:APPROX_SPA-S_stable_3bibpBrief} it is not possible to create a precarious pair, meaning we cannot change a non-precarious project or lecturer to being a precarious project or lecturer. Finally, by Proposition \ref{th:stable_SPA_S_bp_lk_not_prec} no changes can be made to any assignments involving a precarious project or lecturer, hence we cannot change a precarious project or lecturer to be non-precarious. 
% Also, since Algorithm \ref{Alg:APPROX_SPA-S_stable_3bibpBrief} cannot change the fully available status of any project, it is not possible for a pair to change it's precarious status. 
% \todo{does a pair not "change" it's precarious status when it is added to the matching? Precariousness doesn't really exist outside of allocated pairs?}
\end{proof}

Recall, a \emph{worse} student than student $s_i$, according to lecturer $l_k$, is any student with a lower rank than $s_i$ on $l_k$'s preference list, or if $s_i$ is in phase $2$, any student of the same rank that is in phase $1$.

Proposition \ref{th:stable_SPA_S_lecfullworseprec} examines two specific circumstances. Firstly, it shows that if a lecturer $l_k$ is full, then they cannot subsequently accept a student worse than or equal to a worst student in $M(l_k)$, unless $l_k$ is precarious at the point of application. Secondly, it shows that if a project $p_j$ is full and non-precarious at some point before the end of the main while loop, then $p_j$ cannot subsequently accept a student worse than or equal to a worst student in $M(p_j)$ (according to the lecturer who offers $p_j$), before the end of the main while loop. In all propositions of this section up to this point, $T_0$ has been used to denote the point in the algorithm's execution at the end of the main while loop. In Proposition \ref{th:stable_SPA_S_lecfullworseprec} we change this notation to $T_{end}$. This is done in order for $T_0$, $T_{0.5}$, $T_1$ and $T_{end}$ to denote times that are ordered chronologically.

% THEOREM BEGIN
\begin{prop} 
	\label{th:stable_SPA_S_lecfullworseprec}
	Let $T_{end}$ denote the point in Algorithm \ref{Alg:APPROX_SPA-S_stable}'s execution at the end of the main while loop. Then the following statements are true. 
	\begin{enumerate}

		\item If a lecturer $l_k$ is full, then a  student $s_i$ worse than or equal to $l_k$'s worst assignee(s) cannot subsequently become assigned to a project $p_j$ offered by $l_k$ unless this occurs during the main while loop and $l_k$ is precarious when $s_i$ applies to $p_j$.

		\item If a project $p_j$ offered by $l_k$ is full and non-precarious before $T_{end}$, then a student $s_i$ worse than or equal to $l_k$'s worst ranked assignee(s) in $M(p_j)$ cannot subsequently become assigned to $p_j$ before $T_{end}$.

	\end{enumerate}
	  \end{prop}
 
% THEOREM END

\begin{proof}

We deal with each case separately.

\begin{enumerate}

\item Let $T_0$ be a point of the algorithm's execution, mentioned in the statement of the propostion, where $l_k$ is full. Let $T_1$ be the first point after $T_0$ where a student $s_i$ worse than or equal to $l_k$'s worst assignee(s) applies to $p_j$. As this is the first such point, it is not possible for $M(l_k)$ at $T_1$ to have students of lower rank than existed in $M(l_k)$ at $T_0$.

It is clear that in Algorithm \ref{Alg:APPROX_SPA-S_stable_3bibpBrief}, students assigned to a particular lecturer cannot change, hence we concentrate only on the main while loop of Algorithm \ref{Alg:APPROX_SPA-S_stable}. Since $l_k$ is full at $T_0$, $l_k$ must be full at $T_1$ since a pair may only be added and taken away from the same project, or from different projects offered by the same lecturer. Assume $l_k$ is non-precarious at $T_1$. Either $p_j$ is full or undersubscribed. Suppose $p_j$ is full. Since $l_k$ is non-precarious, $p_j$ is also non-precarious by definition. But then, since $l_k$ does not meta-prefer $s_i$ to a worst assignee in $M(l_k)$, the conditions on Line \ref{Alg:APPROX_SPA-S_stable_elseif2} cannot be satisfied and so $s_i$ cannot become assigned to $p_j$. Now, assume that $p_j$ is undersubscribed. Since $l_k$ is full, non-precarious, and does not meta-prefer $s_i$ to a worst assignee in $M(l_k)$, the conditions on Line \ref{Alg:APPROX_SPA-S_stable_elseif1} are also not satisfied and so $s_i$ cannot become assigned to $p_j$. Since it is not possible for a student $s_i$ worse than or equal to $l_k$'s worst assignee(s) at $T_0$ to become assigned to $p_j$ when this application happens for the first time, it is easy to extend this to any subsequent occurence during the main while loop.

If on the other hand, $l_k$ is precarious at $T_1$, since $l_k$ is also full at $T_1$, it is easy to see that $s_i$ may become assigned to $p_j$ (if either $p_j$ is undersubscribed or $p_j$ is full and precarious).

\item Similar to above, let $T_0$ be a point of the algorithm's execution, mentioned in the statement of the propostion, where $p_j$ is full and non-precarious. Let $T_1$ be the first point after $T_0$ and before $T_{end}$ where a student $s_i$ worse than or equal to $l_k$'s worst assignee(s) in $M(p_j)$ applies to $p_j$. As this is the first such point, it is not possible for $M(p_j)$ at $T_1$ to have students of lower rank than existed in $M(p_j)$ at $T_0$. This result can clearly be extended to apply not only to $T_1$, but also to any point after $T_0$ and before $T_1$.

Since $p_j$ is full and non-precarious at $T_0$, it must also be non-precarious at all future points by Proposition \ref{prop-spa-st-subsequentprecarious}. At $T_1$ either $p_j$ is full or undersubscribed. 

\begin{itemize}
	\item If $p_j$ is full at $T_1$, then since it is also non-precarious and there cannot be lower ranked students in $M(p_j)$ at $T_1$ than there were in $M(p_j)$ at $T_0$, the conditions on Line \ref{Alg:APPROX_SPA-S_stable_elseif2} cannot be satisfied and so $s_i$ cannot become assigned to $p_j$. 

	\item If $p_j$ is undersubscribed at $T_1$, then it first became undersubscribed at some point after $T_0$ and before $T_1$. Let $T_{0.5}$ be the point just prior to $p_j$ becoming undersubscribed. At $T_{0.5}$, $l_k$ must be full since $p_j$ can only become undersubscribed on Line \ref{Alg:APPROX_SPA-S_stable:1}. Let $(s_{i'}, p_j)$ be the pair that is removed from the matching just after $T_{0.5}$. We know $p_j$ is non-precarious after $T_0$ and so pair $(s_{i'}, p_j)$ is removed as a non-precarious pair, by definition. But, since $(s_{i'}, p_j)$ is non-precarious, it can only have been removed just after $T_{0.5}$ (on Line \ref{Alg:APPROX_SPA-S_stable:1}) if $l_k$ is non-precarious and $s_{i'}$ was not only a worst assignee in $M(p_j)$ but also in $M(l_k)$. This means that a worst assignee in $M(l_k)$ at $T_{0.5}$ is of equal rank to a worst assignee in $M(p_j)$ at $T_{0.5}$ and cannot be worse than a worst assignee in $M(p_j)$ at $T_0$ (by the result given at the start of this case). From this we can deduce that $s_i$ is worse than or equal to $l_k$'s worst assignee(s) in $M(l_k)$ at $T_{0.5}$. Thus, by the proof of Case 1 above, since $l_k$ was full and non-precarious at $T_{0.5}$ with $s_i$ worse than or equal to $l_k$'s worst assignee(s) in $M(l_k)$ at $T_{0.5}$, $s_i$ cannot subsequently become assigned to $p_j$ at $T_1$. Since it is not possible for a student $s_i$ worse than or equal to $l_k$'s worst assignee(s) in $M(p_j)$ at $T_0$ to become assigned to $p_j$ when this application happens for the first time before $T_{end}$, it is easy to extend this to any subsequent occurence at any time before $T_{end}$.
\end{itemize}

\end{enumerate}

Therefore each case is proved.
\end{proof}

\subsection{Stability}
\label{spa-st-proofs-sec-stability}

In this section we present several results building up to Theorem \ref{approx_not_destable} which shows that Algorithm {\sf Max-SPA-ST-Approx} always produces a stable matching.

First, in Lemma \ref{approx_end_while_only_3bi}, we show that after the main while loop of Algorithm \ref{Alg:APPROX_SPA-S_stable}, only blocking pairs of type $(3bi)$ can exist relative to $M$.

% THEOREM BEGIN
\begin{lemma} 
	\label{approx_end_while_only_3bi}
		Let $M_1$ denote the matching constructed immediately after the main while loop in Algorithm \ref{Alg:APPROX_SPA-S_stable} has completed and let $T_1$ denote this point in the algorithm's execution. At $T_1$, no blocking pair of type $(3a)$, $(3bii)$ or $(3c)$ can exist relative to $M_1$.
\end{lemma}
% THEOREM END

\begin{proof} 
 Assume for contradiction that $(s_{b_1},p_{b_2})$ is a blocking pair of $M_1$ of type $(3a)$, $(3bii)$ or $(3c)$. Let $l_{b_3}$ be the lecturer who offers $p_{b_2}$. 

	It must be the case that in $M_1$, $s_{b_1}$ is either assigned to a project of lower rank than $p_{b_2}$ or is assigned to no project. Therefore, $s_{b_1}$ must have removed $p_{b_2}$ from their preference list during the main while loop of Algorithm \ref{Alg:APPROX_SPA-S_stable}. Let $M_0$ denote the matching constructed immediately before $p_{b_2}$ was first removed from $s_{b_1}$'s list and let $T_0$ denote this point in the algorithm's execution. We know that $p_{b_2}$ cannot be fully available at $T_0$ (otherwise $(s_{b_1},p_{b_2})$ would have been added to $M_0$) and cannot subsequently become fully available by Proposition \ref{th:stable_SPA_S_propfullyavailable}. There are three places where $p_{b_2}$ could be removed from $s_{b_1}$'s list, namely Lines \ref{Alg:APPROX_SPA-S_stable:rempref1}, \ref{Alg:APPROX_SPA-S_stable:rempref2} and \ref{Alg:APPROX_SPA-S_stable:rempref3}. We look at each type of blocking pair in turn. 
	
\begin{itemize}
	\item $(3a)$ - Assume we have a blocking pair of type $(3a)$ in $M_1$. Then, $p_{b_2}$ and $l_{b_3}$ are both undersubscribed (and hence $p_{b_2}$ is fully available) in $M_1$. But this contradicts the above statement that $p_{b_2}$ cannot be fully available after $T_0$.

	\item $(3bii)$ \& $(3c)$ - Assume we have a blocking pair of type $(3bii)$ or $(3c)$ in $M_1$. At $T_0$, $p_{b_2}$ is not fully available and so either $p_{b_2}$ is undersubscribed with $l_{b_3}$ being full, or $p_{b_2}$ is full. 
	
	\begin{itemize}
	
	\item If $p_{b_2}$ was undersubscribed and $l_{b_3}$ was full at $T_0$ then $l_{b_3}$ cannot have been precarious (since $s_{b_1}$ is about to remove $p_{b_2}$ from their list) and by Proposition \ref{prop-spa-st-subsequentprecarious}, cannot subsequently become precarious. By Proposition \ref{th:stable_SPA_S_lecfullworseprec}, $l_{b_3}$ cannot subsequently accept a student ranked lower than a worst student in $M_0(l_{b_3})$. Also either $s_{b_1}$ is a worst assignee in $M_0(l_{b_3})$ (Line \ref{Alg:APPROX_SPA-S_stable:rempref1}), or $l_{b_3}$ ranks $s_{b_1}$ at least as badly as a worst student in $M_0(l_{b_3})$ (Line \ref{Alg:APPROX_SPA-S_stable:rempref3}). Lecturer $l_{b_3}$ must remain full for the rest of the algorithm since if a student is removed from a project offered by $l_{b_3}$ then they are immediately replaced. Therefore $l_{b_3}$ must be full in $M_1$, and since $l_{b_3}$ cannot have accepted a worse ranked student than $s_{b_1}$ after $T_0$, $(s_{b_1},p_{b_2})$ cannot be a blocking pair of $M_1$ of type $(3bii)$ or $(3c)$. 
	
\item Instead assume at $T_0$ that $p_{b_2}$ is full in $M_0$. As $s_{b_1}$ is about to remove $p_{b_2}$, we know $p_{b_2}$ cannot be precarious, and by Proposition \ref{prop-spa-st-subsequentprecarious}, cannot subsequently become precarious. Either $s_{b_1}$ is a worst assignee in $M_0(p_{b_2})$ (Line \ref{Alg:APPROX_SPA-S_stable:rempref2}), or $l_{b_3}$ ranks $s_{b_1}$ at least as badly as a worst student in $M_0(p_{b_2})$ (Line \ref{Alg:APPROX_SPA-S_stable:rempref3}).

As $p_{b_2}$ is full and non-precarious, by Proposition \ref{th:stable_SPA_S_lecfullworseprec}, $p_j$ cannot accept a worse student than already exists in $M(p_{b_2})$ at $T_0$. If $p_{b_2}$ is full at $T_1$, then clearly $(s_{b_1},p_{b_2})$ cannot block $M_1$. So assume $p_{b_2}$ becomes undersubscribed at some point between $T_0$ and $T_1$ for the first time, say at $T_{0.5}$. Since $p_{b_2}$ is non-precarious, all pairs in $M$ associated with $p_{b_2}$ are also non-precarious. Therefore $p_{b_2}$ can only become undersubscribed at $T_{0.5}$ if $l_{b_3}$ is full and there is a student $s_i$ who assigns to another project $p_j$ that $l_{b_3}$ offers, where $s_i$ is meta-preferred to a worst student in $M_0(l_{b_3})$.  This worst student must also be a worst student in $M_0(p_{b_2})$ since we are removing from $M_0$ a pair associated with $p_{b_2}$. But then $s_{b_1}$ must be ranked at least as badly as a worst student in $M_0(l_{b_3})$. Using similar reasoning to the previous case, $l_{b_3}$ must be full in $M_1$, non-precarious, and since $l_{b_3}$ cannot have accepted a worse ranked student than $s_{b_1}$ after $T_{0.5}$, $(s_{b_1},p_{b_2})$ cannot be a blocking pair of $M_1$ of type $(3bii)$ or $(3c)$. 

	\end{itemize}

\end{itemize} 
Hence it is not possible for $(s_{b_1},p_{b_2})$ to be a blocking pair of $M$ of type $(3a)$, $(3bii)$ or $(3c)$ after the main while loop.
\end{proof}

Let $p_j$ be a project, where $l_k$ offers $p_j$. Propositions \ref{th:stable_SPA_S_3biworstinmlk} and \ref{th:stable_SPA_S_3bicreatedworstinmlk} show that from the end of the main while loop to the end of the execution of Algorithm \ref{Alg:APPROX_SPA-S_stable}, if $(s_i, p_j)$ is a blocking pair of type $(3bi)$, then $s_i$ must be one of the worst students in $M(l_k)$.

% THEOREM BEGIN
\begin{prop} 
	\label{th:stable_SPA_S_3biworstinmlk}
	Let $M_1$ be the matching constructed immediately at the end of the main while loop of Algorithm \ref{Alg:APPROX_SPA-S_stable}'s execution, and let $T_1$ denote this point in the algorithm's execution. At $T_1$, for each blocking pair $(s_i,p_{j'})$ of type $(3bi)$, $s_i$ must be one of the worst assignees of $M_1(l_k)$, where $l_k$ offers $p_{j'}$. 
	\end{prop}
% THEOREM END

\begin{proof}
Since $(s_i,p_{j'})$ is a blocking pair of type $(3bi)$, we know that $s_i$ is assigned to another project, say $p_j$, of $l_k$'s in $M_1$, where $s_i$ prefers $p_{j'}$ to $p_j$. 

During the main while loop's execution, student $s_i$ must have removed $p_{j'}$ from their preference list in order to eventually assign to $p_j$. We note that although $s_i$ may be in either phase, $s_i$ must have removed $p_{j'}$ from their preference list in the same phase that $s_i$ assigned to $p_j$ (also the same phase $s_i$ is in at $T_1$). For the remainder of this proof, we discuss removal of $p_{j'}$ from $s_i$'s preference list in the context of this phase. Student $s_i$'s removal of $p_{j'}$ could only happen if $p_{j'}$ was non-precarious at this point. By Proposition \ref{prop-spa-st-subsequentprecarious}, since $p_{j'}$ or $l_k$ are full, $p_{j'}$ cannot subsequently become precarious. Let the matching constructed immediately before this removal be denoted by $M_0$ and let $T_0$ denote this point in the algorithm's execution. Project $p_{j'}$ was either full or undersubscribed at $T_0$. We show in the former case that, $s_i$ is one of the worst students in $M_1(l_k)$, and that the latter case leads to a contradiction. 

\begin{itemize}
	\item Suppose $p_{j'}$ was full at $T_0$. Then as $p_{j'}$ is non-precarious, $p_{j'}$ cannot subsequently be assigned a student worse than the worst assignee in $M_0(p_{j'})$ up until $T_1$, by Proposition \ref{th:stable_SPA_S_lecfullworseprec}. Since $s_i$ removed $p_{j'}$ from their list while $p_{j'}$ was full we know that either $s_i$ is a worst assignee in $M_0(p_{j'})$ (Line \ref{Alg:APPROX_SPA-S_stable:rempref2}), or $l_k$ ranks $s_i$ at least as badly as a worst student in $M_0(p_{j'})$ (Line \ref{Alg:APPROX_SPA-S_stable:rempref3}). Between $T_0$ and $T_1$, $s_i$ assigns to the project $p_j$, at a worse rank than $p_{j'}$ in $s_i$'s list, where $p_j$ is also offered by $l_k$. 

Now, we know that $p_{j'}$ becomes undersubscribed by $T_1$ and so it must be the case that there is a point $T_{0.5}$ between $T_0$ and $T_1$, such that another student $s_{i'}$ assigns to a project (not $p_{j'}$) of $l_k$'s which removes pair $(s_{i''},p_{j'})$ from the matching constructed just before that removal, denoted by $M_{0.5}$. Let $T_{0.5}$ be the first point at which $p_{j'}$ becomes undersubscribed after $T_0$. Lecturer $l_k$ must be full at this point since the addition of $s_{i'}$ removes a student (namely $s_{i''}$) from a different project (namely $p_{j'}$). Also, lecturer $l_k$ cannot have been precarious, otherwise $p_{j'}$ would have been identified as a precarious project at Line \ref{Alg:APPROX_SPA-S_stable:pjidentified_asprecarious}, but we know $p_{j'}$ cannot have been precarious after $T_0$. So $s_{i''}$ must have been a worst assignee in $M_{0.5}(l_k)$ and therefore $M_{0.5}(p_{j'})$. From the beginning of this bullet point, we know that the worst student in $M_{0.5}(p_{j'})$ can be no worse than the worst student in $M_0(p_{j'})$. Thus since student $s_i$ cannot have changed phase from point $T_0$, $s_i$ must be either a worst student in $M_{0.5}(l_k)$, or be as bad as a worst student in $M_{0.5}(l_k)$. By Propositions \ref{prop-spa-st-subsequentprecarious} and \ref{th:stable_SPA_S_lecfullworseprec}, since $l_k$ is full and non-precarious at $T_{0.5}$, $l_k$ cannot be assigned a worse student than $s_i$ between $T_{0.5}$ to $T_1$, and so as $s_i$ is assigned to $l_k$ at $T_1$ and $s_i$ remains in the same phase until $T_1$, then they must be one of the worst students in $M_1(l_k)$.

\item Suppose then that $p_{j'}$ is undersubscribed at $T_0$. Since a preference element is being removed, $l_k$ must have been full, non-precarious and either $s_i$ is a worst assignee in $M_0(l_k)$ (Line \ref{Alg:APPROX_SPA-S_stable:rempref1}), or $l_k$ ranks $s_i$ at least as badly as a worst student in $M_0(l_k)$ (Line \ref{Alg:APPROX_SPA-S_stable:rempref3}). But then by Propositions \ref{prop-spa-st-subsequentprecarious} and \ref{th:stable_SPA_S_lecfullworseprec}, $l_k$ must have remained non-precarious until $T_1$ and been unable to assign to a student worse than or equal to $s_i$, including $s_i$ in the same phase, a contradiction.
\end{itemize}

Therefore, for any blocking pair $(s_i,p_{j'})$ of type $(3bi)$ of $M_1$, $s_i$ must be one of the worst students in $M_1(l_k)$.
\end{proof}

% THEOREM BEGIN
\begin{prop} 
\label{th:stable_SPA_S_3bicreatedworstinmlk} 
In Algorithm \ref{Alg:APPROX_SPA-S_stable_3bibpBrief}, if a blocking pair $(s_{i'},p_j)$ of type $(3bi)$ is created (in the process of removing a different blocking pair of type $(3bi)$) then $s_{i'}$ must be one of the worst students in $M_2(l_k)$, where $l_k$ is the lecturer who offers $p_j$ and $M_2$ is the matching constructed immediately after this removal occurs.

	\end{prop}
% THEOREM END

\begin{proof}
	Let $M_0$ denote the matching at the end of the main while loop of Algorithm \ref{Alg:APPROX_SPA-S_stable}, and let $T_0$ denote this point in the algorithm's execution. Assume that during Algorithm \ref{Alg:APPROX_SPA-S_stable_3bibpBrief}'s execution, the first promotion to reveal a blocking pair $(s_{i'},p_j)$ of type $(3bi)$ occurs such that $s_{i'}$ is not a worst student in $M_1(l_k)$. Let $M_1$ be the matching constructed just before this promotion occurs. Suppose the promotion involves student $s_i$ moving from a less preferred $p_j$ to a more preferred project $p_{j'}$. It is clear that in the removal of blocking pairs of type $(3bi)$ there is no change in regards to which students are assigned to projects of $l_k$, therefore the same students are assigned to each lecturer in $M_0$, $M_1$ and $M_2$. By Proposition \ref{th:stable_SPA_S_3biworstinmlk}, $s_i$ is and remains one of the worst assignees of $l_k$ in $M_0$, $M_1$ and $M_2$. Since $s_{i'}$ is not a worst assignee in $M_2(l_k)$, $l_k$ must prefer $s_{i'}$ to $s_i$. But this would mean $(s_{i'},p_j)$ was a blocking pair of type $(3c)$ in $M_0$ if $p_j$ were full or $(3bii)$ in $M_0$ if $p_j$ were undersubscribed, both a contradiction to Lemma \ref{approx_end_while_only_3bi}.
\end{proof}

In Proposition \ref{th:stable_SPA_S_finalpartonlybp3bi} we show that only blocking pairs of type $(3bi)$ may be created in $M$ after the main while loop of Algorithm \ref{Alg:APPROX_SPA-S_stable}.

% THEOREM BEGIN
\begin{prop} 
	\label{th:stable_SPA_S_finalpartonlybp3bi}
	It is not possible in Algorithm \ref{Alg:APPROX_SPA-S_stable_3bibpBrief}'s execution, for a blocking pair of any type other than $(3bi)$ to be created.
	\end{prop}
% THEOREM END

\begin{proof}
Let $M_0$ denote the matching constructed immediately after the main while loop of Algorithm \ref{Alg:APPROX_SPA-S_stable} terminates and let $T_0$ denote this point in the algorithm's execution. By Lemma \ref{approx_end_while_only_3bi}, only blocking pairs of type $(3bi)$ of $M_0$ may exist at $T_0$ therefore we restrict our attention to the removal of such pairs. Assume for a contradiction that during Algorithm \ref{Alg:APPROX_SPA-S_stable_3bibpBrief}'s execution, the first promotion to reveal a blocking pair of type not equal to $(3bi)$ occurs. Let $M_1$ (respectively $M_2$) be the matching constructed just before (respectively after) this promotion occurs with $T_1$ (respectively $T_2$) denoting this point in the algorithm's execution. Suppose that this promotion involves student $s_i$ being promoted from project $p_j$ to project $p_{j'}$ as pair $(s_i,p_{j'})$ is a blocking pair of $M_1$ of type $(3bi)$. Since $(s_i,p_{j'})$ is a blocking pair of $M_1$ of type $(3bi)$ we know that $p_j$ and $p_{j'}$ are both offered by the same lecturer, say $l_k$. Assume that this promotion has now revealed a blocking pair $(s_{i'},p_j)$ of type $(3a)$, $(3bii)$ or $(3c)$ in $M_2$. We look at each case in turn.

\begin{itemize}
	\item $(3a)$ - Since in $M_1$, $(s_i,p_{j'})$ was a blocking pair of type $(3bi)$ we know that $l_k$ is full at $T_1$. The promotion involves moving $s_i$ from one project offered by $l_k$ to another, therefore at $T_2$, $l_k$ must be full and so $p_j$ cannot be involved in a blocking pair of type $(3a)$ in $M_2$, a contradiction.

	\item $(3bii)$ - Suppose $(s_{i'},p_j)$ is a blocking pair of type $(3bii)$ in $M_2$. Since it is of type $(3bii)$, $l_k$ must prefer $s_{i'}$ to a worst assignee in $M_2(l_k)$ (and consequently $M_1(l_k)$ as students do not change lecturer). If $p_j$ was undersubscribed at $T_1$ then $(s_{i'},p_j)$ would have constituted a blocking pair of type $(3bii)$, a contradiction. Therefore $p_j$ must have been full in $M_1$. We know that $(s_i,p_j) \in M_1$ and that $s_i$ is a worst assignee in $M_1(l_k)$ by Proposition \ref{th:stable_SPA_S_3biworstinmlk} and \ref{th:stable_SPA_S_3bicreatedworstinmlk}, therefore $l_k$ prefers $s_{i'}$ to $s_i$. It follows that $(s_{i'},p_j)$ would have constituted a blocking pair in $M_1$ of type $(3c)$, a contradiction to the fact that no blocking pair of any type other than $(3bi)$ was revealed prior to $T_2$.
	 
	\item $(3c)$ - Suppose finally that $(s_{i'},p_j)$ is a blocking pair of $M_2$ of type $(3c)$. But blocking pairs of type $(3c)$ require $p_j$ to be full in $M_2$ which it cannot be since $(s_i,p_j)$ has been removed just before $T_2$, hence $p_j$  cannot be involved in a blocking pair of type $(3c)$.
\end{itemize}
Therefore it is not possible for a blocking pair of type $(3a)$, $(3bii)$ or $(3c)$ to be created during the first promotion of a student, and hence any promotion.
\end{proof}

Finally, Theorem \ref{approx_not_destable} proves that any matching produced by Algorithm {\sf Max-SPA-ST-Approx} is stable.

% THEOREM BEGIN
\begin{theorem} 
	\label{approx_not_destable}
		Any matching produced by Algorithm \ref{Alg:APPROX_SPA-S_stable} must be stable.
\end{theorem}
% THEOREM END

\begin{proof} 
Let $M_0$ be the matching constructed immediately after the termination of the main while loop of Algorithm \ref{Alg:APPROX_SPA-S_stable} and let $T_0$ denote this stage of the algorithm. Recall that by Lemma \ref{approx_end_while_only_3bi}, only blocking pairs of type $(3bi)$ may exist relative to $M_0$. Also, by Lemma \ref{th:stable_SPA_S_finalpartonlybp3bi}, no blocking pair of any other type can exist relative to the matching constructed after $T_0$. 
	
	Algorithm \ref{Alg:APPROX_SPA-S_stable_3bibpBrief} systematically removes blocking pairs of type $(3bi)$ in a series of student promotions. Each promotion improves the outcome for a student.
	
	Therefore there are no blocking pairs of any type in the finalised matching $M_s$ and so $M_s$ is stable. 
\end{proof}

Since this proof relies only on the fact that $p_{b_2}$ is removed from $s_{b_1}$'s list once for $(s_{b_1},p_{b_2})$ not to become a blocking pair, we can infer that if we allowed students to only iterate once through their preference preference list rather than twice, this would still result in a stable matching.

%%%%%%%%%%%%%%%%%%%%%%%%%%%%%%%%%%%%%%%%%%%%%%%%%%%
%%%%%%%%%%%%%%%%%%%%%%%%%%%%%%%%%%%%%%%%%%%%%%%%%%%
\subsection{Time complexity and termination}
\label{spa-st-proofs-sec-time}

In this section we prove that Algorithm {\sf Max-SPA-ST-Approx} runs in linear time with respect to the total length of student preference lists.

First, in Proposition \ref{th:stable_SPA_S_maxapplications}, we show that during the main while loop each student may only apply to a project on their preference list a maximum of three times.

% THEOREM BEGIN
\begin{prop} 
	\label{th:stable_SPA_S_maxapplications}
	The maximum number of times a student $s_i$ can apply to a project $p_j$ on their preference list during the main while loop of Algorithm \ref{Alg:APPROX_SPA-S_stable} is three.
	\end{prop}
% THEOREM END

\begin{proof}
First we note that as soon as $s_i$ removes $p_j$ from their preference list once during Algorithm \ref{Alg:APPROX_SPA-S_stable}'s execution, $(s_i,p_j)$ cannot subsequently become a precarious pair by definition (since a precarious pair must be assigned in phase $1$).

Focussing on the main while loop, assume for some iteration, that phase $1$ student $s_i$ applies to project $p_j$ on their preference list. Either $(s_i,p_j)$ is added to the matching being built $M$, or $p_j$ is removed from $s_i$'s list. If $p_j$ is removed from $s_i$'s list then $s_i$ may still apply to project $p_j$ in phase $2$ but as noted above $(s_i,p_j)$ cannot become a precarious pair.
	
	Assume instead that $(s_i,p_j)$ is added to $M$. If it remains in $M$ until the algorithm completes then $s_i$ cannot apply to $p_j$ again. So assume that $(s_i,p_j)$ is removed from $M$ at some point due to another pair being added to $M$. If $(s_i,p_j)$ was non-precarious at the point it is removed from $M$ then $s_i$ removes $p_j$ from their list and the next time $s_i$ could apply to $p_j$ is when $s_i$ is in phase $2$ when as above $(s_i,p_j)$ cannot become a precarious pair. 
	
	Assume therefore that $(s_i,p_j)$ was precarious when removed from $M$. Then $s_i$ does not remove $p_j$ from their list and $s_i$ can again apply to $p_j$ during phase $1$. Note that if $(s_i,p_j)$ is re-added to $M$ it must be as a non-precarious pair. This is because, using similar reasoning that was used in Proposition \ref{prop-spa-st-subsequentprecarious}, at the point at which $s_i$ reapplies to $p_j$ they must have exhausted all fully available projects at the head of their list, therefore $(s_i,p_j)$ cannot again become precarious. Therefore, $s_i$ can apply to $p_j$ a maximum of three times during the execution of the while loop: at most twice while $s_i$ is in phase $1$ (twice only if $(s_i,p_j)$ is removed as a precarious pair) and at most once in phase $2$.
\end{proof}

Next, using the data structures summarised in Figure \ref{alg_SPA_ST_fig_datastructures}, Lemma \ref{th:alg_SPA_S_operations_constant} proves that all operations inside the main while loop of Algorithm \ref{Alg:APPROX_SPA-S_stable} run in constant time.

% THEOREM BEGIN
\begin{lemma} 
	\label{th:alg_SPA_S_operations_constant}
	All operations inside the main while loop of Algorithm \ref{Alg:APPROX_SPA-S_stable} run in constant time.
	\end{lemma}
% THEOREM END

\begin{proof}
The data structures required are described below and are summarised in Figure \ref{alg_SPA_ST_fig_datastructures}. For initialisation purposes, each student, project and lecturer has a list of length $n_2$, $n_1$ and $n_1$ respectively, each entry of which requires $O(1)$ space. In order to not exceed a time complexity of order the sum of lengths of preference lists, a process of virtual initialisation is used on these data structures \cite[p. 149]{BB96}.

\begin{figure}

\begin{tikzpicture}
\newcommand*{\TikzSstartY}{-0.5}%
\newcommand*{\TikzPstartY}{\TikzSstartY-6}%
\newcommand*{\TikzLstartY}{\TikzPstartY-7}%
\newcommand*{\TikzMstartY}{\TikzLstartY-5}%
\newcommand*{\TikzKstartY}{\TikzMstartY-2.5}%

% single boxes
\foreach \y in {-3.3+\TikzSstartY, -4+\TikzSstartY, -4.7+\TikzSstartY, -5.1+\TikzPstartY, -5.8+\TikzPstartY, \TikzLstartY-4, \TikzKstartY-1}
{
\draw[gray,very thin] (0,0+\y) -- (1.5,0+\y) -- (1.5,0.5+\y) -- (0,0.5+\y) -- cycle;
}

% linked lists
\foreach \y in {-2.6+\TikzSstartY, -0.5+\TikzPstartY, -2.2+\TikzPstartY, \TikzLstartY-0.5, \TikzLstartY-2.2, \TikzMstartY-1.5, \TikzKstartY-1.7}
{
\draw[gray,very thin] (0,0+\y) -- (1.5,0+\y) -- (1.5,0.5+\y) -- (0,0.5+\y) -- cycle;
\draw[gray,very thin] (1.5,0+\y) -- (3,0+\y) -- (3,0.5+\y) -- (1.5,0.5+\y);
\draw[gray,very thin] (3,0+\y) -- (4.5,0+\y) -- (4.5,0.5+\y) -- (3,0.5+\y);
\draw[gray,very thin] (4.5,0+\y) -- (5,0+\y);
\draw[gray,very thin] (4.5,0.5+\y) -- (5,0.5+\y);

\draw[dashed] (5.2,0.25+\y) -- (5.8,0.25+\y);
}

% arrays ends
\foreach \y in {-2.6+\TikzSstartY, -0.5+\TikzPstartY, -2.2+\TikzPstartY, \TikzLstartY-0.5, \TikzLstartY-2.2, \TikzMstartY-1.5, \TikzKstartY-1.7}
{
\draw[gray,very thin] (6,0+\y) -- (6.5,0+\y);
\draw[gray,very thin] (6,0.5+\y) -- (6.5,0.5+\y);
\draw[gray,very thin] (6.5,0+\y) -- (8,0+\y) -- (8,0.5+\y) -- (6.5,0.5+\y) -- cycle;
}

% doubly linked list embedded in an array
\foreach \y in {\TikzSstartY-0.7, \TikzPstartY-3.1, \TikzPstartY-4.2, \TikzLstartY-3.1, \TikzKstartY-2.6}
{
\draw[gray,very thin] (0.2,0+\y) -- (1.7,0+\y) -- (1.7,0.5+\y) -- (0.2,0.5+\y) -- cycle;
\draw[gray,very thin] (1.7,0+\y) -- (3.2,0+\y) -- (3.2,0.5+\y) -- (1.7,0.5+\y);
\draw[gray,very thin] (3.2,0+\y) -- (4.7,0+\y) -- (4.7,0.5+\y) -- (3.2,0.5+\y);
\draw[gray,very thin] (4.7,0+\y) -- (5.2,0+\y);
\draw[gray,very thin] (4.7,0.5+\y) -- (5.2,0.5+\y);
\draw[dashed] (5.4,0.25+\y) -- (6,0.25+\y);

\draw[gray,very thin] (6.2,0+\y) -- (6.7,0+\y);
\draw[gray,very thin] (6.2,0.5+\y) -- (6.7,0.5+\y);
\draw[gray,very thin] (6.7,0+\y) -- (8.2,0+\y) -- (8.2,0.5+\y) -- (6.7,0.5+\y) -- cycle;

\draw[gray,very thin] (2,-0.2+\y) -- (0,-0.2+\y) -- (0,0.7+\y) -- (2,0.7+\y);
\draw[gray,very thin] (6.4,-0.2+\y) -- (8.4,-0.2+\y) -- (8.4,0.7+\y) -- (6.4,0.7+\y);
}

% extra student things
\node[anchor=west] at (0,0.5+\TikzSstartY) {\textbf{Student} $\boldsymbol{s_i}$}; 
\node[anchor=west] at (8.5,-0.45+\TikzSstartY) {{\tt prefList}, (project, rank) tuples};

\node[anchor=west] at (0.3,-0.45+\TikzSstartY) {($p_3$, $1$)};
\node[anchor=west] at (1.8,-0.45+\TikzSstartY) {($p_2$, $1$)};
\node[anchor=west] at (3.3,-0.45+\TikzSstartY) {($p_1$, $1$)};

\draw [->] (0.75,-1.4+\TikzSstartY) -- (0.75,-0.9+\TikzSstartY);
\draw [->] (3.75,-1.4+\TikzSstartY) -- (3.75,-0.9+\TikzSstartY);
\node[] at (0.75,-1.65+\TikzSstartY) {{\tt first}};
\node[] at (3.75,-1.65+\TikzSstartY) {{\tt second}};

\node[anchor=west] at (8.5,-2.35+\TikzSstartY) {{\tt projPosition}};
\node[anchor=west] at (0.5,-2.35+\TikzSstartY) {$3$};
\node[anchor=west] at (2,-2.35+\TikzSstartY) {$2$};
\node[anchor=west] at (3.5,-2.35+\TikzSstartY) {$1$};

\node[anchor=west] at (2,-3.05+\TikzSstartY) {phase of $s_i$};
\node[anchor=west] at (0.5,-3.05+\TikzSstartY) {$1$};

\node[anchor=west] at (2,-3.75+\TikzSstartY) {{\tt firstFin}, True if {\tt first} pointer reaches end of tie};
\node[anchor=west] at (0.5,-3.75+\TikzSstartY) {F};
\node[anchor=west] at (2,-4.45+\TikzSstartY) {{\tt secondFin}, True if {\tt second} pointer reaches end of tie};
\node[anchor=west] at (0.5,-4.45+\TikzSstartY) {F};

% extra project things
\node[anchor=west] at (0,0.5+\TikzPstartY) {\textbf{Project} $\boldsymbol{p_j}$};
\node[anchor=west] at (8.5,-0.25+\TikzPstartY) {{\tt projectedPrefList},};
\node[anchor=west] at (8.5,-0.7+\TikzPstartY) {(student, rank, boolean) tuples};
\draw [->] (3.75,-1+\TikzPstartY) -- (3.75,-0.5+\TikzPstartY);
\node[] at (3.75,-1.25+\TikzPstartY) {{\tt last}};
\node[anchor=west] at (-0.12,-0.25+\TikzPstartY) {\small ($s_7$, $1$, T)};
\node[anchor=west] at (1.4,-0.25+\TikzPstartY) {\small ($s_4$, $1$, F)};
\node[anchor=west] at (2.9,-0.25+\TikzPstartY) {\small ($s_6$, $2$, T)};

\node[anchor=west] at (8.5,-1.95+\TikzPstartY) {{\tt studentPositions}};
\node[anchor=west] at (0.3,-1.95+\TikzPstartY) {null};
\node[anchor=west] at (2,-1.95+\TikzPstartY) {$5$};
\node[anchor=west] at (3.3,-1.95+\TikzPstartY) {null};

\node[anchor=west] at (8.5,-2.85+\TikzPstartY) {{\tt precariousList}};
\node[anchor=west] at (0.5,-2.85+\TikzPstartY) {null};
\node[anchor=west] at (2,-2.85+\TikzPstartY) {null};
\node[anchor=west] at (3.5,-2.85+\TikzPstartY) {null};

\node[anchor=west] at (8.5,-3.95+\TikzPstartY) {{\tt supportList}};
\node[anchor=west] at (0.5,-3.95+\TikzPstartY) {null};
\node[anchor=west] at (2,-3.95+\TikzPstartY) {null};
\node[anchor=west] at (3.6,-3.95+\TikzPstartY) {$s_3$};

\node[anchor=west] at (2,-4.85+\TikzPstartY) {lecturer offering $p_j$};
\node[anchor=west] at (0.5,-4.85+\TikzPstartY) {$l_5$};

\node[anchor=west] at (2,-5.55+\TikzPstartY) {number of allocations in $M$};
\node[anchor=west] at (0.5,-5.55+\TikzPstartY) {$2$};

% extra lecturer things
\node[anchor=west] at (0,0.5+\TikzLstartY) {\textbf{Lecturer} $\boldsymbol{l_k}$};
\node[anchor=west] at (8.5,-0.25+\TikzLstartY) {{\tt prefList}, (student, rank, boolean) tuples};
\node[anchor=west] at (-0.1,-0.25+\TikzLstartY) {\small ($s_8$, $1$, F)};
\node[anchor=west] at (1.4,-0.25+\TikzLstartY) {\small ($s_2$, $2$, T)};
\node[anchor=west] at (2.9,-0.25+\TikzLstartY) {\small ($s_3$, $2$, T)};

\draw [->] (8,-1+\TikzLstartY) -- (8,-0.5+\TikzLstartY);
\node[] at (8,-1.25+\TikzLstartY) {{\tt last}};
\node[anchor=west] at (8.5,-1.95+\TikzLstartY) {{\tt studentPositions}};
\node[anchor=west] at (0.3,-1.95+\TikzLstartY) {null};
\node[anchor=west] at (2,-1.95+\TikzLstartY) {$2$};
\node[anchor=west] at (3.5,-1.95+\TikzLstartY) {$3$};

\node[anchor=west] at (8.5,-2.75+\TikzLstartY) {{\tt precariousProjList}};
\node[anchor=west] at (0.5, -2.85+\TikzLstartY) {null};
\node[anchor=west] at (2.15,-2.85+\TikzLstartY) {$p_2$};
\node[anchor=west] at (3.5,-2.85+\TikzLstartY) {null};

\node[anchor=west] at (2,-3.75+\TikzLstartY) {number of allocations in $M$};
\node[anchor=west] at (0.5,-3.75+\TikzLstartY) {$3$};

% matching things
\node[anchor=west] at (0,\TikzMstartY) {\textbf{Matching} $\boldsymbol{M}$};
\node[anchor=west] at (0,-0.7+\TikzMstartY) {{\tt matchArray}, cell $i-1$ contains $p_j$ if $(s_i,p_j)\in M$ or null if $s_i$ is unassigned};
\draw[gray,very thin] (0,-2+\TikzMstartY) -- (\textwidth,-2+\TikzMstartY);
\node[anchor=west] at (0.3,-1.25+\TikzMstartY) {null};
\node[anchor=west] at (1.9,-1.25+\TikzMstartY) {$p_6$};
\node[anchor=west] at (3.3,-1.25+\TikzMstartY) {null};

% key things
\node[anchor=west] at (0,0\TikzKstartY) {\textbf{Key}};
\node[anchor=west] at (8.5,-0.75+\TikzKstartY) {boolean/int/link/tuple};
\node[anchor=west] at (8.5,-1.45+\TikzKstartY) {Array};
\node[anchor=west] at (8.5,-2.35+\TikzKstartY) {Doubly linked list embedded in};
\node[anchor=west] at (8.5,-2.75+\TikzKstartY) {an array};

\end{tikzpicture}
\caption[Data structures guide for Lemma \ref{th:alg_SPA_S_operations_constant}]{Data structures guide for Lemma \ref{th:alg_SPA_S_operations_constant}.}
\label{alg_SPA_ST_fig_datastructures}
\end{figure}

\textbf{Student data structures.} For each student a doubly-linked list of (project, rank) tuples embedded in an array, {\tt prefList}, stores their preference list in order of rank, representing the undeleted entries. A small example is shown in Figure \ref{alg_SPA_ST_fig_datastructures} with $p_3$, $p_2$ and $p_1$ all of rank $1$ on $s_i$'s preference list. Entries may be deleted from this array; a copy of this list prior to any deletions being carried out is retained in order to allow a second iteration through a student's preference list, if they move into phase $2$. An array {\tt projPosition} of length $n_2$ retains links to the position of (project, rank) tuples in {\tt prefList}, allowing a constant time removal of projects from {\tt prefList}. An integer variable associated with each student stores which phase this student is in. Examples for these final two data structures are also shown in Figure  \ref{alg_SPA_ST_fig_datastructures}.

\textbf{Project data structures.} Each project has a link to their supervising lecturer. An array, {\tt projectedPrefList} stores the projected preference list of $l_k$ for $p_j$ in the form of (student, rank, boolean) tuples. As an example, suppose $p_j$ has a projected preference list starting with $s_7$ at rank $1$, $s_4$ at rank $1$ and $s_6$ at rank $2$ as is shown in Figure \ref{alg_SPA_ST_fig_datastructures}. The boolean values indicate which student-project pairs are currently in the matching.

Once a project is full and non-precarious it cannot accept a worse student than it already has for the remainder of the algorithm, according to Proposition \ref{prop-spa-st-subsequentprecarious}. Assume $p_j$ is full and non-precarious.  Let the worst student assigned to $p_j$ be given by $s_w$. We retain a pointer, {\tt last}, which points to the rightmost student at the same rank as $s_w$ in $p_j$'s {\tt projectedPrefList}. This pointer must move from right to left in a linear fashion (moving up in ranks) given the above proposition. 

During the course of the algorithm, we may need to remove the worst student according to $l_k$ from $M(p_j)$. It is possible that there are two or more students who are worst assignees (according to rank) with some being in phase $1$ and some in phase $2$. In order to ensure that we prioritise the removal of phase $1$ students, two pointers are added for each entry in {\tt projectedPrefList}, which point to the head of a phase $1$ and a phase $2$ doubly-linked list associated with that tie embedded in the {\tt projectedPrefList} array (this data structure is not shown in Figure \ref{alg_SPA_ST_fig_datastructures}). Adding or removing a phase $1$ or $2$ student to either list takes constant time, as they do not need to be kept in order. Then, un-assigning a student requires a check to be made in the tie associated with the worst position (found using {\tt last}), in order to prioritise a phase $1$ student's removal. In total this takes constant time. Note that a student can only change phase if they are not allocated and therefore updating an allocated student's phase in these lists is not necessary unless they have just been added.

Each project also contains a doubly-linked list embedded in an array of students, denoted by {\tt precariousList}, containing students who have formed a precarious pair with this project. In the example in Figure \ref{alg_SPA_ST_fig_datastructures}, $p_j$ is non-precarious and so no students form a precarious pair with $p_j$. Adding to and removing from this list takes constant time if we assume that $s_i$ is stored at index $i-1$. A project $p_j$ \emph{supports} a student $s_i$ in being precarious if $(s_i,p_{j'})\in M$ and $p_j$ is the first fully available project at the same rank as $p_{j'}$ in $s_i$'s list. Then, a doubly-linked list embedded in an array of students, {\tt supportList}, stores the students for which $p_j$ gives their support. As before, adding to and removing from this list takes constant time. A counter stores the number of students assigned to $p_j$ in $M$.

\textbf{Lecturer data structures.} For each lecturer an array of (student, rank, boolean) tuples, {\tt prefList}, stores their preference list in order of rank, with a True value stored in the $i$th boolean if student $s_i$ is assigned to a project of $l_k$'s. Figure \ref{alg_SPA_ST_fig_datastructures} shows an example with $s_8$ not assigned to $l_k$ at rank $1$ and $s_2$ and $s_3$ both assigned at rank $2$. Each lecturer also has an array of length $n_1$, {\tt studentPositions} which retains links to the position of (student, rank, boolean) tuples in $l_k$'s {\tt prefList}, and a counter stores the number of students assigned to $l_k$ in $M$. A doubly-linked list embedded in an array of projects, {\tt precariousProjList}, stores the projects offered by $l_k$ that are precarious, where project $p_j$ is stored at index $j-1$ in this list if it is precarious. Figure \ref{alg_SPA_ST_fig_datastructures} shows $p_2$ being a precarious project of $l_k$.

Similar to projects, by Proposition \ref{prop-spa-st-subsequentprecarious}, once a lecturer is full and non-precarious they cannot accept a worse student than they already have for the remainder of the algorithm. Assume $l_k$ is full and non-precarious. Using similar data structures described in the Project Data Structures section above we are able to find $l_k$'s worst assigned student in constant time.

\textbf{Student pointers.} A student $s_i$ retains two pointers for the project(s) tied at the head of their list. One pointer {\tt first} stores the first fully available project when iterating from left to right, and {\tt second} stores the second fully available project. If {\tt first} (respectively {\tt second}) reaches the end of the tie, then a boolean {\tt firstFin} (respectively {\tt secondFin}) is set to True. For each iteration of the main while loop of Algorithm \ref{Alg:APPROX_SPA-S_stable}, each student $s_i$ first seeks a project at the head of their list that is fully available. This will be precisely the project $p_j$ that {\tt first} points to. Then if {\tt second} has not reached the end of the tie, $(s_i,p_j)$ (if added to the matching) is precarious and the project that {\tt second} points to, $p_{j'}$, supports $p_j$. If however, {\tt secondFin} is set to True, then $(s_i,p_j)$ (if added to the matching) is non-precarious. Finally, if {\tt firstFin} is set to True, then the leftmost project at the head of $s_i$'s list is a favourite project, with $(s_i,p_j)$ also being unable to become precarious. Proposition \ref{th:stable_SPA_S_maxapplications} shows that maximum number of applications a student can make to a project on their preference list is $3$. At most twice in phase $1$ (twice if removed as a precarious pair) and once in phase $2$.

During phase $2$ there are no fully available projects on $s_i$'s list since $s_i$ must have applied and been rejected (in some way) from every project on their preference list at least once already. Therefore, the {\tt first} and {\tt second} pointers are not required in phase $2$. During phase $1$ the the {\tt first} and {\tt second} pointers are only required to iterate once over each tie as described above. Hence, the maximum number of times a student's list is iterated over is $4$; once each for the two pointers {\tt first} and {\tt second}, once again after {\tt first} and {\tt second} have reached the end of the tie at the head of a student's list (the student may have retained projects at the head of their list after this point if the projects were precarious), and finally once during phase $2$.

\textbf{Matching data structures.} The current matching is stored in an array of cells {\tt matchArray} where cell $i-1$ contains project $p_j$ if $(s_i,p_j) \in M$ or null otherwise.  Figure \ref{alg_SPA_ST_fig_datastructures} shows and example with student $s_2$ being assigned to project $p_6$.

\textbf{Processes (in the order encountered in Algorithm \ref{Alg:APPROX_SPA-S_stable}):}

Let lecturer $l_k$ offer project $p_j$.
\begin{enumerate}

	\item \emph{A student $s_i$ applies to a favourite project}: if $s_i$'s {\tt FirstFin} is set to True then there are no fully available projects at the head of $s_i$'s list, and a favourite project of $s_i$ will be the leftmost project. If however, {\tt FirstFin} is False then there are fully available projects at the head of their list and a favourite project of $s_i$ is pointed to by {\tt first}, which is retrievable in constant time.
	
		\item \emph{Deciding if a project $p_j$ is undersubscribed or full or deciding if a lecturer $l_k$ is undersubscribed or full}: Using the counters described above a comparison can be made between $p_j$'s capacity and their current number of allocations. A similar comparison can be made for $l_k$. Both can be achieved in constant time.
	
		\item \emph{Deciding if project $p_j$ is fully available}: $p_j$ would not be fully available if either $p_j$ is full or $l_k$ is full. Therefore a comparison of the number of allocations for $p_j$ and $l_k$ and their respective capacities is required. Again this can be achieved in constant time.

\item \emph{Adding a pair $(s_i,p_j)$ to $M$}:
Project $p_j$ is placed in the $i-1$'th cell of {\tt matchArray} and $p_j$ and $l_k$'s allocation counters are incremented. Project $p_j$'s {\tt projectedPrefList} and lecturer $l_k$'s {\tt prefList} booleans are updated in constant time using their associated {\tt studentPositions} data structures. Each tuple in these lists has a link to the head of phase $1$ and phase $2$ lists for their tie. When the pair is added the tuple is added to either the phase $1$ or phase $2$ list. If $(s_i,p_j)$ is precarious then $s_i$ is added to $p_j$'s {\tt precariousList} and the project pointed to by {\tt second}, $p_{j'}$ adds $s_i$ to their {\tt supportList}.  If $p_j$ has just changed from being non-precarious to precarious then $l_k$ adds $p_j$ to their {\tt precariousProjList}. If the addition of $(s_i,p_j)$ to $M$ means that $p_j$ goes from being fully available to not being fully available then we need to ensure that other students who rely on $p_j$ as their support are updated. Therefore $p_j$ alerts each student on their {\tt supportList} that they are no longer to be relied upon as a fully available project. This triggers each of those students to update their {\tt second} pointers. The time required for this can be attributed to the movement of {\tt second} pointers as noted earlier. After being alerted, some other pair in $M$ may stop being precarious, but any changes can be conducted in constant time as described above. If on adding pair $(s_i,p_j)$, $p_j$ has now become full and non-precarious then the {\tt last} pointer will move from right to left over {\tt projectedPrefList} until it reaches the end of a tie whose phase $1$ and phase $2$ lists are non-empty. From this point on {\tt last} is only updated upon removing a pair from $M$ (Point \ref{point:removingpair}).

	\item \label{point:decideLKprec} \emph{Deciding if $l_k$ is precarious, and returning a precarious project if one exists}: Checking whether {\tt precariousProjList} is empty for $l_k$ is a simple process that takes constant time. Retrieving a precarious pair should one exist requires selection of the first student from $l_k$'s {\tt precariousProjList} and can be done in constant time.

	\item \label{point:findworstLK} \emph{Finding a worst assignee of $l_k$ and deciding if $l_k$ meta-prefers $s_i$ to this worst assignee}: This operation only needs to be executed if $l_k$ is full and non-precarious. In that situation $l_k$'s {\tt last} pointer will point to the rightmost position in a tie in {\tt prefList} such that $l_k$'s current worst student $s_w$ is assigned at the same rank. Then as previously discussed, all that is required is to check the links to phase $1$ and phase $2$ students for this tie and return a phase $1$ student if one exists, or phase $2$ student if not. This can be conducted in constant time. Deciding if $l_k$ meta-prefers $s_i$ to $s_w$ can also be done in constant time by comparing rank and phase.

	\item \emph{Removing a preference list entry from $s_i$'s list}: This process is shown in Algorithm \ref{Alg:APPROX_SPA-S_stable_extra}	which runs in constant time, since we can find a specific project $p_j$ in $s_i$'s {\tt prefList} using the {\tt projPosition} array.

\item \label{point:removingpair} \emph{Removing a pair $(s_i,p_j)$ from $M$}: 
The $i-1$'th cell of {\tt matchArray} is set to null, and $p_j$ and $l_k$'s allocation counters are decremented. Project $p_j$'s {\tt projectedPrefList} and lecturer $l_k$'s {\tt prefList} booleans are updated in constant time. The tuples associated with $s_i$ in these lists are removed from their phase $1$ or phase $2$ list in constant time. If a pair $(s_i,p_j)$ is removed from $M$, then this is either because $p_j$ or $l_k$ is full. By Proposition \ref{th:stable_SPA_S_propfullyavailable}, $p_j$ cannot subsequently become fully available. Thus, the removal of a pair cannot change $p_j$'s fully available status. All that is required then is to check whether $(s_i,p_j)$ was precarious, and update $p_j$'s {\tt precariousList} and $l_k$'s {\tt precariousProjList}. If on removing pair $(s_i,p_j)$, the {\tt last} pointer now points to a tie with empty phase $1$ and phase $2$ lists, {\tt last} needs to be updated and accordingly moves from right to left until it reaches the end of a tie with a non-empty phase $1$ or phase $2$ list.

	\item \emph{Deciding if $p_j$ is precarious, and returning a precarious pair if one exists}: Similar to Point \ref{point:decideLKprec} above but using the {\tt precariousList} of $p_j$.

\item \emph{Finding a worst assignee of $p_j$ according to $l_k$ and deciding if $l_k$ meta-prefers $s_i$ to this worst assignee}: Similar to Point \ref{point:findworstLK} above, this operation is only required if $p_j$ is full and non-precarious, at which point $p_j$'s {\tt last} pointer will point to the rightmost position in a tie in $\tt projectedPrefList$ such that $l_k$'s current worst student assigned to $p_j$, $s_w$ is assigned at the same rank. As above retrieving $s_w$ and comparing its rank and phase with $s_i$ takes constant time.
\end{enumerate}

Therefore, all operations inside the main while loop of Algorithm \ref{Alg:APPROX_SPA-S_stable} run in constant time.
\end{proof}

Algorithm \ref{Alg:APPROX_SPA-S_stable_3bibp} is a more detailed version of Algorithm \ref{Alg:APPROX_SPA-S_stable_3bibpBrief}, indicating how the operations in Algorithm \ref{Alg:APPROX_SPA-S_stable_3bibpBrief} can be implemented efficiently. Proposition \ref{th:alg_SPA_S_polytime_promoting} shows that Algorithm \ref{Alg:APPROX_SPA-S_stable_3bibp} runs in linear time.

% ALGORITHM 
\begin{algorithm} [t]

  \caption[{\sf Promote-Students}($M$), subroutine for Algorithm \ref{Alg:APPROX_SPA-S_stable} (detailed view).]{{\sf Promote-Students}($M$), subroutine for Algorithm \ref{Alg:APPROX_SPA-S_stable} (detailed view). Removes all blocking pairs of type $(3bi)$.}
  \label{Alg:APPROX_SPA-S_stable_3bibp}
	\begin{algorithmic}[1]
          \Require \acrshort{spa-st} instance $I$ and matching $M$ which does not contain blocking pairs of type \emph{(3a)}, \emph{(3bii)} or \emph{(3c)}.
          \Ensure Return a stable matching $M$.

	\State Create data structures as described in Proposition \ref{th:alg_SPA_S_polytime_promoting}
	\While {$S \neq \emptyset$}
	\State Pop $p_j$ from stack $S$
	\State Remove the first student $s_i$ from list $\rho_j$
	\State Let $p_k = M(s_i)$
	
	\If {$s_i$ prefers $p_j$ to $p_k$} \Comment{$p_j$ is undersubscribed, $s_i$ is assigned and prefers $p_j$ to $M(s_i)$}
	\State $M \gets M\backslash \{(s_i,p_k)\}$
	\State $M \gets M \cup \{(s_i,p_j)\}$
	\State Let $\rho_k$ be the list of student and rank tuples associated with project $p_k$ and let boolean $\beta_k$ indicate whether $M(s_i)$ is on stack $S$
	\If {$\rho_k \neq \emptyset$}
	\State Push $p_k$ onto stack $S$ if it is not already on $S$ \Comment{Using $\beta_k$.}
	\EndIf
	\EndIf 
	
	\If {$\rho_j \neq \emptyset$ \textbf{and} $p_j$ is undersubscribed}
	
	\State Push $p_j$ onto stack $S$ \Comment{$p_j$ cannot currently be on $S$}
	\EndIf
	\EndWhile
	\State \Return $M$
	\end{algorithmic}
\end{algorithm}
% ALGORITHM END

% THEOREM BEGIN
\begin{prop} 
	\label{th:alg_SPA_S_polytime_promoting}
	The time complexity of Algorithm \ref{Alg:APPROX_SPA-S_stable_3bibp} is $O(m)$ where $m$ is the total length of student preference lists.
	\end{prop}
% THEOREM END

\begin{proof}

Abraham \emph{et al.} \cite{ACMM04} describes the process of a sequence of promotions for houses in \acrshort{ha} in order to return a \emph{trade-in-free} matching. A similar process is described here to remove all blocking pairs of type \emph{(3bi)}. We create the following data structures.

\begin{itemize}
	\item A linked list $\rho_j$ of students $s_i$ for each project $p_j$ such that $s_i$ is assigned in $M$ and finds $p_j$ acceptable. We may also start by assuming that $\rho_j$ involves only students who prefer $p_j$ to $M(s_i)$, however the time complexity is unaffected by this.

	\item A ranking list $r_i$ for each student $s_i$ built as an array such that $r_i=j$ contains the rank of $p_j$ for student $s_i$;
	
	\item A stack $S$ of undersubscribed projects $p_j$, such that $\rho_j$ is non-empty, is created;
	
	\item A variable $\beta_j$ for each project $p_j$ which records whether $p_j$ is already in $S$.

\end{itemize}

These data structures can be initialised in $O(m)$ time where $m$ is the total length of student preference lists.

Execution of Algorithm \ref{Alg:APPROX_SPA-S_stable_3bibp} proceeds as follows. For each iteration of the while loop a project $p_j$ is taken from stack $S$. Project $p_j$ must be undersubscribed and have non-empty list $\rho_j$. The first student $s_i$ from $\rho_j$, is removed and if $s_i$ would prefer to be assigned to $p_j$ than $p_k = M(s_i)$ (found by comparing ranking list entries for $p_j$ and $p_k$) then we remove pair $(s_i,p_k)$ from $M$ and add $(s_i,p_j)$. Now, $p_k$ is certainly an undersubscribed project and it is added to $S$ (unless it already exists on $S$). Whether or not $(s_i,p_j)$ is added to $M$, $p_j$ may still be undersubscribed. If $\rho_j$ is non-empty and $p_j$ is undersubscribed, then $p_j$ is added to $S$.

With each iteration we remove a tuple from some project's $\rho$ list. These lists must be finite because preference lists are finite and therefore Algorithm \ref{Alg:APPROX_SPA-S_stable_3bibp} will terminate with empty $S$. It is clear that Algorithm \ref{Alg:APPROX_SPA-S_stable_3bibp} will take $O(m)$ time where $m$ is the total length of student preference lists.
\end{proof}

Finally, Theorem \ref{th:alg_SPA_S_polytime} establishes that Algorithm {\sf Max-SPA-ST-Approx} runs in linear time with respect to the total length of student preference lists.

% THEOREM BEGIN
\begin{theorem} 
	\label{th:alg_SPA_S_polytime}
	Algorithm \ref{Alg:APPROX_SPA-S_stable} always terminates and runs in linear time with respect to the total length of student preference lists.
	\end{theorem}
% THEOREM END

\begin{proof}
By Proposition \ref{th:stable_SPA_S_maxapplications} each student can apply to a project on their preference list a maximum of three times during the main while loop of Algorithm \ref{Alg:APPROX_SPA-S_stable}. Since all operations within this while loop run in constant time by Proposition \ref{th:alg_SPA_S_operations_constant}, this part of the algorithm must run in $O(3m)=$ $O(m)$ time, where $m$ is the total length of student preference lists. Proposition \ref{th:alg_SPA_S_polytime_promoting} shows that Algorithm \ref{Alg:APPROX_SPA-S_stable_3bibp} also runs in $O(m)$ time, therefore so must Algorithm \ref{Alg:APPROX_SPA-S_stable}. Finally, since student preference lists are of finite length, Algorithm \ref{Alg:APPROX_SPA-S_stable} must terminate.
\end{proof}

%%%%%%%%%%%%%%%%%%%%%%%%%%%%%%%%%%%%%%%%%%%%%%%%%%%%%%%%%
%%%%%%%%%%%%%%%%%%%%%%%%%%%%%%%%%%%%%%%%%%%%%%%%%%%%%%%%%

\subsection{Performance guarantee}
\label{spa-st-proofs-sec-performance}

\subsubsection{Introduction}
In Section \ref{spa-st-proofs-sec-stability} we showed that Algorithm {\sf Max-SPA-ST-Approx} always produces a stable matching. In this section we show that any stable matching produced by this algorithm is always at least two-thirds of the size of a maximum stable matching. 

We begin in Section \ref{spa-st-proofs-sec-32-prelim} by giving preliminary definitions of an \emph{underlying graph} and \emph{mapped graph} of an instance of \acrshort{spa-st}, which are used throughout the proofs in this section. Both the \emph{underlying graph} and the \emph{mapped graph} are illustrated in an example in Section \ref{spa-st-proofs-sec-32-example}. In Section \ref{spa-st-proofs-sec-32-components} we move on to look at the possible structures that may exist in \emph{mapped graphs}. Finally, using these structures, in Section \ref{spa-st-proofs-sec-32-proofs32} we prove that Algorithm {\sf Max-SPA-ST-Approx} has a performance guarantee of $\frac{3}{2}$.

%%%%%%%%%%%%%%%%%%%%%%%%%%%%%%%%%%%%%%%%%%%%%%%%%%%%%%%%%
%%%%%%%%%%%%%%%%%%%%%%%%%%%%%%%%%%%%%%%%%%%%%%%%%%%%%%%%%
\subsubsection{Preliminary definitions}
\label{spa-st-proofs-sec-32-prelim}

The \emph{underlying graph} $G$ of an \acrshort{spa-st} instance $I$ consists of sets of student, project and lecturer vertices. Edges exist between a student vertex $s_i$ and a project vertex $p_j$ if $s_i$ finds $p_j$ acceptable. Edges exist between a project vertex $p_j$ and lecturer vertex $l_k$ if $l_k$ offers $p_j$.

We now introduce the notion of a \emph{mapped graph} $G'$ of the underlying graph $G$ of the \acrshort{spa-st} instance $I$. This graph is not created in Algorithm \ref{Alg:APPROX_SPA-S_stable}, but is intended only to make it easier to prove the performance guarantee in Section \ref{spa-st-proofs-sec-32-proofs32}. The \emph{mapped graph} $G'$ is created in the following way. Let all student vertices remain unchanged. Let $M$ be the matching found by Algorithm \ref{Alg:APPROX_SPA-S_stable} for instance $I$ and let $M_{opt}$ be a maximum stable matching in $I$. For each lecturer vertex $l_k$ we create multiple \emph{cloned} vertices $l_k^1 \dots l_k^{r_k}$, where $r_k = d_k - |M(l_k) \cap M_{opt}(l_k)|$ and $d_k$ is $l_k$'s capacity. In $G$ there are edges between students and projects, and projects and lecturers, whereas $G'$ contains only edges between students and lecturer clones. 

An $(s_i,p_j)$ edge in $G$ corresponds to an $(s_i,l_k^r)$ edge in $G'$, where $l_k^r$ denotes the $r$\textsuperscript{th} lecturer clone of lecturer $l_k$. 
Edges in $G'$ are given by $M'$ and $M_{opt}'$, defined below. $M'_{opt}$ edges are defined as follows. For each lecturer $l_k$, if $\bigcup_{p_{j'} \in P_k} \{M_{opt}(p_{j'}) \backslash M(p_{j'})\} = \{s_{i_1},...,s_{i_t}\}$ then add $(s_{i_r},l_k^r), (1 \leq r \leq t)$ to $M'_{opt}$, the mapped version of $M_{opt}$ in $G'$. $M'$ edges are then added using Algorithm \ref{SPA_ST_approx_cloned_M}. By using this algorithm we ensure that where possible, pairs of edges in $M'$ and $M'_{opt}$ involving the same project are assigned to the same lecturer clone in $G'$. According to Algorithm \ref{SPA_ST_approx_cloned_M} we do the following. A copy of $M \backslash M_{opt}$ is created and denoted $M_0$ which intuitively contains the set of student-project pairs that have not yet been mapped. $L_0$ is a copy of the set of all lecturer clone vertices, and $L_0'$ is the empty set. Intuitively, $L_0'$ will collect up any remaining lecturer clones, after pairs of edges in $M'$ and $M_{opt}'$ involving the same project are dealt with. For each lecturer clone $l_k^r \in L_0$, if there is an edge $(s_i,l_k^r)$ in $M'_{opt}$ for some $s_i$ then we let $p_j$ be the project assigned to $s_i$ in $M_{opt}$. If there is not, then  $l_k^r$ is added to $L_0'$. Assuming $(s_i,l_k^r) \in M'_{opt}$, then we check if there is an edge $(s_{i'},p_j)$ in $M_0$ for some student $s_{i'}$. Again, if there is not then $l_k^r$ is added to $L_0'$. If $(s_{i'},p_j) \in M_0$ for some student $s_{i'}$ then we add edge $(s_{i'},l_k^r)$ to $M'$ and remove $(s_{i'},p_j)$ from $M_0$. After all lecturer clones have been tested, then for each student-project pair $(s_i,p_j)$ remaining in $M_0$ we find an unused lecturer clone $l_k^r \in L_0'$, where $l_k$ offers $p_j$, and add $(s_i,l_k^r)$ to $M'$. Project vertices and all other edges are ignored in $G'$.

%%%%%%%%%%%%%%%%%%%%%%%%%%%%%%%%%%%%%%%%%%%%%%%%%%%%%%%%%
%%%%%%%%%%%%%%%%%%%%%%%%%%%%%%%%%%%%%%%%%%%%%%%%%%%%%%%%%
% ALGORITHM 

\begin{algorithm} [t]

  \caption{{\sf Create-Mapped}$(M)$, obtains a set of edges $M'$ for the mapped graph $G'$ corresponding to edges in $M \backslash M_{opt}$.}
  \label{Alg:APPROX_SPA-ST_create_cloned}
	\begin{algorithmic}[1]
          \Require An instance $I$ of \acrshort{spa-st}, a stable matching $M$ and maximum stable matching $M_{opt}$ of $I$ and a mapped version $M'_{opt}$ of $M_{opt}$.
          \Ensure Return a mapped version $M'$ of $M \backslash M_{opt}$.           
          \State $M_0 \gets M \backslash M_{opt}$ \Comment where $M_0$ is the working set of student-project pairs in $M$ \label{alg-line-spa-st-create-mapped-m0}
          \State Let $L_0$ be a copy of the set of all lecturer clones
          \State $L'_0 \gets \emptyset$
          \State $M' \gets \emptyset$
          \While {$L_0$ is non-empty} \label{Alg:APPROX_SPA-ST_create_cloned_while1}
          \State Remove a lecturer clone $l_k^r$ from $L_0$
          \If {$(s_i,l_k^r)$ is an edge in $M'_{opt}$ for some $s_i$}
          \State Let $p_j$ be the project assigned to $s_i$ in $M_{opt}$
          \If {$(s_{i'},p_j)$ is in $M_0$ for some student $s_{i'}$}
          \State $M' \gets M' \cup \{(s_i',l_k^r)\}$
          \State $M_0 \gets M_0 \backslash \{(s_i',p_j)\}$

          \Else
          \State $L'_0 \gets L'_0 \cup \{l_k^r\}$

          \EndIf

          \Else 
          \State $L'_0 \gets L'_0 \cup \{l_k^r\}$ 
          \EndIf
          \EndWhile

          \While {$M_0$ is non-empty}  \label{Alg:APPROX_SPA-ST_create_cloned_while2}
          \State Pick some $(s_i,p_j) \in M_0$
          \State $M_0 \gets M_0 \backslash \{(s_i,p_j)\}$
          \State Let $l_k^r$ be some lecturer clone in $L'_0$, where $l_k$ offers $p_j$ \Comment $l_k^r$ must exist since there are $d_k - |M(l_k) \cap M_{opt}(l_k)|$ clones for $l_k$
          \State $L'_0 \gets L'_0 \backslash \{l_k^r\}$
   		  \State $M' \gets M' \cup \{(s_i,l_k^r)\}$
          \EndWhile 

          \State \Return $M'$
	\end{algorithmic}
	\label{SPA_ST_approx_cloned_M}
\end{algorithm}
% ALGORITHM END

%%%%%%%%%%%%%%%%%%%%%%%%%%%%%%%%%%%%%%%%%%%%%%%%%%%%%%%%%
%%%%%%%%%%%%%%%%%%%%%%%%%%%%%%%%%%%%%%%%%%%%%%%%%%%%%%%%%
\subsubsection{Example mapped graph}
\label{spa-st-proofs-sec-32-example}
In this section we introduce an example to demonstrate the creation of mapped graph $G'$ from underlying graph $G$ and matchings $M$ and $M_{opt}$. Figure \ref{fig:spast_example_demo_3_2_proof_cloning} shows example instance $I_4$ of \acrshort{spa-st}.

\begin{figure}
\begin{minipage}[t]{0.3\textwidth}
Student preferences:\\
$s_1$: $p_1$ $p_2$ $p_3$\\
$s_2$: ($p_1$ $p_2$)\\
$s_3$: ($p_3$ $p_2$)\\
$s_4$: $p_3$\\
\end{minipage}
\begin{minipage}[t]{0.3\textwidth}
Project details:\\
$p_1$: lecturer $l_1$, $c_1=2$\\
$p_2$: lecturer $l_2$, $c_2=1$\\
$p_3$: lecturer $l_2$, $c_3=1$\\
\end{minipage}
\begin{minipage}[t]{0.35\textwidth}
Lecturer preferences:\\
\begin{minipage}[t]{0.6\textwidth}
$l_1$: $s_1$ $s_2$\\
$l_2$: ($s_2$ $s_3$) $s_4$ $s_1$\\
\end{minipage}
\begin{minipage}[t]{0.25\textwidth}
$d_1=2$\\
$d_2=2$\\
\end{minipage}
\end{minipage}
\caption[{\sc spa-st} instance $I_4$.]{\acrshort{spa-st} instance $I_4$.}
\label{fig:spast_example_demo_3_2_proof_cloning}
\end{figure}

Let $M = \{(s_1,p_1),(s_2,p_2),(s_3,p_3)\}$ and $M_{opt} = \{(s_1,p_1),(s_2,p_1),(s_3,p_2),(s_4,p_3)\}$ be stable matchings in $I$. Clearly, $M_{opt}$ is also a maximum stable matching as all students are assigned. Figure \ref{SPA_approx_example1_G} shows the underlying graph $G$ of instance $I_4$. To create the vertices of $G'$, student vertices are copied, and multiple lecturer cloned vertices are created. For lecturer vertices $l_1$ and $l_2$ in $G$ with capacities of $2$, we create $l_1^1, l_1^2, l_2^1$ and $l_2^2$ in $G'$.  Using the definition of $M_{opt}'$ above, we obtain the edge set $M_{opt}'=\{(s_2,l_1^1),(s_3,l_2^1),(s_4,l_2^2)\}$.
Figure \ref{SPA_approx_example1_G_dash_mopt} shows a part built $G'$ with all $M'_{opt}$ edges added.

Next $M'$ is calculated using Algorithm \ref{SPA_ST_approx_cloned_M}. A copy of $M \backslash M_{opt}$ is created and denoted $M_0 = \{(s_2,p_2),(s_3,p_3)\}$. $L_0=\{l_1^1, l_1^2, l_2^1, l_2^2\}$ is a copy of the set of all lecturer cloned vertices, and $L_0'$ is the empty set. We iterate through $L_0$ as follows. 

\begin{itemize}
	\item Lecturer clone $l_1^1$ is removed from $L_0$. Since there is an edge $(s_2,l_1^1) \in M'_{opt}$ and $s_2$ is assigned $p_1$ in $M_{opt}$, but $(s_{i'},p_j) \notin M_0$ for each student $s_{i'}$, $l_1^1$ is added to $L_0'$;
	\item Lecturer clone $l_1^2$ is removed from $L_0$. As there is no edge $(s_i,l_1^2) \in M'_{opt}$ for any student $s_i$, $l_1^2$ is added to $L_0'$; 
	\item Next lecturer clone $l_2^1$ is removed from $L_0$. There is an edge $(s_3,l_2^1) \in M'_{opt}$, $s_3$ is assigned $p_2$ in $M_{opt}$ and there is an edge $(s_2,p_2) \in M_0$, hence $(s_2,l_2^1)$ is added to $M'$ and $(s_2,p_2)$ is removed from $M_0$;
	\item Using the same reasoning when the final lecturer clone $l_2^2$ is removed from $L_0$, $(s_3,l_2^2)$ is also added to $M'$ and $(s_3,p_3)$ is removed from $M_0$.
\end{itemize}

   As $M_0$ is now empty, we do not enter the final while loop on Line \ref{Alg:APPROX_SPA-ST_create_cloned_while2} of Algorithm \ref{Alg:APPROX_SPA-ST_create_cloned} therefore $M'$ is now complete. Figure \ref{SPA_approx_example1_G_dash_mopt_m} shows the completed mapped graph $G'$ with edge set $M' \cup M'_{opt}$. 

\begin{figure}[]
\centering
  \begin{subfigure}[t]{0.3\textwidth}
  \centering
        \begin{tikzpicture}
		\node [simpleNodes](st1) at (0,0) {$s_1$};
		\node [simpleNodes](st2) at (0,-\graphyStd) {$s_2$};
		\node [simpleNodes](st3) at (0,-\graphyStd * 2) {$s_3$};
		\node [simpleNodes](st4) at (0,-\graphyStd * 3) {$s_4$};
		\node [simpleNodes](pr1) at (\graphxStd,-\graphyStd * 0.5) {$p_1$};
		\node [simpleNodes](pr2) at (\graphxStd,-\graphyStd * 1.5) {$p_2$};
		\node [simpleNodes](pr3) at (\graphxStd,-\graphyStd * 2.5) {$p_3$};
		\node [simpleNodes](le1) at (\graphxStd * 2,-\graphyStd * 0.5) {$l_1$};
		\node [simpleNodes](le2) at (\graphxStd * 2,-\graphyStd * 2) {$l_2$};

		\draw[thinLine] ([yshift=\graphShiftStdL ]st1.east) -- ([yshift=\graphShiftStdXL ]pr1.west);
		\draw[thickLine] ([yshift=-\graphShiftStdL ]st1.east) -- (pr1.west);
		\draw[thinLine] (st2.east) -- (pr1.west);
		\draw[thinLine] (st3.east) -- (pr2.west);
		\draw[thinLine] (st4.east) -- (pr3.west);
		\draw[thickLine] (st2.east) -- (pr2.west);
		\draw[thickLine] (st3.east) -- (pr3.west);
		\draw[thinLine] ([yshift=\graphShiftStdL ]pr2.east) -- ([yshift=\graphShiftStdXL ]le2.west);
		\draw[thickLine] ([yshift=-\graphShiftStdL ]pr2.east) -- (le2.west);
		\draw[thinLine] ([yshift=\graphShiftStd ] pr1.east) -- ([yshift=\graphShiftStd ] le1.west);
		\draw[thickLine] ([yshift=-\graphShiftStd ] pr1.east) -- ([yshift=-\graphShiftStd ] le1.west);
		\draw[thinLine] ([yshift=\graphShiftStdL ]pr3.east) -- (le2.west);
		\draw[thickLine] ([yshift=-\graphShiftStdL ] pr3.east) -- ([yshift=-\graphShiftStdXL ]le2.west);
		\draw[thinLine] ([yshift=\graphShiftStd * 2 ] pr1.east) -- ([yshift=\graphShiftStd * 2 ] le1.west);

		\draw[dashed] (st1.east) -- (pr2.west);
		\draw[dashed] (st1.east) -- (pr3.west);
		\end{tikzpicture}
    \caption{The underlying graph $G$. $M$ and $M_{opt}$ are shown in bold and non-bold edges respectively. Edges not in $M \cup M_{opt}$ are dashed.}
    \label{SPA_approx_example1_G}
  \end{subfigure}
  \hspace*{\fill}
    \begin{subfigure}[t]{0.3\textwidth}
  \centering
        \begin{tikzpicture}
        \node [simpleNodes](st1) at (0,0) {$s_1$};
		\node [simpleNodes](st2) at (0,-\graphyStd) {$s_2$};
		\node [simpleNodes](st3) at (0,-\graphyStd * 2) {$s_3$};
		\node [simpleNodes](st4) at (0,-\graphyStd * 3) {$s_4$};
		\node [simpleNodes](le1) at (\graphxStd * 1.5,0) {$l_1^1$};
		\node [simpleNodes](le2) at (\graphxStd * 1.5,-\graphyStd) {$l_1^2$};
		\node [simpleNodes](le3) at (\graphxStd * 1.5,-\graphyStd* 2) {$l_2^1$};
		\node [simpleNodes](le4) at (\graphxStd * 1.5,-\graphyStd * 3) {$l_2^2$};

		\draw[thinLine] (st2.east) -- (le1.west);
		\draw[thinLine] (st3.east) -- (le3.west);
		\draw[thinLine] (st4.east) -- (le4.west);

		\end{tikzpicture}
    \caption{Part-built $G'$. All student and lecturer clone vertices are added. $M'_{opt} = \{(s_2,l_1^1),(s_3,l_2^1),(s_4,l_2^2)\}$ edges are also shown.}
    \label{SPA_approx_example1_G_dash_mopt}
  \end{subfigure}
  \hspace*{\fill}
  \begin{subfigure}[t]{0.3\textwidth}
  \centering
        \begin{tikzpicture}
        \node [simpleNodes](st1) at (0,0) {$s_1$};
		\node [simpleNodes](st2) at (0,-\graphyStd) {$s_2$};
		\node [simpleNodes](st3) at (0,-\graphyStd * 2) {$s_3$};
		\node [simpleNodes](st4) at (0,-\graphyStd * 3) {$s_4$};
		\node [simpleNodes](le1) at (\graphxStd * 1.5,0) {$l_1^1$};
		\node [simpleNodes](le2) at (\graphxStd * 1.5,-\graphyStd) {$l_1^2$};
		\node [simpleNodes](le3) at (\graphxStd * 1.5,-\graphyStd * 2) {$l_2^1$};
		\node [simpleNodes](le4) at (\graphxStd * 1.5,-\graphyStd * 3) {$l_2^2$};

		\draw[thinLine] (st2.east) -- (le1.west);
		\draw[thinLine] (st3.east) -- (le3.west);
		\draw[thinLine] (st4.east) -- (le4.west);

		\draw[thickLine] (st2.east) -- (le3.west);
		\draw[thickLine] (st3.east) -- (le4.west);
		\end{tikzpicture}
    \caption{$G'$ with edge set $M' \cup M'_{opt}$, where $M'_{opt} = \{(s_2,l_1^1),(s_3,l_2^1),(s_4,l_2^2)\}$ (non-bold edges) and $M'=\{(s_2,l_2^1),(s_3,l_2^2) \}$ (bold edges), is shown.}
    \label{SPA_approx_example1_G_dash_mopt_m}
  \end{subfigure}
  \caption[Example illustrating the underlying graph $G$ and mapped graph $G'$ of instance $I_4$.]{Example illustrating the underlying graph $G$ and mapped graph $G'$ of instance $I_4$, relative to two stable matchings $M$ and $M_{opt}$ in $G$.}
  \label{}
\end{figure}

%%%%%%%%%%%%%%%%%%%%%%%%%%%%%%%%%%%%%%%%%%%%%%%%%%%%%%%%%
%%%%%%%%%%%%%%%%%%%%%%%%%%%%%%%%%%%%%%%%%%%%%%%%%%%%%%%%%
\subsubsection{Components in $G'$}
\label{spa-st-proofs-sec-32-components}

In this section we define the possible structures that may exist in the mapped graph $G'$.

An \emph{alternating path} in $G'$ is defined as a path that comprises edges in $M_{opt}$ and in $M$ alternately. A path or alternating path is described as $even$ if there are an even number of edges in the path, \emph{odd} otherwise. Finally, an \emph{alternating cycle} is a sequence of edges in $M_{opt}$ and $M$ alternately, which forms a cycle. 

A \emph{component} $c$ in $G'$ is defined as any maximal connected subgraph in $G'$. Figure \ref{SPASTpossiblecomps} shows the possible component structures that may be found in $G'$ which are described in more detail below. Let $n_{c,l}$ and $n_{c,s}$ denote the maximum number of lecturer clone vertices and student vertices respectively, in some component $c$ of $G'$, and let $n_c = \max \{n_{c,l},n_{c,s}\}$. Notation for a lecturer clone in component $c$ is defined as $l^{c,r}$ indicating the $r$\textsuperscript{th} lecturer clone of component $c$. Similarly, $s^{c,r}$ indicates the $r$\textsuperscript{th} student of component $c$.

\begin{figure}[]
\centering
\captionsetup[subfigure]{justification=centering}
  \begin{subfigure}[b]{0.3\textwidth}
  \centering
        \begin{tikzpicture}
          \node [simpleNodesBigger](st1) at (0,0) {$s^{c,1}$};
          \node [simpleNodesBigger](st2) at (0,-\graphyStd) {$s^{c,2}$};
          \node [simpleNodesBigger](st3) at (0,-\graphyStd * 2) {$s^{c,3}$};
          \node [simpleNodesBigger](st4) at (0,-\graphyStd * 4) {$s^{c,\gamma}$};
          \node [simpleNodesBigger](le1) at (\graphxStd * 1.5,0) {$l^{c,1}$};
          \node [simpleNodesBigger](le2) at (\graphxStd * 1.5,-\graphyStd) {$l^{c,2}$};
          \node [simpleNodesBigger](le3) at (\graphxStd * 1.5,-\graphyStd *2) {$l^{c,3}$};
          \node [simpleNodesBigger](le4) at (\graphxStd * 1.5,-\graphyStd * 4) {$l^{c,\gamma}$};

          \draw[dashed] (0,-\graphyStd * 2.75) -- (0,-\graphyStd * 3.25);
          \draw[dashed] (\graphxStd * 1.5,-\graphyStd * 2.75) -- (\graphxStd * 1.5,-\graphyStd * 3.25);
          \draw[thinLine] (st1.east) -- (le1.west);
          \draw[thickLine] (st1.east) -- (le2.west);
          \draw[thinLine] (st2.east) -- (le2.west);
          \draw[thickLine] (st2.east) -- (le3.west);
          \draw[thinLine] (st3.east) -- (le3.west);
          \draw[thickLine] (st4.east) -- (le1.west);
          \draw[thinLine] (st4.east) -- (le4.west);

          \draw[thickLine] (st3.east) -- (\graphyStd *0.75,-\graphyStd * 2.5);
          \draw[thickLine] (\graphyStd *0.75, -\graphyStd *3.5) -- (le4.west);

          \end{tikzpicture}
    \caption{}
    \label{SPASTpossiblecomps_1}
  \end{subfigure}
  \hspace*{\fill}
  \begin{subfigure}[b]{0.3\textwidth}
  \centering
        \begin{tikzpicture}
          \node [simpleNodesBigger](st1) at (0,0) {$s^{c,1}$};
          \node [simpleNodesBigger](st2) at (0,-\graphyStd) {$s^{c,2}$};
          \node [simpleNodesBigger](st3) at (0,-\graphyStd * 3) {$s^{c,\mu}$};
          \node [simpleNodesBigger](le1) at (\graphxStd * 1.5,0) {$l^{c,1}$};
          \node [simpleNodesBigger](le2) at (\graphxStd * 1.5,-\graphyStd) {$l^{c,2}$};
          \node [simpleNodesBigger](le3) at (\graphxStd * 1.5,-\graphyStd *3) {$l^{c,\mu}$};
          \node [simpleNodesBigger](le4) at (\graphxStd * 1.5,-\graphyStd * 4) {$l^{c,\gamma}$};

          \draw[dashed] (0,-\graphyStd * 1.75) -- (0,-\graphyStd * 2.25);
          \draw[dashed] (\graphxStd * 1.5,-\graphyStd * 1.75) -- (\graphxStd * 1.5,-\graphyStd * 2.25);
          \draw[thinLine] (st1.east) -- (le1.west);
          \draw[thickLine] (st1.east) -- (le2.west);
          \draw[thinLine] (st2.east) -- (le2.west);
          \draw[thinLine] (st3.east) -- (le3.west);
          \draw[thickLine] (st3.east) -- (le4.west);
          \draw[thickLine] (st2.east) -- (\graphyStd *0.75,-\graphyStd * 1.5);
          \draw[thickLine] (\graphyStd *0.75, -\graphyStd *2.5) -- (le3.west);

          \end{tikzpicture}
    \caption{}
    \label{SPASTpossiblecomps_2}
  \end{subfigure}
  \hspace*{\fill}
  \begin{subfigure}[b]{0.3\textwidth}
  \centering
        \begin{tikzpicture}
          \node [simpleNodesBigger](st1) at (0,0) {$s^{c,1}$};
          \node [simpleNodesBigger](st2) at (0,-\graphyStd) {$s^{c,2}$};
          \node [simpleNodesBigger](st3) at (0,-\graphyStd * 3) {$s^{c,\mu}$};
          \node [simpleNodesBigger](st4) at (0,-\graphyStd * 4) {$s^{c,\gamma}$};
          \node [simpleNodesBigger](le1) at (\graphxStd * 1.5,0) {$l^{c,1}$};
          \node [simpleNodesBigger](le2) at (\graphxStd * 1.5,-\graphyStd) {$l^{c,2}$};
          \node [simpleNodesBigger](le3) at (\graphxStd * 1.5,-\graphyStd *3) {$l^{c,\mu}$};

          \draw[dashed] (0,-\graphyStd * 1.75) -- (0,-\graphyStd * 2.25);
          \draw[dashed] (\graphxStd * 1.5,-\graphyStd * 1.75) -- (\graphxStd * 1.5,-\graphyStd * 2.25);
          \draw[thinLine] (st1.east) -- (le1.west);
          \draw[thickLine] (st2.east) -- (le1.west);
          \draw[thinLine] (st2.east) -- (le2.west);
          \draw[thinLine] (st3.east) -- (le3.west);
          \draw[thickLine] (st4.east) -- (le3.west);

          \draw[thickLine] (le2.west) -- (\graphyStd *0.75,-\graphyStd * 1.5);
          \draw[thickLine] (\graphyStd *0.75, -\graphyStd *2.5) -- (st3.east);

          \end{tikzpicture}
    \caption{}
    \label{SPASTpossiblecomps_3}
  \end{subfigure}
  \vspace{1cm}

\begin{subfigure}[b]{0.3\textwidth}
  \centering
        \begin{tikzpicture}
          \node [simpleNodesBigger](st1) at (0,0) {$s^{c,1}$};
          \node [simpleNodesBigger](st2) at (0,-\graphyStd) {$s^{c,2}$};
          \node [simpleNodesBigger](st3) at (0,-\graphyStd * 3) {$s^{c,\gamma}$};
          \node [simpleNodesBigger](le1) at (\graphxStd * 1.5,0) {$l^{c,1}$};
          \node [simpleNodesBigger](le2) at (\graphxStd * 1.5,-\graphyStd) {$l^{c,2}$};
          \node [simpleNodesBigger](le3) at (\graphxStd * 1.5,-\graphyStd *3) {$l^{c,\gamma}$};

          \draw[dashed] (0,-\graphyStd * 1.75) -- (0,-\graphyStd * 2.25);
          \draw[dashed] (\graphxStd * 1.5,-\graphyStd * 1.75) -- (\graphxStd * 1.5,-\graphyStd * 2.25);
          \draw[thickLine] (st1.east) -- (le1.west);
          \draw[thinLine] (st1.east) -- (le2.west);
          \draw[thickLine] (st2.east) -- (le2.west);
          \draw[thickLine] (st3.east) -- (le3.west);

          \draw[thinLine] (st2.east) -- (\graphyStd *0.75,-\graphyStd * 1.5);
          \draw[thinLine] (\graphyStd *0.75, -\graphyStd *2.5) -- (le3.west);

          \end{tikzpicture}
    \caption{}
    \label{SPASTpossiblecomps_4}
  \end{subfigure}
  \hspace*{\fill}
  \begin{subfigure}[b]{0.3\textwidth}
  \centering
        \begin{tikzpicture}
          \node [simpleNodesBigger](st1) at (0,0) {$s^{c,1}$};
          \node [simpleNodesBigger](st2) at (0,-\graphyStd) {$s^{c,2}$};
          \node [simpleNodesBigger](st3) at (0,-\graphyStd * 3) {$s^{c,\gamma}$};
          \node [simpleNodesBigger](le1) at (\graphxStd * 1.5,0) {$l^{c,1}$};
          \node [simpleNodesBigger](le2) at (\graphxStd * 1.5,-\graphyStd) {$l^{c,2}$};
          \node [simpleNodesBigger](le3) at (\graphxStd * 1.5,-\graphyStd *3) {$l^{c,\gamma}$};

          \draw[dashed] (0,-\graphyStd * 1.75) -- (0,-\graphyStd * 2.25);
          \draw[dashed] (\graphxStd * 1.5,-\graphyStd * 1.75) -- (\graphxStd * 1.5,-\graphyStd * 2.25);
          \draw[thinLine] (st1.east) -- (le1.west);
          \draw[thickLine] (st1.east) -- (le2.west);
          \draw[thinLine] (st2.east) -- (le2.west);
          \draw[thinLine] (st3.east) -- (le3.west);

          \draw[thickLine] (st2.east) -- (\graphyStd *0.75,-\graphyStd * 1.5);
          \draw[thickLine] (\graphyStd *0.75, -\graphyStd *2.5) -- (le3.west);

          \end{tikzpicture}
    \caption{$\gamma \geq 3$}
    \label{SPASTpossiblecomps_5}
  \end{subfigure}
  \hspace*{\fill}
  \begin{subfigure}[b]{0.3\textwidth}
  \centering
        \begin{tikzpicture}
          \node [simpleNodesBigger](st1) at (0,0) {$s^{c,1}$};
          \node [simpleNodesBigger](st2) at (0,-\graphyStd) {$s^{c,2}$};

          \node [simpleNodesBigger](le1) at (\graphxStd * 1.5,0) {$l^{c,1}$};
          \node [simpleNodesBigger](le2) at (\graphxStd * 1.5,-\graphyStd) {$l^{c,2}$};
          \draw[thinLine] (st1.east) -- (le1.west);
          \draw[thickLine] (st1.east) -- (le2.west);
          \draw[thinLine] (st2.east) -- (le2.west);

          \end{tikzpicture}
    \caption{$\gamma = 2$}
    \label{SPASTpossiblecomps_6}
  \end{subfigure}

  \vspace{1cm}
  \begin{subfigure}[b]{0.3\textwidth}
  \centering
        \begin{tikzpicture}
          \node [simpleNodesBigger](st1) at (0,0) {$s^{c,1}$};

          \node [simpleNodesBigger](le1) at (\graphxStd * 1.5,0) {$l^{c,1}$};

          \draw[thinLine] (st1.east) -- (le1.west);
          \end{tikzpicture}
    \caption{$\gamma = 1$}
    \label{SPASTpossiblecomps_7}
  \end{subfigure}
  
  \caption[Possible component structures in $G'$.]{Possible component structures in $G'$ for a component $c$, where $\gamma = n_c$, the size of the component, and $\mu = \gamma - 1$. $M'$ and $M'_{opt}$ edges are shown in bold and non-bold, respectively.}
  \label{SPASTpossiblecomps}
\end{figure}

Each vertex in $G'$ is incident to at most one $M'$ edge and at most one $M'_{opt}$ edge, meaning every component must be a path or cycle comprising alternating $M'$ and $M'_{opt}$ edges. Therefore the structure of each component must have one of the following forms.

\begin{enumerate}[(a)]
  \item An alternating cycle; \label{SPASTpossiblecomps_1desc}
  \item An even length alternating path, with lecturer clone end vertices; \label{SPASTpossiblecomps_2desc}
  \item An even length alternating path, with student end vertices; \label{SPASTpossiblecomps_3desc}
  \item An odd length alternating path, with end edges in $M'$; \label{SPASTpossiblecomps_4desc}
  \item An odd length alternating path, with end edges in $M'_{opt}$, for $n_c \geq 3$; \label{SPASTpossiblecomps_5desc}
  \item An odd length alternating path, with end edges in $M'_{opt}$, for $n_c = 2$; \label{SPASTpossiblecomps_6desc}
  \item An odd length alternating path, with end edges in $M'_{opt}$, for $n_c = 1$; \label{SPASTpossiblecomps_7desc}
\end{enumerate}

We wish to show that any stable matching found by Algorithm \ref{Alg:APPROX_SPA-S_stable} must be at least two-thirds of the size of $M_{opt}$.

%%%%%%%%%%%%%%%%%%%%%%%%%%%%%%%%%%%%%%%%%%%%%%%%%%%%%%%%%
%%%%%%%%%%%%%%%%%%%%%%%%%%%%%%%%%%%%%%%%%%%%%%%%%%%%%%%%%
\subsubsection{Proof of the $\frac{3}{2}$ performance guarantee}
\label{spa-st-proofs-sec-32-proofs32}

In this section we prove that any stable matching produced by Algorithm {\sf Max-SPA-ST-Approx} must be at least two-thirds of the size of a maximum stable matching.

Propositions \ref{th:stable_SPA_S_alt_p_j_undersubscribed} and \ref{th:stable_SPA_S_alt_p_j_undersubscribed_mopt} detail two configurations of components in $G'$ where we may infer that a project is undersubscribed in either $M$ or $M_{opt}$. The following terminology, used in these first two propositions, is now introduced.
Let $M_{opt}'(p_j)$ denote the set of students who are paired with lecturer clones in $M_{opt}'$ associated with project $p_j$. Then $M_{opt}'(p_j) = M_{opt}(p_j) \backslash M(p_j)$ by construction. Finally, let $M_0^\alpha$ denote the value of $M_0=M\setminus M_{opt}$ on Line \ref{alg-line-spa-st-create-mapped-m0} of Algorithm \ref{Alg:APPROX_SPA-ST_create_cloned}, i.e. the original set of unmapped student-project pairs.

% THEOREM BEGIN
\begin{prop} 
	\label{th:stable_SPA_S_alt_p_j_undersubscribed}
	 Let $(s_i,l^{c,r}) \in M'_{opt}$ be an edge in $G'$ where $(s_i,p_j) \in M_{opt}$. If $l^{c,r}$ is unassigned in $M'$ or if there exists an edge $(s_{i'},l^{c,r}) \in M'$ where $s_{i'}$ is assigned to a project other than $p_j$ in $M$, then $|M_{opt}(p_j)| > |M(p_j)|$, and hence $p_j$ is undersubscribed in $M$. %{\color{blue} Further $|M(p_{j'})| > |M_{opt}(p_{j'})|$, and hence $p_{j'}$ is undersubscribed in $M_{opt}$.}
	\end{prop}
% THEOREM END

\begin{proof}
Suppose $l^{c,r}$ is unassigned in $M'$ or $(s_{i'},l^{c,r}) \in M'$ where $s_{i'}$ is assigned a project other than $p_j$ in $M$.
Assume for contradiction that $|M_{opt}(p_j)| \leq |M(p_j)|$. During the execution of Algorithm \ref{Alg:APPROX_SPA-ST_create_cloned}, the first while loop iterates over the lecturer clones in $G'$ once. This means that all edges in $M_{opt}'(p_j)$ are iterated over. 

Since $|M_{opt}(p_j)| \leq |M(p_j)|$, we know that $|M_{opt}(p_j) \backslash M(p_j)| \leq |M(p_j) \backslash M_{opt}(p_j)|$. But then $|M_{opt}'(p_j)| \leq |M_0^{\alpha}(p_j)|$ and so it must be the case that every lecturer clone in a pair of $M'_{opt}(p_j)$ (including $l^{c,r}$) is paired with a student in $M_0^{\alpha}(p_j)$. This contradicts the fact that $l^{c,r}$ is either unassigned in $M'$ or $(s_{i'},l^{c,r}) \in M'$ where $s_{i'}$ is assigned to a project other than $p_j$ in $M$. It follows immediately that $p_j$ is undersubscribed in $M$.
\end{proof}

% THEOREM BEGIN
\begin{prop} 
	\label{th:stable_SPA_S_alt_p_j_undersubscribed_mopt}
	 Let $(s_i,l^{c,r}) \in M'$ be an edge in $G'$ where $(s_i,p_j) \in M$. If there exists an edge $(s_{i'},l^{c,r}) \in M_{opt}'$ where $s_{i'}$ is assigned to a project other than $p_j$ in $M_{opt}$, then $|M(p_j)| > |M_{opt}(p_j)|$, and hence $p_j$ is undersubscribed in $M_{opt}$.
	\end{prop}
% THEOREM END

\begin{proof}
We use a similar proof to Proposition \ref{th:stable_SPA_S_alt_p_j_undersubscribed}. Suppose there exists an edge $(s_{i'},l^{c,r}) \in M_{opt}'$ where $s_{i'}$ is assigned a project other than $p_j$ in $M_{opt}$ and assume for contradiction that $|M(p_j)| \leq |M_{opt}(p_j)|$. During the execution of Algorithm \ref{Alg:APPROX_SPA-ST_create_cloned}, the first while loop iterates over the lecturer clones in $G'$ once, meaning that all edges in $M_{opt}'(p_j)$ are iterated over. 

As $|M(p_j)| \leq |M_{opt}(p_j)|$, we know that $|M(p_j) \backslash M_{opt}(p_j)| \leq |M_{opt}(p_j) \backslash M(p_j)|$. But then $|M_0^{\alpha}(p_j)| \leq |M_{opt}'(p_j)|$ and so each student in $M_0^{\alpha}(p_j)$ is paired with a lecturer clone that exists in a pair of $M_{opt}'(p_j)$. But this contradicts the fact that both $(s_i,l^{c,r}) \in M'$ with $s_i \in M_0^{\alpha}(p_j)$ and $(s_i',l^{c,r}) \in M'_{opt}$ with $s_i \notin M_{opt}'(p_j)$. Hence $|M(p_j)| > |M_{opt}(p_j)|$ and so $p_j$ is undersubscribed in $M_{opt}$.
\end{proof}

We now give three proofs of preliminary results that are used to aid Lemma \ref{th:stable_SPA_S_lemma_no_alt_size_3} and Lemma \ref{th:stable_SPA_S_lemma_no_alt_size_1} in showing that it is not possible for a component of the types shown in Figure \ref{SPASTpossiblecomps_6} or in Figure \ref{SPASTpossiblecomps_7} to exist in $G'$.

First, Proposition \ref{th:stable_SPA_S_alt_3_p_j_never_applied} shows that for a component of the type shown in Figure \ref{SPASTpossiblecomps_6}, neither student $s^{c,1}$ nor $s^{c,2}$ can have applied to the project $s^{c,1}$ is assigned to in $M_{opt}$, during Algorithm \ref{Alg:APPROX_SPA-S_stable}'s execution.

% THEOREM BEGIN
\begin{prop} 
	\label{th:stable_SPA_S_alt_3_p_j_never_applied}
	 Let $c$ be the component of $G'$ in Figure \ref{SPASTpossiblecomps_6}. Let $s^{c,1}$ be assigned to project $p_j$ in $M_{opt}$. Then project $p_j$ is fully available in $M$ and, $s^{c,1}$ and $s^{c,2}$ can never have applied to $p_j$ at any point in Algorithm \ref{Alg:APPROX_SPA-S_stable}'s execution.
	\end{prop}
% THEOREM END

\begin{proof}
Let lecturer clone $l^{c,1}$ correspond to lecturer $l_k$. In $G'$, lecturer clone $l^{c,1}$ is unassigned in $M'$, therefore we know that $l_k$ is undersubscribed in $M$. Since there is no edge in $M'$ incident to $l^{c,1}$, by Proposition \ref{th:stable_SPA_S_alt_p_j_undersubscribed}, we know $p_j$ is undersubscribed in $M$. Project $p_j$ is, by definition, fully available in $M$ and, by Proposition \ref{th:stable_SPA_S_propfullyavailable}, must have been fully available throughout the algorithm's execution.

Now we prove that neither $s^{c,1}$ nor $s^{c,2}$ can have applied to $p_j$. In Algorithm \ref{Alg:APPROX_SPA-S_stable_3bibpBrief}, students can only apply to projects that are not fully available by Proposition \ref{prop-spa-st-prom_subsequentprecarious}, hence we only look at the main while loop of Algorithm \ref{Alg:APPROX_SPA-S_stable}.
We consider $s^{c,1}$ first. Assume for contradiction that $s^{c,1}$ applied to $p_j$ at some point during the main while loop of Algorithm \ref{Alg:APPROX_SPA-S_stable}'s execution. Then $(s^{c,1},p_j)$ would be added to $M$ as $p_j$ was always fully available. But we know that $(s^{c,1},p_j)$ is not in the final matching $M$ hence it must have been rejected by $l_k$ at some point. But this can only have happened if $p_j$ was not fully available, which contradicts the fact that $p_j$ is always fully available above. Therefore $s^{c,1}$ can never have applied to $p_j$ at any point. By identical reasoning $s^{c,2}$ can also never have applied to $p_j$.
\end{proof}

Next, in Proposition \ref{th:stable_SPA_S_phase2_not_fa_not_prec}, we prove a more general result that if a student applies to a project while they are in phase $2$, then that project is neither fully available nor precarious.

% THEOREM BEGIN
\begin{prop} 
	\label{th:stable_SPA_S_phase2_not_fa_not_prec}
	 Suppose that student $s_i$ applied to project $p_j$ in phase $2$ of Algorithm \ref{Alg:APPROX_SPA-S_stable} and denote this time by $T_0$. Then at time $T_0$, $p_j$ is not fully available and is non-precarious.
	\end{prop}
% THEOREM END
\begin{proof}
Let $l_k$ be the lecturer offering $p_j$. Assume for contradiction that $p_j$ is fully available at $T_0$. Since $s_i$ is applying to $p_j$ in phase $2$, $l_k$ must have rejected $s_i$ when $s_i$ was in phase $1$. But this can only happen if $p_j$ is not fully available and by Proposition \ref{th:stable_SPA_S_propfullyavailable}, $p_j$ cannot again become fully available.

Assume then that $p_j$ is precarious at $T_0$. Then there must exist a precarious pair $(s_{i'},p_j)$ in the matching for some student $s_{i'}$. We know from Proposition \ref{prop-spa-st-subsequentprecarious} that when a project is not fully available and non-precarious, it cannot again become precarious. Therefore, when $s_i$ applied in phase $1$ to $p_j$, it was either fully available or precarious (or both), and so $(s_i,p_j)$ must have been added to the matching. But at some point before $T_0$, since $s_i$ is applying in phase $2$, $(s_i,p_j)$ was removed from the matching. This can only happen when $p_j$ is not fully available and either $(s_i,p_j)$ is precarious or is a worst student in $M(p_j)$ (also a worst assignee of $M(l_k)$).

If $(s_i,p_j)$ was precarious then, once removed, $s_i$ would again apply to $p_j$ in phase $1$ and must be successfully added for the same reason as before, although this time as a non-precarious pair (since other fully available projects tied with $p_j$ on $s_i$'s list would be applied to by $s_i$ before $p_j$). The removal of non-precarious $(s_i,p_j)$ can only happen because $p_j$ is non-precarious which contradicts the assumption that $p_j$ is precarious at $T_0$ by Proposition \ref{prop-spa-st-subsequentprecarious}, since $p_j$ is also not fully available. Therefore $p_j$ is non-precarious at $T_0$. 

Therefore $p_j$ can be neither fully available nor precarious at $T_0$.
\end{proof}

% {\color{blue}The following proposition extends Case 2 of Proposition \ref{th:stable_SPA_S_lecfullworseprec}.}

Finally, Proposition \ref{th:stable_SPA_S_pfullthennoworseafterwhile}, shows that if a project is full and non-precarious at some point before the end of the main while loop of Algorithm \ref{Alg:APPROX_SPA-S_stable}, then this project cannot subsequently accept a worse student (according to the lecturer who offers it).

% THEOREM BEGIN
\begin{prop} 
	\label{th:stable_SPA_S_pfullthennoworseafterwhile}
	 Let $T_{end}$ denote the point in Algorithm \ref{Alg:APPROX_SPA-S_stable}'s execution at the end of the main while loop. If a project $p_j$ offered by $l_k$ is full and non-precarious before $T_{end}$, then a student $s_i$ worse than $l_k$'s worst ranked assignees in $M(p_j)$ cannot subsequently become assigned to $p_j$.
	\end{prop}
% THEOREM END
\begin{proof}

As in Proposition \ref{th:stable_SPA_S_lecfullworseprec}, let $T_0$ be a point of the algorithm's execution, mentioned in the statement of the propostion, where $p_j$ is full and non-precarious. By Proposition \ref{th:stable_SPA_S_lecfullworseprec}, we know that after $T_0$ a student $s_i$ worse than $l_k$'s worst ranked assignees in $M(p_j)$ cannot subsequently become assigned to $p_j$ before $T_{end}$. Hence we concentrate only on changes made by Algorithm \ref{Alg:APPROX_SPA-S_stable_3bibpBrief}. Additionally, due to this same result, it suffices to show that a student $s_i$ worse than $l_k$'s worst ranked assignees in $M(p_j)$ at $T_x$ (for $T_x$ in the range $T_0$ to $T_{end}$) cannot subsequently become assigned to $p_j$. 

Project $p_j$ is either full or undersubscribed at $T_{end}$. We deal with each case in turn.

	\begin{itemize}
	\item Assume first, that $p_j$ is full at $T_{end}$. If $p_j$ remains full then no student can become assigned to $p_j$. Therefore assume that $p_j$ becomes undersubscribed after $T_{end}$ and let $T_1$ be the first time this occurs. It must be that, just before $T_1$, a blocking pair $(s_{i'}, p_j)$ exists, for some student $s_{i'}$. By Propositions \ref{th:stable_SPA_S_3biworstinmlk} and \ref{th:stable_SPA_S_3bicreatedworstinmlk}, $s_{i'} \in M(p_j)$ is one of the worst students in $M(l_k)$. Since $T_1$ is the first time $p_j$ becomes undersubscribed, no students have been removed from $M(p_j)$ since $T_{end}$, and so $s_{i'}$ must have existed in $M(p_j)$ at $T_{end}$ as well. Since no students are introduced in Algorithm \ref{Alg:APPROX_SPA-S_stable_3bibpBrief}, $s_{i'} \in M(p_j)$ is also one of the worst students in $M(l_k)$ at $T_{end}$. Therefore, by Proposition \ref{th:stable_SPA_S_lecfullworseprec}, $p_j$ cannot subsequently be assigned a worse student than exists in $M(p_j)$ at $T_{end}$.

	\item Assume now that $p_j$ is undersubscribed at $T_{end}$. Then there was some point $T_2$, before $T_{end}$ and after $T_0$, that $p_j$ became undersubscribed by the removal from $M$ of $(s_{i'},p_j)$ for some student $s_{i'}$. We know that $(s_{i'},p_j)$ is non-precarious since $p_j$ is non-precarious after $T_0$ by Proposition \ref{prop-spa-st-subsequentprecarious}. The removal of a pair $(s_{i'},p_j)$ at $T_2$ can only have happened if $l_k$ was full, non-precarious (since we are removing a non-precarious pair) and $l_k$ meta-preferred the student they are adding, to $s_{i'}$. But then $s_{i'} \in M(p_j)$ was one of the worst students in $M(l_k)$ just before $T_2$, and so, by Proposition \ref{th:stable_SPA_S_lecfullworseprec}, $p_j$ cannot subsequently be assigned a worse student.
	
	\end{itemize}

Therefore, a student $s_i$ worse than $l_k$'s worst ranked assignees in $M(p_j)$ at $T_0$ cannot subsequently become assigned to $p_j$.
\end{proof}

In Lemma \ref{th:stable_SPA_S_lemma_no_alt_size_3} we prove that it is not possible for a component of the type shown in Figure \ref{SPASTpossiblecomps_6} to exist in $G'$.

% THEOREM BEGIN

\begin{lemma} 
	\label{th:stable_SPA_S_lemma_no_alt_size_3}
	Let $M$ be a stable matching found by Algorithm \ref{Alg:APPROX_SPA-S_stable} for instance $I$ of \acrshort{spa-st}, and let $M_{opt}$ be a maximum stable matching in $I$. No component of the type given in Figure \ref{SPASTpossiblecomps_6} can exist in the mapped graph $G'$. 
	\end{lemma}
% THEOREM END

\begin{proof}
Assume for contradiction that there is a component $c$ of the type shown in Figure \ref{SPASTpossiblecomps_6} in $G'$. Let $p_j$ be the project assigned to $s^{c,1}$ in $M'_{opt}$ in Figure \ref{SPASTpossiblecomps_6}.
%Component $c$ is a  consists of an alternating path of length $3$ with end edges in $M'_{opt}$. 

We look at the possible configurations in $G$ that could map to $c$ in $G'$. Lecturer clones $l^{c,1}$ and $l^{c,2}$ may or may not be the same lecturer in $G$. It may also be the case that $s^{c,1}$ and $s^{c,2}$ are assigned to the same or different projects in $M'$ and $M'_{opt}$, respecting the fact that projects may only be offered by one lecturer. Let $s^{c,1} = s_i$ and $s^{c,2} = s_{i'}$. Figure \ref{alt_size3} shows the possible configurations in $G$ relating to $c$ in $G'$. 
They are found by noting that all configurations in $G$ must have: $2$ students; $1$ or $2$ lecturers; between $2$ and $3$ projects; student $s_{i'}$ must be unassigned in $M$; and $s_{i'}$ must be assigned a project of the same lecturer in $M_{opt}$ as $s_i$ is in $M$.
Note that it is not possible for there to be only one project $p_j$ in the configuration since $s_i$ would be assigned to $p_j$ in both $M$ and $M_{opt}'$, meaning $(s_i,p_j)$ would not exist in $M_{opt}\backslash M$ or $M\backslash M_{opt}$ and so neither of the edges from $s^{c,1}$ would exist in $G'$, a contradiction.

\begin{figure}[]
\centering
\captionsetup[subfigure]{justification=centering}
  \begin{subfigure}[b]{0.3\textwidth}
  \centering
        \begin{tikzpicture}
		\node [simpleNodes](st1) at (0,0) {$s_i$};
		\node [simpleNodes](st2) at (0,-\graphyStd) {$s_{i'}$};
		\node [simpleNodes](pr1) at (\graphxStd,0) {$p_j$};
		\node [simpleNodes](pr2) at (\graphxStd,-\graphyStd) {$p_{j'}$};
		\node [simpleNodes](le1) at (\graphxStd * 2,-\graphyStd * 0.5) {$l_k$};

		\draw[thinLine] (st1.east) -- (pr1.west);
		\draw[thickLine] (st1.east) -- (pr2.west);
		\draw[thinLine] (pr1.east) -- (le1.west);
		\draw[thinLine] (st2.east) -- (pr2.west);

		\draw[thinLine] ([yshift=\graphShiftStdL ]pr2.east) -- (le1.west);
		\draw[thickLine] ([yshift=-\graphShiftStdL ] pr2.east) -- ([yshift=-\graphShiftStdXL ]le1.west);

		\end{tikzpicture}
    \caption{}
    \label{SPA_alt_size3_1}
  \end{subfigure}
  \hspace*{\fill}
    \begin{subfigure}[b]{0.3\textwidth}
  \centering
        \begin{tikzpicture}
		\node [simpleNodes](st1) at (0,0) {$s_i$};
		\node [simpleNodes](st2) at (0,-\graphyStd) {$s_{i'}$};
		\node [simpleNodes](pr1) at (\graphxStd,0) {$p_j$};
		\node [simpleNodes](pr2) at (\graphxStd,-\graphyStd) {$p_{j'}$};
		\node [simpleNodes](le1) at (\graphxStd * 2,-\graphyStd * 0.5) {$l_k$};

		\draw[thinLine] (st1.east) -- (pr1.west);
		\draw[thickLine] (st1.east) -- (pr2.west);

		\draw[thinLine] (st2.east) -- (pr1.west);
		\draw[thickLine] (pr2.east) -- (le1.west);
		\draw[thinLine] ([yshift=\graphShiftStd ] pr1.east) -- ([yshift=\graphShiftStd ] le1.west);
		\draw[thinLine] ([yshift=-\graphShiftStd ] pr1.east) -- ([yshift=-\graphShiftStd ] le1.west);
		
		\end{tikzpicture}
    \caption{}
    \label{SPA_alt_size3_2}
  \end{subfigure}
  \hspace*{\fill}
  \begin{subfigure}[b]{0.3\textwidth}
  \centering
  		\begin{tikzpicture}
		\node [simpleNodes](st1) at (0,-\graphyStd * 0.5) {$s_i$};
		\node [simpleNodes](st2) at (0,-\graphyStd * 1.5) {$s_{i'}$};
		\node [simpleNodes](pr1) at (\graphxStd,0) {$p_j$};
		\node [simpleNodes](pr2) at (\graphxStd,-\graphyStd) {$p_{j'}$};
		\node [simpleNodes](pr3) at (\graphxStd,-\graphyStd * 2) {$p_{j''}$};
		\node [simpleNodes](le1) at (\graphxStd * 2,-\graphyStd) {$l_k$};

		\draw[thinLine] (st1.east) -- (pr1.west);
		\draw[thickLine] (st1.east) -- (pr2.west);
		\draw[thinLine] (pr1.east) -- (le1.west);
		\draw[thinLine] (pr1.east) -- (le1.west);
		\draw[thinLine] (st2.east) -- (pr3.west);
		\draw[thickLine] (pr2.east) -- (le1.west);
		\draw[thinLine] (pr3.east) -- (le1.west);
		\end{tikzpicture}
    \caption{}
    \label{SPA_alt_size3_3}
  \end{subfigure}
  
  \vspace{0.5cm}
  \begin{subfigure}[b]{0.3\textwidth}
  \centering
  		\begin{tikzpicture}
		\node [simpleNodes](st1) at (0,0) {$s_i$};
		\node [simpleNodes](st2) at (0,-\graphyStd) {$s_{i'}$};
		\node [simpleNodes](pr1) at (\graphxStd,0) {$p_j$};
		\node [simpleNodes](pr2) at (\graphxStd,-\graphyStd) {$p_{j'}$};
		\node [simpleNodes](le1) at (\graphxStd * 2,0) {$l_k$};
		\node [simpleNodes](le2) at (\graphxStd * 2,-\graphyStd) {$l_{k'}$};

		\draw[thinLine] (st1.east) -- (pr1.west);
		\draw[thickLine] (st1.east) -- (pr2.west);

		\draw[thinLine] (st2.east) -- (pr2.west);
		\draw[thinLine] (pr1.east) -- (le1.west);
		\draw[thinLine] ([yshift=\graphShiftStd ] pr2.east) -- ([yshift=\graphShiftStd ] le2.west);
		\draw[thickLine] ([yshift=-\graphShiftStd ] pr2.east) -- ([yshift=-\graphShiftStd ] le2.west);
		\end{tikzpicture}
    \caption{}
    \label{SPA_alt_size3_4}
  \end{subfigure}
  \hspace{1cm}
  \begin{subfigure}[b]{0.3\textwidth}
  \centering
  		\begin{tikzpicture}
		\node [simpleNodes](st1) at (0,-\graphyStd * 0.5) {$s_i$};
		\node [simpleNodes](st2) at (0,-\graphyStd * 1.5) {$s_{i'}$};
		\node [simpleNodes](pr1) at (\graphxStd,0) {$p_j$};
		\node [simpleNodes](pr2) at (\graphxStd,-\graphyStd) {$p_{j'}$};
		\node [simpleNodes](pr3) at (\graphxStd,-\graphyStd * 2) {$p_{j''}$};
		\node [simpleNodes](le1) at (\graphxStd * 2,-\graphyStd * 0.5) {$l_k$};
		\node [simpleNodes](le2) at (\graphxStd * 2,-\graphyStd * 1.5) {$l_{k'}$};
		\draw[thinLine] (st2.east) -- (pr3.west);
		\draw[thickLine] (st1.east) -- (pr2.west);
		\draw[thinLine] (pr1.east) -- (le1.west);
		\draw[thinLine] (st1.east) -- (pr1.west);
		\draw[thinLine] (pr3.east) -- (le2.west);
		\draw[thickLine] (pr2.east) -- (le2.west);
		\end{tikzpicture}
    \caption{}
        
    \label{SPA_alt_size3_6}
  \end{subfigure}
  \caption[Possible configurations in $G$ for an alternating path of size 3 in $G'$.]{Possible configurations in $G$ for an alternating path of size 3 in $G'$ with $M'_{opt}$ end edges. $M$ and $M_{opt}$ are shown in bold and non-bold edges respectively. Any project and lecturer vertices shown may have additional assignments involving other vertices not shown in the graphs.}
  \label{alt_size3}
\end{figure}

%%%%%%% used to be part e from the above figure - removed as it is not actually a possible configuration in G
  % \hspace*{\fill}
  %   \begin{subfigure}[b]{0.3\textwidth}
  % \centering
  % 		\begin{tikzpicture}
		% \node [simpleNodes](st1) at (0,0) {$s_i$};
		% \node [simpleNodes](st2) at (0,-\graphyStd) {$s_{i'}$};
		% \node [simpleNodes](pr1) at (\graphxStd,0) {$p_j$};
		% \node [simpleNodes](pr2) at (\graphxStd,-\graphyStd) {$p_{j'}$};
		% \node [simpleNodes](le1) at (\graphxStd * 2,0) {$l_k$};
		% \node [simpleNodes](le2) at (\graphxStd * 2,-\graphyStd) {$l_{k'}$};

		% \draw[thinLine] (st1.east) -- (pr1.west);
		% \draw[thickLine] (st1.east) -- (pr2.west);

		% \draw[thinLine] (st2.east) -- (pr1.west);
		% \draw[thickLine] (pr2.east) -- (le2.west);
		% \draw[thinLine] ([yshift=\graphShiftStd * 0.75 ] pr1.east) -- ([yshift=\graphShiftStd * 0.75 ] le1.west);
		% \draw[thinLine] ([yshift=-\graphShiftStd * 0.75 ] pr1.east) -- ([yshift=-\graphShiftStd * 0.75 ] le1.west);
		% \end{tikzpicture}
  %   \caption{}
  %   \label{SPA_alt_size3_5}
  % \end{subfigure}
  % \hspace*{\fill}

We now show that none of the subgraphs shown in Figure \ref{alt_size3} can occur in a matching $M$ with respect to $G$ found using Algorithm \ref{Alg:APPROX_SPA-S_stable}. Assume for contradiction that one does occur. We consider each type of subgraph separately.

\begin{enumerate}[(a)]

  \item Students $s_i$ and $s_{i'}$ are assigned to $p_j$ and $p_{j'}$ in $M_{opt}$ respectively, and $s_i$ is assigned to $p_{j'}$ in $M$. Lecturer $l_k$ offers both $p_j$ and $p_{j'}$.

  There are three sub-cases to consider. 
  \label{SPA_3_aug_3_2_a}

  \begin{enumerate}[i.]
  	\item \emph{$s_i$ strictly prefers $p_j$ to $p_{j'}$:} If $s_i$ strictly prefers $p_j$ to $M(s_i)$ then $s_i$ must have applied to $p_j$ at least once. But this contradicts Proposition \ref{th:stable_SPA_S_alt_3_p_j_never_applied}.
  	\label{SPA_3_aug_3_2_ai}

  	\item \emph{$p_j$ and $p_{j'}$ are tied on $s_i$'s preference list:} Project $p_j$ is fully available in the finalised matching $M$ by Proposition \ref{th:stable_SPA_S_alt_3_p_j_never_applied} and has always been fully available by Proposition \ref{th:stable_SPA_S_propfullyavailable}. As there is a fully available project tied with $p_{j'}$ on $s_i$'s list, once edge $(s_i,p_{j'})$ is added to $M$, as long as it remains, it must be precarious. Pair $(s_i,p_{j'})$ cannot be removed at any stage before the end of the main while loop, since doing so would mean $s_i$ would apply to $p_j$ before again applying to $p_{j'}$ (since $p_j$ is fully available) contradicting Proposition \ref{th:stable_SPA_S_alt_3_p_j_never_applied}. Also Algorithm \ref{Alg:APPROX_SPA-S_stable_3bibpBrief} cannot change the allocations of any precarious lecturer by Proposition \ref{th:stable_SPA_S_bp_lk_not_prec}. Therefore Algorithm \ref{Alg:APPROX_SPA-S_stable} must terminate with $(s_i,p_{j'})$ as a precarious pair. 

Student $s_{i'}$ must have applied to $p_{j'}$ in phase $2$ since they are unassigned in the finalised matching $M$. By Proposition \ref{th:stable_SPA_S_phase2_not_fa_not_prec}, at the point of application, $p_{j'}$ is not fully available and is non-precarious. But by Proposition \ref{prop-spa-st-subsequentprecarious} $p_{j'}$ cannot subsequently become precarious and so the algorithm will terminate with a non-precarious $p_{j'}$, contradicting the above.

  	\label{SPA_3_aug_3_2_aii}

  	\item \emph{$s_i$ strictly prefers $p_{j'}$ to $p_j$:} We consider three further sub-cases based on $l_k$'s preference list.
  	\label{SPA_3_aug_3_2_aiii}
  	
  	\begin{enumerate}[1.]
  	
  	  \item \emph{$l_k$ strictly prefers $s_i$ to $s_{i'}$:} We know that $s_i$ strictly prefers $p_{j'}$ to $p_j$ and that $l_k$ strictly prefers $s_i$ to $s_{i'}$. But then $(s_i,p_{j'})$ forms a blocking pair of stable matching $M_{opt}$, a contradiction.
 
   		\item \emph{$s_{i'}$ and $s_i$ are tied on $l_k$'s preference list:} Project $p_j$ is fully available in $M$ by Proposition \ref{th:stable_SPA_S_alt_3_p_j_never_applied} and has always been fully available by Proposition \ref{th:stable_SPA_S_propfullyavailable}. Student $s_i$ must have been  assigned to $p_{j'}$ in phase $1$, otherwise $s_i$ would have applied to $p_j$, contradicting Proposition \ref{th:stable_SPA_S_alt_3_p_j_never_applied}.
   	Student $s_{i'}$ must have applied to $p_{j'}$ in phase $2$ since $s_{i'}$ is unassigned in $M$. Denote the point at which $s_{i'}$ applies to $p_{j'}$ in phase $2$ as $T_0$. By Proposition \ref{th:stable_SPA_S_phase2_not_fa_not_prec}, at $T_0$, $p_{j'}$ is not fully available and is non-precarious. By Proposition \ref{prop-spa-st-subsequentprecarious}, $p_{j'}$ remains non-precarious from time $T_0$ until the algorithm's termination. Regardless of whether $(s_i,p_{j'})$ exists in the matching at time $T_0$, we know from time $T_0$, $(s_i,p_{j'})$ cannot be removed from $M$, otherwise $s_i$ would remove $p_{j'}$ from their preference list (as $p_{j'}$ is non-precarious) contradicting the fact that $s_i$ assigned to $p_{j'}$ in phase $1$.
   	
   	\label{SPA_3_aug_3_2_aiii2}
   	
   	We consider the following two possibilities.

   	\begin{itemize}
   		\item \emph{$s_{i'}$ applied to $p_{j'}$ in phase $2$ before pair $(s_i,p_{j'})$ was added:} Assume first that $s_{i'}$ is unsuccessful in its application at $T_0$. From above we know that $p_{j'}$ remains non-precarious from $T_0$ onwards.
   		   		
   		\begin{itemize}
   		\item If $p_{j'}$ is undersubscribed at time $T_0$, then $l_k$ cannot be precarious (as $s_{i'}$ was rejected) and so $l_k$ does not meta-prefer $s_{i'}$ to its worst assignee $s_w \in M(l_k)$ at $T_0$. Lecturer $l_k$ must be full at this point since $p_{j'}$ is not fully available and is undersubscribed. As $l_k$ is full and must remain non-precarious by Proposition \ref{prop-spa-st-subsequentprecarious}, after $T_0$, it is only possible for $l_k$ to improve their allocations, by Proposition \ref{th:stable_SPA_S_lecfullworseprec}. Since $l_k$ meta-prefers $s_{i'}$ to $s_i$ ($s_{i'}$ is in phase $2$), when $s_i$ applies to $p_{j'}$, $s_i$ must also be rejected. Project $p_{j'}$ is non-precarious after $T_0$ and so $s_i$ must remove $p_{j'}$ from their list, contradicting the fact that $s_i$ must be assigned to $p_{j'}$ in phase $1$.
   		
   		\item If $p_{j'}$ is full at $T_0$ then since it is also non-precarious, we know $l_k$ does not meta-prefer $s_{i'}$ to its worst assignee $s_w$ in $M(p_{j'})$. By Proposition \ref{th:stable_SPA_S_lecfullworseprec}, $p_{j'}$ cannot accept assignments that are worse than or equal to the worst assignee in $M(p_{j'})$ until after the main while loop. Therefore when $s_i$ applies to $p_{j'}$ before the main while loop, as in the previous case, they must also be rejected, a contradiction as above.
   		\end{itemize}
   		
   		Assume therefore that $s_{i'}$ is successful in their application at $T_0$. Pair $(s_{i'},p_{j'})$ must be removed at some point after $T_0$ since $(s_{i'},p_{j'}) \notin M$. Denote the time $(s_{i'},p_{j'})$ is removed as $T_1$. The removal at $T_1$ must have occurred before the end of the main while loop since otherwise $s_{i'}$ would be assigned to some project in the finalised matching $M$ (the same students are assigned when removing blocking pairs of type $(3bi)$), which it is not. We know that $(s_i,p_{j'})$ was added either after $T_0$ and before $T_1$ or after $T_1$. Once added $(s_i,p_{j'})$ cannot be removed from above.
   		\begin{itemize}
   			\item Assume $(s_i,p_{j'})$ was added before $T_1$. At $T_1$ (before the end of the main while loop) pair $(s_{i'},p_{j'})$ is removed. This must either be because $p_{j'}$ is undersubscribed and $l_k$ is full, or because $p_{j'}$ is full. 
   			\begin{itemize}
   				\item If the former then $l_k$ is full and cannot be precarious since we are removing a non-precarious pair ($p_{j'}$ is non-precarious after $T_0$). But this removal can only happen if $s_{i'}$ is the worst student assigned to $l_k$ at $T_1$. But by the definition of a worst assignee $s_i$ (being in phase $1$) would be removed before $s_{i'}$. Therefore, $(s_i,p_{j'})$ must have already been removed from the matching, a contradiction to the fact that $(s_i,p_{j'})$ cannot be removed.
   				\item Using similar reasoning, if $p_{j'}$ is full at $T_1$, then (as we are removing a non-precarious pair), $s_{i'}$ must be the worst student assigned in $M(p_{j'})$. But this would mean $(s_i,p_{j'})$ had already been removed, a contradiction.
   			\end{itemize}

   			\item Assume $(s_i,p_{j'})$ was added after $T_1$. Again we consider two sub-cases.
   		\begin{itemize}
   		\item If $p_{j'}$ was undersubscribed at time $T_1$, then $l_k$ must have been full and must be non-precarious since non-precarious pair $(s_{i'},p_{j'})$ was removed. This pair was removed as $s_{i'}$ was a worst student in $l_k$ at $T_1$. By Proposition \ref{prop-spa-st-subsequentprecarious}, $l_k$ remains non-precarious from this point onwards and therefore by Proposition \ref{th:stable_SPA_S_lecfullworseprec}, $l_k$ can only improve their allocations from time $T_1$. Therefore, as $l_k$ meta-prefers $s_{i'}$ to $s_i$ (as $s_{i'}$ is in phase $2$), it must be that $s_i$, applying after $T_1$, will be rejected. This would result in the removal of $p_{j'}$ from $s_i$'s list contradicting the fact that $s_i$ assigned to $p_{j'}$ in phase $1$.

   		\item If $p_{j'}$ is full at $T_1$ then using similar reasoning to above, we know that at $T_1$, $s_{i'}$ is a worst assignee in $M(p_{j'})$. Since, at $T_1$, $p_{j'}$ is full and non-precarious, and remains non-precarious, by Proposition \ref{th:stable_SPA_S_lecfullworseprec}, $p_{j'}$ cannot subsequently accept assignments that are worse than or equal to the worst assignee in $M(p_{j'})$ until the end of the main while loop. Therefore $s_i$ will be rejected on application, a contradiction as above.
   		\end{itemize}

   		\end{itemize}

   		\item \emph{$s_{i'}$ applied to $p_{j'}$ in phase $2$ after pair $(s_i,p_{j'})$ was added:} At $T_0$, since $l_k$ meta-prefers $s_{i'}$ to $s_i$ ($s_i$ is in phase $1$ whereas $s_{i'}$ is in phase $2$), $(s_{i'},p_{j'})$ must be added to the matching with some student other than $s_{i'}$ being removed, since $s_i$ cannot be removed from $M$ from time $T_0$. But now we are in the same position as before where $(s_{i'},p_{j'})$ must be removed from $M$, but this can only happen if $(s_i,p_{j'})$ is removed first, a contradiction.

   	\end{itemize}

  		\item \emph{$l_k$ strictly prefers $s_{i'}$ to $s_i$:} Since $s_{i'}$ is unassigned in $M$, $(s_{i'},p_{j'})$ is a blocking pair of stable matching $M$, a contradiction.

  	\end{enumerate}
  \end{enumerate}

  \item Students $s_i$ and $s_{i'}$ are both assigned to the same project $p_j$ in $M_{opt}$ and $s_i$ is assigned to project $p_{j'}$ in $M$. Both $p_j$ and $p_{j'}$ are offered by lecturer $l_k$. Since $s_{i'}$ is unassigned in $M$, we know that $s_{i'}$ has to have applied to $p_j$ during the algorithm's execution, but this directly contradicts Proposition \ref{th:stable_SPA_S_alt_3_p_j_never_applied}.
\label{SPA_3_aug_3_2_b}

  \item Students $s_i$ and $s_{i'}$ are assigned to $p_j$ and $p_{j''}$ in $M_{opt}$ respectively, and $s_i$ is assigned to $p_{j'}$ in $M$. Lecturer $l_k$ offers $p_j$, $p_{j'}$ and $p_{j''}$.
  By Proposition \ref{th:stable_SPA_S_alt_p_j_undersubscribed}, $p_{j''}$ is undersubscribed in $M$ since lecturer clone $l^{c,2}$ in Figure \ref{SPASTpossiblecomps_6} is connected to an edge in $M'_{opt}$ corresponding to $p_{j''}$ and an edge in $M'$ corresponding to $p_{j'}$. We know that $l_k$ is undersubscribed in $M$ since $l^{c,1}$ is not assigned in $M'$, and so $p_{j''}$ must be fully available. By Proposition \ref{th:stable_SPA_S_propfullyavailable} we know that $p_{j''}$ has always been fully available during the algorithm's execution. Since $s_{i'}$ is unassigned in $M$ they must have applied to $p_{j''}$ during the course of the algorithm. But as $p_{j''}$ has always been fully available this must have been accepted. As we end up with $s_{i'}$ being unassigned it must also be the case that $l_k$ rejects pair $(s_{i'},p_{j''})$ but we know that $p_{j''}$ is always fully available and so this cannot have happened, a contradiction.
  \label{SPA_3_aug_3_2_c}

  \item Students $s_i$ and $s_{i'}$ are assigned to $p_j$ and $p_{j'}$ in $M_{opt}$ respectively, and $s_i$ is assigned to $p_{j'}$ in $M$. Lecturers $l_k$ and $l_{k'}$ offer projects $p_j$ and $p_{j'}$, respectively. Identical arguments to those found in Case \ref{SPA_3_aug_3_2_a}\ref{SPA_3_aug_3_2_ai} and \ref{SPA_3_aug_3_2_a}\ref{SPA_3_aug_3_2_aii} can be used to show a contradiction. Similarly, identical arguments to those found in Case \ref{SPA_3_aug_3_2_a}\ref{SPA_3_aug_3_2_aiii} can also be used to show a contradiction, but exchanging $l_k$ for $l_{k'}$.   
  \label{SPA_3_aug_3_2_d}

% \item \todo{I don't think e is possible to make from $G'$ in the first place.} Students $s_i$ and $s_{i'}$ are both assigned to project $p_j$ in $M_{opt}$ and $s_i$ is assigned to project $p_{j'}$ in $M$. Projects $p_j$ and $p_{j'}$ are offered by lecturers $l_k$ and $l_{k'}$ respectively. Using identical arguments to Case \ref{SPA_3_aug_3_2_b}, as student $s_{i'}$ is unassigned in $M$, they must have applied to $p_j$ during the algorithm's execution, but this contradicts Proposition \ref{th:stable_SPA_S_alt_3_p_j_never_applied}.
% \label{SPA_3_aug_3_2_e}

  \item Students $s_i$ and $s_{i'}$ are assigned to $p_j$ and $p_{j''}$ in $M_{opt}$ respectively, and $s_i$ is assigned to $p_{j'}$ in $M$. Lecturer $l_k$ offers project $p_j$, whereas $l_{k'}$ offers projects $p_{j'}$ and $p_{j''}$. We consider the following $3$ sub-cases.

  	\label{SPA_3_aug_3_2_f}
  \begin{enumerate}[i.]
  	\item \emph{$s_i$ strictly prefers $p_j$ to $p_{j'}$:} Identical arguments to those found in Case \ref{SPA_3_aug_3_2_a}\ref{SPA_3_aug_3_2_ai} can be used to show a contradiction.
  	\label{SPA_3_aug_3_2_fi}
	
  	\item \emph{$p_j$ and $p_{j'}$ are tied on $s_i$'s preference list:} Using similar reasoning to Case \ref{SPA_3_aug_3_2_a}\ref{SPA_3_aug_3_2_aii} we know that once edge $(s_i,p_{j'})$ is added to $M$ it is, and remains, a precarious pair and cannot be removed at any stage. Also $p_{j'}$ must have been fully available on application by $s_i$ otherwise fully available $p_j$ at the same rank would have been applied to by $s_i$. Therefore $p_{j'}$ is either fully available or precarious throughout the algorithm's execution. Student $s_{i'}$ is not assigned in $M$ and so must have applied to $p_{j''}$ whilst they were in phase $2$. Let this time of application be denoted $T_0$. By Proposition \ref{th:stable_SPA_S_phase2_not_fa_not_prec}, $p_{j''}$ cannot be precarious at $T_0$. \label{SPA_3_aug_3_2_fii}
  	\begin{itemize}
  		\item If $p_{j'}$ is fully available at $T_0$ then $l_{k'}$ is undersubscribed and so $s_{i'}$ would only be rejected if $p_{j''}$ was non-precarious, full and $s_{i'}$ was not meta-preferred by $l_{k'}$ to an student in $M(p_{j''})$.
  		\item If $p_{j'}$ is precarious at $T_0$ then $l_{k'}$ is also precarious and therefore $s_{i'}$ would again  only be rejected if $p_{j''}$ was non-precarious, full and $s_{i'}$ was not meta-preferred by $l_{k'}$ to an student in $M(p_{j''})$.
  	\end{itemize}
 Therefore we have the following two cases.

  	\begin{enumerate}[1.]

	\item If $s_{i'}$ was rejected then it must be because $p_{j''}$ was non-precarious, full and $s_{i'}$ was not meta-preferred by $l_{k'}$ to any student in $M(p_{j''})$, by above. But similar to Case \ref{SPA_3_aug_3_2_c} we can say that by Proposition \ref{th:stable_SPA_S_alt_p_j_undersubscribed}, $p_{j''}$ is undersubscribed in the finalised matching $M$. Therefore, at least one non-precarious pair involved with $p_{j''}$ must be removed (without a pair involving $p_{j''}$ immediately replacing it) before the end of the algorithm. Denote this point in the algorithm's execution as $T_1$ and the removed pair $(s_{i''},p_{j''})$ for some student $s_{i''}$. Note $T_1$ may either be before of after the end of the main while loop. This type of removal can only happen when $l_{k'}$ is full (this is clear before the main while loop, and is true after the main while loop by Proposition \ref{th:stable_SPA_S_bp_lk_full}), and once a lecturer is full they remain full (since any pair deletion involving a project of $l_{k'}$ can only occur with a pair addition involving a project of $l_{k'}$). But $(s_i,p_{j'})$ was assigned when $p_{j'}$ was fully available and so $(s_i,p_{j'})$ was assigned before $T_1$. If $T_1$ occurs after the end of the main while loop then $l_{k'}$ must be non-precarious at $T_1$ by Proposition \ref{th:stable_SPA_S_bp_lk_not_prec}. But this means before $T_1$, pair $(s_i,p_{j'})$ is either in the matching but no longer precarious, or has been removed from the matching, a contradiction to the fact that $(s_i,p_{j'})$, once added, remains a precarious pair that can never be removed. Therefore the removal of $(s_{i''},p_{j''})$ at $T_1$ must have occurred before the end of the main while loop. But, since $(s_i,p_{j'})$ is precarious, it would be removed before non-precarious $(s_{i''},p_{j''})$, a contradiction.

	\item If $s_{i'}$ was accepted then pair $(s_{i'},p_{j''})$ would need to be removed before the algorithm terminated (since $(s_{i'},p_{j''}) \notin M$). We know from before that $p_{j''}$ is non-precarious at the point of application and further that pair $(s_{i'},p_{j''})$ must remain non-precarious by definition since $s_{i'}$ applied in phase $2$. Therefore, we need to remove non-precarious pair $(s_{i'},p_{j''})$ from the matching which can only happen if either $p_{j''}$ is full or $l_{k'}$ is full. Firstly assume that  $p_{j''}$ is full and pair $(s_{i'},p_{j''})$ is replaced with a meta-preferred student assigned to $p_{j''}$ ($p_{j''}$ must be non-precarious since we are removing a non-precarious pair). Since $p_{j''}$ needs to be undersubscribed in the finalised matching $M$ we are in the same position and contradiction as the previous case. Secondly, $l_{k'}$ is full (and therefore remains full), in which case pair $(s_i,p_{j'})$ must already be in $M$ (since it was added to $M$ when $p_{j'}$ was fully available) and we can use a similar reasoning to the latter half of the previous case to show a contradiction.

  	\end{enumerate}

  	\item \label{SPA_3_aug_3_2_fiii} \emph{$s_i$ strictly prefers $p_{j'}$ to $p_j$:} We now consider three sub-cases based on $l_{k'}$'s preference list. 

  	\begin{enumerate}[1.]
  	\item \emph{$l_{k'}$ strictly prefers $s_i$ to $s_{i'}$:} We know that $p_{j'}$ is undersubscribed in $M_{opt}$ by Proposition \ref{th:stable_SPA_S_alt_p_j_undersubscribed_mopt}. Either $l_{k'}$ is undersubscribed (in which case $p_{j'}$ is fully available), or $l_{k'}$ is full (and strictly prefers $s_i$ to $s_{i'}$). In either case, $(s_i, p_{j'})$ is a blocking pair of stable $M_{opt}$, a contradiction.

   		\item \label{SPA_3_aug_3_2_fiii2} \emph{$s_{i'}$ and $s_i$ are tied on $l_{k'}$'s preference list:} For this initial paragraph we use some similar reasoning to Case \ref{SPA_3_aug_3_2_a}\ref{SPA_3_aug_3_2_aiii}\ref{SPA_3_aug_3_2_aiii2}. Student $s_i$ must have assigned to $p_{j'}$ in phase $1$, otherwise $s_i$ would have applied to $p_j$ a contradiction to Proposition \ref{th:stable_SPA_S_alt_3_p_j_never_applied}. Unlike Case \ref{SPA_3_aug_3_2_a}\ref{SPA_3_aug_3_2_aiii}\ref{SPA_3_aug_3_2_aiii2}, $p_{j'}$ may be precarious at this point of application. Also, student $s_{i'}$, not being assigned in $M$, must have applied to $p_{j''}$ whilst in phase $2$. Denote the point at which $s_{i'}$ applies to $p_{j''}$ in phase $2$ as $T_0$. At $T_0$ we know that $p_{j''}$ is not fully available and is non-precarious by Proposition \ref{th:stable_SPA_S_phase2_not_fa_not_prec} and that $p_{j''}$ remains non-precarious from this point onwards by Proposition \ref{prop-spa-st-subsequentprecarious}.
   		
   	We look at two possibilities:

   \begin{enumerate}
   \item \label{SPA_3_aug_3_2_fiii2i} \emph{$s_{i'}$ was rejected at $T_0$:} There would be two possible reasons for the rejection. Firstly, that $l_{k'}$ is non-precarious, full and $l_{k'}$ does not meta-prefer $s_{i'}$ to any student in $M(l_{k'})$. Secondly, that $p_{j''}$ is non-precarious, full and $l_{k'}$ does not meta-prefer $s_{i'}$ to any student in $M(p_{j''})$. We can rule out the first option as follows. If $l_{k'}$ is non-precarious and full at $T_0$ then by Proposition \ref{th:stable_SPA_S_lecfullworseprec}, $l_{k'}$ cannot accept a worse student than currently exists in $M(l_{k'})$ for the remainder of the algorithm. Since $s_{i'}$ was rejected we can conclude that no worse student than $s_{i'}$ can exist in $M(l_{k'})$ at $T_0$ and cannot exist in $M(l_{k'})$ from $T_0$ onwards. But $s_{i'}$ being in phase $2$ is meta-preferred to $s_i$ in phase $1$. This contradicts the fact that $(s_i,p_{j'}) \in M$. Therefore, $s_{i'}$ was rejected because $p_{j''}$ is non-precarious, full and $l_{k'}$ does not meta-prefer $s_{i'}$ to any student in $M(p_{j''})$. 
   % {\color{red}By Proposition \ref{th:stable_SPA_S_pfullthennoworseafterwhile}\todo{5.4.22}, $l_{k'}$ cannot accept a student to project $p_{j''}$ that is worse or equal to a worst student existing in $M(p_{j''})$ for the remainder of the algorithm.}
   
   Using a similar strategy to Case \ref{SPA_3_aug_3_2_f}\ref{SPA_3_aug_3_2_fii}, we know that $p_{j''}$ is undersubscribed in the finalised matching $M$ by Proposition \ref{th:stable_SPA_S_alt_p_j_undersubscribed}, therefore before the algorithm terminates a pair $(s_{i''},p_{j''})$ involving $p_{j''}$ must be removed without being immediately replaced with another pair involving $p_{j''}$. Denote the first such occurrence as happening at time $T_1$, where $T_1$ occurs after $T_0$. Note that $T_1$ may be either before or after the end of the main while loop.    
%   As before we note that this kind of removal requires $l_{k'}$ to be full \todoc{reference} and $l_{k'}$ is also non-precarious at $T_1$ as we are removing a non-precarious pair. 
  
  \begin{itemize}
  	\item Assume $T_1$ occurs before the end of the main while loop. We know any removal of the type occurring at $T_1$ must be due to $l_{k'}$ being full and non-precarious (since $(s_{i''},p_{j''})$ is removed as a non-precarious pair). By Proposition \ref{th:stable_SPA_S_lecfullworseprec}, we know $l_{k'}$ cannot subsequently be assigned in $M$ a student worse than or equal to a worst student in $M(l_{k'})$ at $T_1$. Using the same proposition we know that $p_{j''}$ cannot be assigned in $M$ a worse student until the end of the main while loop than exists in the matching at $T_0$. 
  
 Since a pair involving $p_{j''}$ was removed at $T_1$, a worst assignee in $M(l_{k'})$ at $T_1$ can be no worse than a worst assignee in $M(p_{j''})$ at $T_0$. Finally, this means that no student assigned to $l_{k'}$ from $T_1$ onwards can be worse than $s_{i'}$ rejected at $T_0$, but $s_i$ being in phase $1$ is worse than $s_{i'}$ in phase $2$ according to $l_{k'}$, a contradiction to the fact that $(s_i,p_{j'})\in M$.
  	
  	\item Assume therefore that $T_1$ occurs after the end of the main while loop. Then $s_i$ has to be assigned to $p_{j'}$ at this point. Since $p_{j''}$ becomes undersubscribed at $T_1$, pair $(s_{i''},p_{j'''})$ must be a blocking pair of type $(3bi)$, where $p_{j'''}$ is an undersubscribed project of $l_{k'}$, and $s_i$ strictly prefers $p_{j'''}$ to $p_{j''}$. By Propositions \ref{th:stable_SPA_S_3biworstinmlk} and \ref{th:stable_SPA_S_3bicreatedworstinmlk}, $s_{i''}$ must be a worst student in $M(l_{k'})$ and therefore $M(p_{j''})$ at $T_1$. But, we also know that at $T_0$ when $s_{i'}$ was rejected, $l_{k'}$ could not subsequently accept a student to project $p_{j''}$ that is worse than a worst student existing in $M(p_{j''})$ for the remainder of the algorithm, by Proposition \ref{th:stable_SPA_S_pfullthennoworseafterwhile}. Also, since $s_{i'}$ was rejected, no worse student can exist in $p_{j''}$ at $T_0$. Therefore,  since $(s_{i''}, p_{j''})$ exists in $M$ just before $T_1$, $s_{i''}$ cannot be worse than $s_{i'}$ according to $l_{k'}$ (and so must either be of equal rank and in phase $2$ or of higher rank). Recall that $(s_i,p_{j'})$ was assigned in phase $1$ and so must exist in the matching at $T_1$. This means that $s_{i''}$ is either at an equal rank to $s_i$ (as the rank of $s_i$ and $s_{i'}$ are equal) but is in phase $2$ with $s_i$ being in phase $1$, or $s_{i''}$ is at a higher rank than $s_i$. In either case this contradicts the fact that $s_{i''}$ is a worst student in $M(l_{k'})$ at $T_1$.
  \end{itemize}

   		\item \emph{$s_{i'}$ application to $p_{j''}$ was accepted at $T_0$:} Pair $(s_{i'}, p_{j''})$ does not exist in the finalised matching $M$ and therefore must be removed sometime after $T_0$. We know $(s_{i'}, p_{j''})$ is always non-precarious by definition (as $s_{i'}$ applied in phase $2$ and hence is removed as a non-precarious pair. Since $s_{i'}$ is unassigned in the finalised matching $M$ and Algorithm \ref{Alg:APPROX_SPA-S_stable_3bibpBrief} cannot change which students are assigned, $(s_{i'}, p_{j''})$ must be removed before the end of the main while loop. Denote this time as $T_2$. 
   		
   		This can only happen if either $l_{k'}$ is full and $s_{i'}$ is a worst assignee in $M(l_{k'})$ or $p_{j''}$ is full and $s_{i'}$ is a worst assignee in $M(p_{j''})$. If the former, then by Proposition \ref{th:stable_SPA_S_lecfullworseprec}, $l_{k'}$ cannot accept a worse student than a current worst student in $M(l_{k'})$. This worst student cannot be worse than $s_{i'}$, since $s_{i'}$ was just removed, hence $s_i$ (of equal rank to $s_{i'}$ and in phase $1$) cannot be assigned a project of $l_{k'}$'s in the finalised matching $M$, a contradiction. Therefore at $T_2$, $p_{j''}$ is full and $s_{i'}$ is removed as a worst assignee in $M(p_{j''})$. 
   		% From $T_2$ onwards, no student worse than a worst student in $M(p_{j''})$ can be assigned to $p_{j''}$ by Proposition \ref{th:stable_SPA_S_pfullthennoworseafterwhile}\todo{5.4.22}.

   		But since $p_{j''}$ must be undersubscribed in the finalised matching $M$ by Proposition \ref{th:stable_SPA_S_alt_p_j_undersubscribed}, we can now use almost identical arguments as in case \ref{SPA_3_aug_3_2_f}\ref{SPA_3_aug_3_2_fiii}\ref{SPA_3_aug_3_2_fiii2}\ref{SPA_3_aug_3_2_fiii2i} to show a contradiction, noting that $T_2$ replaces $T_0$ and $s_{i'}$ was removed rather than rejected.

   	\end{enumerate}

   		\item \emph{$l_{k'}$ strictly prefers $s_{i'}$ to $s_i$:} 
   		Project $p_{j''}$ is undersubscribed in $M$ by Proposition \ref{th:stable_SPA_S_alt_p_j_undersubscribed}. Either $l_{k'}$ is undersubscribed (in which case $p_{j''}$ is fully available), or $l_{k'}$ is full (and strictly prefers $s_{i'}$ to $s_i$). In either case, $(s_{i'}, p_{j''})$ is a blocking pair of stable $M$, a contradiction.
   		 \end{enumerate}
   	
  	\end{enumerate}

\end{enumerate}

Therefore it is not possible for a component structured as in Figure \ref{SPASTpossiblecomps_6} to exist in $G'$.
\end{proof}

Similar to our previous lemma, we prove in Lemma \ref{th:stable_SPA_S_lemma_no_alt_size_1} that it is not possible for a component of the type shown in Figure \ref{SPASTpossiblecomps_7} to exist in $G'$.

% THEOREM BEGIN
\begin{lemma} 
	\label{th:stable_SPA_S_lemma_no_alt_size_1}
	Let $M$ be a stable matching found by Algorithm \ref{Alg:APPROX_SPA-S_stable} for instance $I$ of \acrshort{spa-st}, and let $M_{opt}$ be a maximum stable matching in $I$. No component of the type given in Figure \ref{SPASTpossiblecomps_7} can exist in the mapped graph $G'$. 
	\end{lemma}
% THEOREM END

\begin{proof}
Let $p_j$ be the project that student $s^{c,1}$ is assigned to in $M_{opt}$ and let lecturer clone $l^{c,1}$ correspond to lecturer $l_k$ in $G$. $l_k$ must be undersubscribed in $M$ as there is a lecturer clone $l^{c,1}$ unassigned in $M'$. Also, by Proposition \ref{th:stable_SPA_S_alt_p_j_undersubscribed}, $p_j$ must also be undersubscribed in $M$. Therefore, $p_j$ is fully available at the end of the algorithm's execution and must always have been fully available by Proposition \ref{th:stable_SPA_S_propfullyavailable}.

Student $s_i$ is unassigned in $M$ and so we know that $s_i$ has to have applied to $p_j$ during the algorithm's execution. Since $p_j$ has always been fully available this had to have been accepted. Now, $(s_i,p_j)$ is not in the finalised matching $M$ and so it must have been removed and this could only have happened if $p_j$ or $l_k$ were full. But this contradicts the fact that $p_j$ has always been fully available. Therefore, no component of type \ref{SPASTpossiblecomps_7desc} can exist in the mapped graph $G'$.
\end{proof}

Finally, Theorem \ref{th:stable_SPA_S} proves that Algorithm {\sf Max-SPA-ST-Approx} is a $\frac{3}{2}$-approximation algorithm for \acrshort{max-spa-st}.

% THEOREM BEGIN
\begin{theorem} 
	\label{th:stable_SPA_S}
	Let $M$ be a stable matching found by Algorithm \ref{Alg:APPROX_SPA-S_stable} for instance $I$ of \acrshort{spa-st}, and let $M_{opt}$ be a maximum stable matching in $I$. Then $|M|\geq \frac{2}{3}|M_{opt}|$.
	\end{theorem}
% THEOREM END

\begin{proof}
Let $G'$ be the mapped graph constructed from the underlying graph of the instance $G$. Components in $G'$ may only exist in the forms shown in Figure \ref{SPASTpossiblecomps}. Therefore we need only show that no component in $G'$ can exist where the number of $M'$ edges is less than two-thirds of the number of $M'_{opt}$ edges. We run through each component of Figure \ref{SPASTpossiblecomps} in turn. Let the current component be denoted $c$, where $M'(c)$ and $M'_{opt}(c)$ denote the set of edges in $M'$ and $M'_{opt}$ involved in $c$, respectively. 

For Case \ref{SPASTpossiblecomps_1}, an alternating cycle, and Cases \ref{SPASTpossiblecomps_2} and \ref{SPASTpossiblecomps_3}, alternating paths of even length, it is clear that $|M'(c)| = |M'_{opt}(c)|$. Case \ref{SPASTpossiblecomps_4} involves an odd length alternating path with end edges in $M'$. It must be the case therefore that $|M'(c)| > |M'_{opt}(c)|$ for components of this type. Case \ref{SPASTpossiblecomps_5} shows an odd length alternating path with end edges in $M'_{opt}$, but for path sizes greater than $5$. Therefore, $|M'(c)| \geq \frac{2}{3}|M'_{opt}(c)|$ as required. Neither Case \ref{SPASTpossiblecomps_6} nor \ref{SPASTpossiblecomps_7} can exist in $G'$ by Lemmas \ref{th:stable_SPA_S_lemma_no_alt_size_3} and \ref{th:stable_SPA_S_lemma_no_alt_size_1} respectively.

Hence it is not possible for the mapped graph $G'$ to contain components in which $|M'(c)| < \frac{2}{3}|M'_{opt}(c)|$. Algorithm \ref{Alg:APPROX_SPA-S_stable} is therefore a $\frac{3}{2}$-approximation algorithm for the problem of finding a maximum stable matching in $I$.
\end{proof}

%%%%%%%%%%%%%%%%%%%%%%%%%%%%%%%%%%%%%%%%%%%%
%%%%%%%%%%%%%%%%%%%%%%%%%%%%%%%%%%%%%%%%%%%%

\subsection{Lower bound for the algorithm}
\label{sec:lb_alg_2_3}

Figure \ref{fig:example_2_3_size_instance} shows instance $I_5$ of \acrshort{spa-st}. A maximum stable matching $M'$ in $I$ is given by $M'=\{(s_1,p_2),(s_2,p_3),(s_3,p_1)\}$. The only possible blocking pairs for this matching are $(s_3,p_3)$ and $(s_3,p_2)$. However, neither pair can be a blocking pair since $l_2$ prefers both of their current assignees to $s_3$. 

A trace is given as Table \ref{tab:example_2_3_size_instance_trace} which shows the execution run of Algorithm \ref{Alg:APPROX_SPA-S_stable} over instance $I_5$. The algorithm outputs stable matching $M=\{(s_1,p_3),(s_3,p_2)\}$. The possible blocking pairs of this matching are $(s_2,p_3)$ and $(s_3,p_3)$. Neither can be a blocking pair since  $l_2$ prefers $s_1$ to both $s_2$ and $s_3$.

Therefore, Algorithm {\sf Max-SPA-ST-Approx} has found a stable matching that is \emph{exactly} two-thirds of the size of the maximum stable matching, thus the algorithm cannot guarantee a better bound than $\frac{3}{2}$.

We further note that this result holds for the arbitrarily large family of instances generated by copying instance $I_5$ a constant number of times, such that in each copy, indices of students and projects increase by $3$ and indices of lecturers increase by $2$.

\begin{figure}
\begin{minipage}[t]{0.3\textwidth}
Student preferences:\\
$s_1$: ($p_3$ $p_2$)\\
$s_2$: $p_3$\\
$s_3$: $p_3$ $p_2$ $p_1$\\
\end{minipage}
\begin{minipage}[t]{0.3\textwidth}
Project details:\\
$p_1$: lecturer $l_1$, $c_1=2$\\
$p_2$: lecturer $l_1$, $c_2=1$\\
$p_3$: lecturer $l_2$, $c_3=1$\\
\end{minipage}
\begin{minipage}[t]{0.35\textwidth}
Lecturer preferences:\\
\begin{minipage}[t]{0.6\textwidth}
$l_1$: $s_1$ $s_3$\\
$l_2$: $s_1$ $s_2$ $s_3$\\
\end{minipage}
\begin{minipage}[t]{0.25\textwidth}
$d_1=2$\\
$d_2=1$\\
\end{minipage}
\end{minipage}
\caption[{\sc spa-st} instance $I_5$ in which Algorithm {\sf Max-SPA-ST-Approx} finds a stable matching two-thirds of the size of optimal.]{\acrshort{spa-st} instance $I_5$ in which Algorithm {\sf Max-SPA-ST-Approx} finds a stable matching two-thirds of the size of optimal.}
\label{fig:example_2_3_size_instance}
\end{figure}

\begin{table}[tbp] \centering\begin{tabular}{ p{0.5cm}p{9cm} | p{0.7cm} p{0.7cm}p{0.7cm}}\hline\hline 
& Action & $s_1$ & $s_2$ & $s_3$  \\ 
\hline 
$1$ & $s_1$ applies to $p_3$, accepted & $p_3$ &  & \\
$2$ & $s_2$ applies to $p_3$, accepted &  & $p_3$ & \\
$3$ & $s_3$ applies to $p_3$, rejected, $s_3$ removes $p_3$ &  & $p_3$ & \\
$4$ & $s_3$ applies to $p_2$, accepted &  & $p_3$ & $p_2$ \\
$5$ & $s_1$ applies to $p_3$, accepted, $s_2$ removes $p_3$ & $p_3$ &  & $p_2$ \\
$6$ & $s_2$ moves to phase $2$ &  $p_3$ &  & $p_2$ \\
$7$ & $s_2$ applies to $p_3$, rejected, $s_2$ removes $p_3$ & $p_3$ &  & $p_2$ \\
$8$ & $s_2$ moves to phase $3$ &  $p_3$ &  & $p_2$ \\
   \hline\hline \end{tabular} \caption[Trace of running Algorithm {\sf Max-SPA-ST-Approx} for instance $I_5$.]{Trace of running Algorithm {\sf Max-SPA-ST-Approx} for instance $I_5$ in Figure \ref{fig:example_2_3_size_instance}. In this table, the phrase ``$s_i$ removes $p_j$'' indicates that student $s_i$ removes project $p_j$ from their preference list.}\label{tab:example_2_3_size_instance_trace} \end{table}

\section{IP model for {\sc max-spa-st}}
\label{sec_ip}

In this section we present an \acrshort{ip} model for \acrshort{max-spa-st}.

\subsection{Stability definition}
For the stability constraints in the model, it is advantageous to use an equivalent condition for stability, as given by the following lemma.

% THEOREM BEGIN
\begin{lemma}
	\label{lemma_stability_2}
	Let $I$ be an instance of \acrshort{spa-st} and let $M$ be a matching in $I$. Then $M$ is stable if and only if the following condition, referred to as condition (*) holds: For each student $s_i \in S$ and project $p_j \in P$, if $s_i$ is unassigned in $M$ and finds $p_j$ acceptable, or $s_i$ prefers $p_j$ to $M(s_i)$, then either:
	\begin{itemize}
	\item  $l_k$ is full, $s_i \notin M(l_k)$ and $l_k$ prefers the worst student in $M(l_k)$ to $s_i$ or is indifferent between them, or;
	\item $p_j$ is full and $l_k$ prefers the worst student in $M(p_j)$ to $s_i$ or is indifferent between them, where $l_k$ is the lecturer offering $p_j$.
	\end{itemize}
\end{lemma}

\begin{proof}

Suppose $M$ is stable. Assume for contradiction that condition (*) is not satisfied. Then there exists a student $s_i \in S$ and a project $p_j \in P$ such that $s_i$ is unassigned in $M$ and finds $p_j$ acceptable, or $s_i$ prefers $p_j$ to $M(s_i)$, and one of the following four cases arises:

\begin{enumerate}
  \item $p_j$ and $l_k$ are both undersubscribed;
  \item $p_j$ is undersubscribed, $l_k$ is full and $s_i \in M(l_k)$;
  \item $p_j$ is undersubscribed, $l_k$ is full and $l_k$ prefers $s_i$ to the worst student in $M(l_k)$;
  \item $p_j$ is full and $l_k$ prefers $s_i$ to the worst student in $M(p_j)$.
\end{enumerate} 

Each of these scenarios clearly describes a blocking pair as in Section \ref{sec:spa-st_stab_defns}, hence we have a contradiction to the stability of $M$.

Conversely, assume $M$ satisfies condition (*). Suppose for contradiction that $M$ is not stable. Then, there exists a blocking pair $(s_i,p_j)$, implying that $s_i$ is unassigned in $M$ and finds $p_j$ acceptable, or $s_i$ prefers $p_j$ to $M(s_i)$, and one of the above four cases will be true.

Whichever one of these cases holds, we then obtain a contradiction to the fact that the conditions given in the statement of the theorem holds. Thus $M$ is stable.
\end{proof}
% THEOREM END

\subsection{Description of variables and constraints}
%\label{ip_max_model}
The key variables in the model are binary-valued variables $x_{ij}$, defined for each $s_i\in S$ and $p_j\in P$, where $x_{ij} = 1$ if and only if student $s_i$ is assigned to project $p_j$. Additionally, we have binary-valued variables $\alpha_{ij}$ and $\beta_{ij}$ for each $s_i\in S$ and $p_j\in P$. These variables allow us to more easily describe the stability constraints below. For each $s_i\in S$ and $l_k\in L$, let 
$$T_{ik} = \{s_u \in S: \text{rank}(l_k,s_u) \leq \text{rank}(l_k,s_i) \wedge s_u \neq s_i\}.$$ 
That is, $T_{ik}$ is the set of students ranked at least as highly as student $s_i$ in lecturer $l_k$'s preference list not including $s_i$. Also, for each $p_j\in P$, let 
$$T_{ijk} = \{s_u \in S: \text{rank}(l_k,s_u) \leq \text{rank}(l_k,s_i) \wedge s_u \neq s_i \wedge p_j \in A(s_u)\}.$$ 
That is, $T_{ijk}$ is the set of students $s_u$ ranked at least as highly as student $s_i$ in lecturer $l_k$'s preference list, such that project $p_j$ is acceptable to $s_u$, not including $s_i$. Finally, let $S_{ij} = \{p_r \in P: \text{rank}(s_i,p_r) \leq \text{rank}(s_i,p_j)\}$, that is, $S_{ij}$ is the set of projects ranked at least as highly as project $p_j$ in student $s_i$'s preference list, including $p_j$. Figure \ref{IP:SPA-ST} shows the IP model for \acrshort{max-spa-st}.

\begin{figure}[]
\begin{mdframed}
\begin{flalign*}
&\text{maximise: } \sum_{s_i \in S}\sum_{p_j \in P} x_{ij}&
\end{flalign*}

  subject to:
  
\begin{flalign*}
	1.\hspace{1cm}& x_{ij} \leq 0 & \forall{s_i \in S} \hspace{0.2cm} \forall{p_j \in P},\hspace{0.1cm} p_j \notin A(s_i)\\
	2. \hspace{1cm}& \sum_{p_j \in P}x_{ij} \leq 1 & \forall{s_i \in S}\\
	3. \hspace{1cm}& \sum_{s_i \in S}x_{ij} \leq c_j & \forall{p_j \in P}\\
	4. \hspace{1cm}& \sum_{s_i \in S} \sum_{p_j \in P_k}x_{ij} \leq d_k & \forall{l_k \in L}\\
	5.\hspace{1cm} & 1 - \sum_{p_r \in S_{ij}}x_{ir} \leq \alpha_{ij} + \beta_{ij} & \forall{s_i \in S} \hspace{0.2cm} \forall{p_j \in P}\\
	6.\hspace{1cm}& \sum_{s_u \in T_{ik}} \sum_{p_r \in P_k} x_{ur} \geq d_k  \alpha _{ij}          &  \forall{s_i \in S} \hspace{0.2cm} \forall{p_j \in P}\\
	7.\hspace{1cm} & \sum_{s_u \in T_{ijk}} x_{uj} \geq c_j  \beta_{ij}         &  \forall{s_i \in S} \hspace{0.2cm} \forall{p_j \in P}\\
\end{flalign*}

\begin{flalign*}
	& x_{ij} \in \{0,1\}  & \forall{s_i \in S} \hspace{0.2cm} \forall{p_j \in P}\\
	& \alpha_{ij} \in \{0,1\} & \forall{s_i \in S} \hspace{0.2cm} \forall{p_j \in P}\\
	& \beta_{ij} \in \{0,1\} & \forall{s_i \in S} \hspace{0.2cm} \forall{p_j \in P}\\
\end{flalign*}
\end{mdframed}
\caption[IP model for {\sc max-spa-st}.]{IP model for \acrshort{max-spa-st}.}
\label{IP:SPA-ST}
\end{figure}

Constraint $1$ enforces $x_{ij} = 0$ if $s_i$ finds $p_j$ unacceptable (since $x_{ij} \in \{0, 1\}$). Constraint $2$ ensures that a student may be assigned to a maximum of one project. Constraints $3$ and $4$ ensure that project and lecturer capacities are enforced. In the left hand side of the inequality of Constraint $5$, if $1 - \sum_{p_r \in S_{ij}}x_{ir} = 1$, then either $s_i$ is unassigned or $s_i$ prefers $p_j$ to $M(s_i)$. This ensures that if $s_i$ is unassigned or $s_i$ prefers $p_j$ to $M(s_i)$ then either $\alpha_{ij}=1$ or $\beta_{ij} = 1$ (or both), where $\alpha_{ij}$ and $\beta_{ij}$ are described in Constraints $6$ and $7$. Constraint $6$ ensures that, if $\alpha_{ij}=1$, then the number of students assigned to $l_k$ who are ranked at least as highly as student $s_i$ by $l_k$ (not including $s_i$) must be at least $l_k$'s capacity $d_k$. Constraint $7$ ensures that, if $\beta_{ij}=1$, then the number of students assigned to $p_j$ who are ranked at least as highly as student $s_i$ by $l_k$ (not including $s_i$) must be at least $p_j$'s capacity $c_j$.

Finally, for our optimisation we maximise the sum of all $x_{ij}$ variables in order to maximise the number of students assigned. 
% {\color{blue}Note that in order to solve the problem of finding a minimum stable matching (\acrshort{min-spa-st}), we need only change the optimisation of this \acrshort{ip} model from maximising the sum of all $x_{ij}$ variables to minimising this sum.}
The following section establishes the correctness of the IP model. 

\subsection{Proof of correctness}

% THEOREM BEGIN
\begin{theorem} 
	\label{thm:ip_stable_correct}
	Given an instance $I$ of \acrshort{spa-st}, let $J$ be the IP model as defined in Figure \ref{IP:SPA-ST}. A stable matching in $I$ corresponds to a feasible solution in $J$ and vice versa. 
\end{theorem}
% THEOREM END

\begin{proof}

Assume instance $I$ of \acrshort{spa-st} contains a stable matching $M$. We construct a feasible solution to $J$ involving the variables \textbf{x}, $\boldsymbol\upalpha$ and $\boldsymbol\upbeta$ as follows.

The variables \textbf{x}, $\boldsymbol\upalpha$ and $\boldsymbol\upbeta$ are constructed as follows. For each student $s_i \in S$ and for each project $p_j \in P$, if $s_i$ is assigned to $p_j$ in $M$ then we set variable $x_{ij}=1$, otherwise $x_{ij} = 0$. Let lecturer $l_k$ be the proposer of project $p_j$. Let variable $\alpha_{ij}=1$ if the following two conditions hold: \emph{i)} student $s_i$ is not assigned to lecturer $l_k$, and \emph {ii)} lecturer $l_k$ is full and prefers their worst ranked assignee to $s_i$, or is indifferent between them. Else let $\alpha_{ij}=0$. Let variable $\beta_{ij}=1$ if the following two conditions hold: \emph{i)} student $s_i$ is not assigned to project $p_j$, and \emph {ii)} project $p_j$ is full,  and $l_k$ prefers $p_j$'s worst assignee to $s_i$, or is indifferent between them. Else, set $\beta_{ij}=0$.

Now it must be shown that all constraints described in Figure \ref{IP:SPA-ST} are satisfied.

\begin{enumerate}
	\item \emph{Constraints 1 - 4.} It is clear by the construction of $J$ that Constraints 1-4 are satisfied.

	\item \emph{Constraint 5.} Recall $S_{ij} = \{p_r \in P: \text{rank}(s_i,p_r) \leq \text{rank}(s_i,p_j)\}$ is the set of projects ranked at least as highly as $p_j$ in $s_i$'s preference list. Let $\gamma_{ij} = 1 - \sum_{p_r \in S_{ij}}x_{ir}$. We must show that whenever $\gamma_{ij} = 1$, $\alpha_{ij} + \beta_{ij} \geq 1$. Assume $\gamma_{ij} = 1$, that is $s_i$ is unassigned or would prefer to be assigned to $p_j$ than to $M(p_j)$. As $M$ is stable we know that condition (*) of Lemma \ref{lemma_stability_2} is satisfied. Therefore $\alpha_{ij} + \beta_{ij} \geq 1$ by construction. This directly satisfies Constraint 5.

	\item \emph{Constraint 6.} Recall that for student $s_i$ and lecturer $l_k$, $T_{ik}$ is the set of students ranked at least as highly as student $s_i$ in lecturer $l_k$'s preference list, not including $s_i$. Assume $\alpha_{ij} = 1$. Then, by definition, we know that lecturer $l_k$ is full and prefers their worst ranked assignee to $s_i$, or is indifferent between them. Therefore, the LHS of the inequality must equal $d_k$ and so this constraint is satisfied.

	\item \emph{Constraint 7.} Recall that $T_{ijk}$ is the set of students ranked at least as highly as student $s_i$ in lecturer $l_k$'s preference list, such that the project $p_j$ is acceptable to each student. Similar to above, assume $\beta_{ij} = 1$. Then, by definition, we know that project $p_j$ is full and $l_k$ prefers their worst ranked assignee in $M(p_j)$ to $s_i$, or is indifferent between them. Therefore, the LHS of the inequality must equal $c_j$ and so this constraint is also satisfied.
\end{enumerate}

We have shown that the assignment of values to \textbf{x}, $\boldsymbol\upalpha$ and $\boldsymbol\upbeta$ satisfy all the constraints in $J$, thus if there is a stable matching $M$ in $I$, then there is a feasible solution of $J$.

Conversely, we now show that a feasible solution of $J$ corresponds to a stable matching $M$ in $I$. Let \textbf{x}, $\boldsymbol\upalpha$ and $\boldsymbol\upbeta$ be a feasible solution of $J$. For each $x_{ij}$ variable in $J$, if $x_{ij} = 1$ then add pair $(s_i,p_j)$ to $M$ in $I$. It is now shown that this assignment of students to projects satisfies the definition of a stable matching $M$ in $I$.

\begin{enumerate}
	\item The following constraints are clearly satisfied by Constraints 1-4:
	\begin{itemize}
		\item \emph{A student $s_i$ may be assigned to a maximum of $1$ project;}
		\item \emph{A student $s_i$ may only be assigned to a project that they find acceptable;}
		\item \emph{The number of students assigned to project $p_j$ is less than or equal to $c_j$;}
		\item \emph{The number of students assigned to projects offered by lecturer $l_k$ is less than or equal to $d_k$.}
	\end{itemize}

	\item \emph{$M$ is stable}. Assume for contradiction that there exists a blocking pair $(s_i,p_j) \in M$. Then by Lemma \ref{lemma_stability_2}, neither of the sub-conditions of condition (*) can be true. Both of these sub-conditions being false imply that, as Constraints 6 and 7 must be satisfied, $\alpha_{ij} = 0$ and $\beta_{ij} = 0$. 

Recall $\gamma_{ij} = 1 - \sum_{p_r \in S_{ij}}x_{ir}$. $\sum_{p_r \in S_{ij}}x_{ir}$ is the number of projects that student $s_i$ is assigned to at a higher or equal ranking than $p_j$ in $s_i$'s preference list (including $p_j$). Since $(s_i,p_j)$ is a blocking pair, then it must be the case that $\sum_{p_r \in S_{ij}}x_{ir} = 0$. But this forces $\gamma_{ij} = 1$, and we know that $\alpha_{ij}$ and $\beta_{ij}$ are equal to $0$ so Constraint 5 is contradicted. 
\end{enumerate}

We have shown that if there is a feasible solution of \textbf{x}, $\boldsymbol\upalpha$ and $\boldsymbol\upbeta$ of $J$, then there is a stable matching $M$ in $I$. This completes the proof.
\end{proof}

% THEOREM BEGIN
\begin{corollary} 
	\label{thm:ip_max_correct_app} 
	Given an instance $I$ of \acrshort{spa-st}, let $J$ be the IP model as defined in Figure \ref{IP:SPA-ST}. A maximum stable matching in $I$ corresponds to an optimal solution in $J$ and vice versa. 
\end{corollary}
% THEOREM END

\begin{proof}
	Assume $M$ is a maximum stable matching in $I$. Let $f = \langle$\textbf{x}$,\boldsymbol\upalpha, \boldsymbol\upbeta \rangle$ be the solution in $J$ constructed according to the description in Theorem \ref{thm:ip_stable_correct}. We must show that $f$ forms an optimal solution of $J$. Firstly, since $M$ is stable, we know by Theorem \ref{thm:ip_stable_correct} that $f$ is a feasible solution of $J$. Suppose for contradiction that $f$ is not optimal. Then there is some solution $g = \langle$\textbf{x}$',\boldsymbol\upalpha', \boldsymbol\upbeta' \rangle$ of $J$ in which obj$(g) >$ obj$(f)$, where obj$(f')$ gives the objective value of $f'$. But by construction, $g$ would translate into a stable matching $M'$ such that $|M'| = \text{obj}(g) > \text{obj}(f) = |M|$ in $I$ contradicting the fact that $M$ is maximum.
	
	Conversely, assume $f = \langle$\textbf{x}$,\boldsymbol\upalpha, \boldsymbol\upbeta \rangle$ is an optimal solution in $J$, and let $M$ be the stable matching in $I$ constructed according to the description in Theorem \ref{thm:ip_stable_correct}. Suppose for contradiction that there is some stable matching $M'$ in $I$ such that $|M'| > |M|$. Then by construction, there must be some corresponding solution $g = \langle$\textbf{x}$',\boldsymbol\upalpha', \boldsymbol\upbeta' \rangle$ of $J$ such that $\text{obj}(g) = |M'| > |M| = \text{obj}(f)$, giving the required contradiction.
	\end{proof}

\section{Experimental evaluation}
\label{sec_spa-st_more_sec_exp}

\subsection{Methodology}
Experiments were conducted on our $\frac{3}{2}$-approximation algorithm and an implementation of the IP model in Figure \ref{IP:SPA-ST} using randomly-generated data, in order to measure the effects on matching properties when changing different parameter values (including instance size, probability of ties in preference lists and preference list lengths). Two further experiments explored the scalability of instance size and preference list lengths. Instances were generated using both existing and new software. The existing software is known as the \emph{Matching Algorithm Toolkit} and is a collaborative project developed by students and staff at the University of Glasgow. A web application of the Matching Algorithm Toolkit, which allows access to many of its functions was developed by Lazarov \cite{BL19}.

%Let the number of students, projects and lecturers per instance be given by $n_1$, $n_2$ and $n_3$, respectively.
For a given \acrshort{spa-st} instance, let the total project and lecturer capacities be denoted by $c_P$ and $d_L$, respectively. Capacities were distributed randomly over the sets of projects and lecturers, subject to there being a maximum difference of $1$ between the capacities of any two projects or any two lecturers. The minimum and maximum size of student preference lists is given by $l_{min}$ and $l_{max}$, and $t_s$ represents the probability that a project on a student's preference list is tied with the next project. Lecturer preference lists were generated initially from the student preference lists, where a lecturer $l_k$ must rank a student if a student ranks a project offered by $l_k$. These lists were randomly shuffled and $t_l$ denotes the ties probability for lecturer preference lists. A linear distribution was used to make some projects more popular than others and in all experiments the most popular project is around $5$ times more popular than the least. This distribution influenced the likelihood of a student finding a given project acceptable.  Parameter details for each experiment are given below.

\begin{enumerate}
	\item \emph{Increasing instance size:} $10$ sets of  $10,000$ instances were created (labelled SIZE1, ..., SIZE10). The number of students $n_1$ increased from $100$ to $1000$ in steps of $100$, with $n_2=0.6n_1$, $n_3=0.4n_1$, $c_P=1.4n_1$ and $d_L=1.2n_1$.  The probabilities of ties in preference lists were $t_s = t_l = 0.2$ throughout all instance sets. Student preference list lengths, bound by $l_{min}=3$ and $l_{max}=5$, also remained the same and were kept low to ensure a wide variability in stable matching size per instance. 
	\item \emph{Increasing probability of ties:} $11$ sets of $10,000$ instances were created (labelled TIES1, ..., TIES11). Throughout all instance sets $n_1=300$, $n_2=250$, $n_3=120$, $c_P=420$, $d_L=360$, $l_{min}=3$ and $l_{max}=5$.  The probabilities of ties in student and lecturer preference lists increased from $t_s = t_l = 0.0$ to $t_s = t_l = 0.5$ in steps of $0.05$.  
	\item \emph{Increasing preference list lengths:} $10$ sets of $10,000$ instances were generated (labelled PREF1, ..., PREF10). Similar to the TIES cases, throughout all instance sets $n_1=300$, $n_2=250$, $n_3=120$, $c_P=420$ and $d_L=360$. Additionally, $t_s = t_l = 0.2$. Student preference list lengths increased from $l_{min}=l_{max}=1$ to $l_{min}=l_{max}=10$ in steps of $1$.
	\item \emph{Instance size scalability:} $5$ sets of  $10$ instances were generated (labelled SCALS1, ..., SCALS5). All instance sets in this experiment used the same parameter values as the SIZE experiment, except the number of students $n_1$ increased from $10,000$ to $50,000$ in steps of $10,000$.
	\item \emph{Preference list scalability:} Finally, $6$ sets of $10$ instances were created (labelled SCALP1, ..., SCALP6). Throughout all instance sets $n_1=500$ with the same values for other parameters as in the SIZE experiment. However in this case ties were fixed at $t_s = t_l = 0.4$, and student preference list lengths increased from $l_{min}=l_{max} = 25$ to $l_{min}=l_{max} = 150$ in steps of $25$.
\end{enumerate}

For each generated instance, we ran our $\frac{3}{2}$-approximation algorithm and then used the IP model to find a maximum stable matching. Additionally, we computed a minimum stable matching using a simple adaptation of our IP model for \acrshort{max-spa-st}.  A timeout of $1800$ seconds ($30$ minutes) was imposed on each algorithm, for each instance run, with all instances of Experiments $1$, $2$ and $3$ completing within the timeout time. All experiments were conducted on a machine with 32 cores, 8$\times$64GB RAM and Dual Intel\textsuperscript{\textregistered} Xeon\textsuperscript{\textregistered} CPU E5-2697A v4 processors. The operating system was Ubuntu version 17.04 with all code compiled in Java version 1.8, where the IP models were solved using Gurobi version 7.5.2. Each approximation algorithm instance was run on a single thread while each IP instance was run on two threads. No attempt was made to parallelise Java garbage collection. The statistics collection program, and plot and table generation program, were written in Python and run using Python version 3.6.1. 
Repositories for the software and data used in these experiments can be found at \url{https://doi.org/10.5281/zenodo.1183221} and \url{https://doi.org/10.5281/zenodo.1186823} respectively.

Correctness testing was conducted over all generated instances. This consisted of checking that the matchings produced by the approximation and IP-based algorithms adhered to $(1)$ \emph{capacity}: each student is assigned to a maximum of $1$ project, and each project and lecturer is not assigned to more students than their given capacity; and $(2)$ \emph{stability}: no blocking pair exists. Additionally, it was checked that each matching produced by the approximation algorithm was at least two-thirds the size of a maximum stable matching. All correctness tests passed successfully.

\subsection{Experimental results}
Experimental results can be seen in Figures \ref{spa-st-exp-plots-a}, \ref{spa-st-exp-plots-b} and \ref{spa-st-exp-plots-c}, and in Tables \ref{tab_res_size_sizes}, \ref{tab_res_ties_sizes}, \ref{tab_res_pref_sizes}, \ref{tab_res_size}, \ref{tab_res_ties}, \ref{tab_res_pref} and \ref{tab_res_scal_tim} of Appendix \ref{app_spa-st_more_exp_sec}. 

Figures \ref{spa-st-exp-plots-a}, \ref{spa-st-exp-plots-b} and \ref{spa-st-exp-plots-c} (with associated Tables \ref{tab_res_size_sizes}, \ref{tab_res_ties_sizes} and \ref{tab_res_pref_sizes}) show plots comparing the size of matching returned by the approximation algorithm with the sizes of minimum and maximum stable matchings and a $\frac{3}{2}$ bound, for Experiments $1$, $2$ and $3$ respectively. In each of these Figures, the median values of the size of stable matchings are plotted and a $90\%$ confidence interval is displayed using the $5$th and $95$th percentile measurements. 

Tables \ref{tab_res_size}, \ref{tab_res_ties} and \ref{tab_res_pref} show additional results for Experiments $1$, $2$ and $3$ respectively. From this point onwards an \emph{optimal} matching refers to a maximum stable matching. In these tables, `A' represents statistics relating to the approximation algorithm, and `Min' and `Max' represent statistics relating to the IP models to find a minimum and maximum stable matching, respectively. Column `Minimum A/Max' gives the minimum ratio of approximation algorithm matching size to optimal matching size that occurred, `\% A=Max' displays the percentage of times the approximation algorithm achieved an optimal result, and `\% A$\geq 0.98$Max' shows the percentage of times the approximation algorithm achieved a result at least $98$\% of optimal. The `Mean size' columns are somewhat self explanatory, with sub-columns `A/Max' and `Min/Max'  showing the mean approximation algorithm matching size and minimum stable matching size as a fraction of optimal. Finally, the mean time taken in milliseconds to run each algorithm per instance is given in the last three columns. Table \ref{tab_res_scal_tim} shows the scalability results for increasing instance sizes (Experiment $4$) and increasing preference list lengths (Experiment $5$). The `Instances completed' column indicates the number of instances that completed before timeout occurred.  The mean time taken is also shown, where instances that did not complete within the timeout time were said to have taken the maximum time of $30$ minutes. 

The main findings are summarised below.

\begin{itemize}
%	\item \textbf{In all experiments}
%	\begin{itemize}
		\item \emph{The approximation algorithm consistently far exceeds its $\frac{3}{2}$ bound.} Considering the column labelled `Minimum A/Max' in Tables \ref{tab_res_size}, \ref{tab_res_ties} and \ref{tab_res_pref}, %(showing the smallest ratio of approximation algorithm stable matching size divided by the maximum stable matching size for each instance set),
we see that the smallest value was within the SIZE1 instance set with a ratio of $0.9286$. This is well above the required bound of $\frac{2}{3}$.
		\item \emph{Approximation algorithm matchings are closer in size to the maximum stable matchings than to the minimum stable matchings.} The columns `A/Max' and `Min/Max' of Tables \ref{tab_res_size}, \ref{tab_res_ties} and \ref{tab_res_pref} %columns show the average size of stable matchings found by the approximation algorithm and by the \acrshort{ip} model for \acrshort{min-spa-st} as a proportion of the size of a maximum stable matching, respectively. These
        show that, on average, for each instance set, the approximation algorithm produces a solution that is within $98$\% of maximum and far closer to the maximum size than to the minimum size. This may also be seen each of the Figures \ref{spa-st-exp-plots-a}, \ref{spa-st-exp-plots-b} and \ref{spa-st-exp-plots-c}, where in general, the size of a stable matching output by the approximation algorithm, indicated in blue, is far closer to the size of a maximum stable matching (red) than the size of a minimum stable matching (green), in all cases.
%	\end{itemize} 

	\item \emph{Divergence in sizes of minimum and maximum stable matchings when increasing the probability of ties.} Unlike in Figures \ref{spa-st-exp-plots-a} and \ref{spa-st-exp-plots-c} which show similar trends for mean matching sizes produced by the approximation algorithm, and the minimum and maximum sized stable matchings, Figure \ref{spa-st-exp-plots-b} shows a marked divergence of the minimum stable matching size as the probability of ties increases. With no ties in the preference lists, all matchings are the same size (namely $284$), as expected from the theory. As the probability of ties increases, the maximum stable matching size and size of matching from the approximation algorithm increase steadily to $294.8$ and $299.5$ respectively. However, the minimum stable matching size decreases steadily to $254.8$. This behaviour may be explained by noticing that the higher the probability of ties, the larger the number of stable matchings that will exist. Hence there is a higher probability that the sizes of minimum and maximum stable matchings will diverge. Interestingly however, the decrease in size of minimum stable matchings does not have a noticeable effect on the approximation algorithm output, as matching sizes from this algorithm closely align with optimal.
	\item In Table \ref{tab_res_scal_tim}, Experiment 4 (SCALS) shows the number of instances solved within the $30$-minute timeout reduced from $10$ to $0$ for the IP-based algorithm for \acrshort{max-spa-st}. However, even for the largest instance set sizes the approximation algorithm was able to solve each instance on average within $21$ seconds. For Experiment 5 (SCALP), with a higher probability of ties and increasing preference list lengths, the IP-based algorithm to find a maximum stable matching was only able to solve all the instances of one instance set (SCALP2) within $30$ minutes each, however the approximation algorithm took less than $0.3$ seconds on average to return a solution for each instance. This shows that the approximation algorithm is useful for either larger or more complex instances than the IP-based algorithm can handle, motivating its use for real world scenarios.
\end{itemize}

	\begin{figure}
	\centering	
	\includegraphics[scale=0.7]{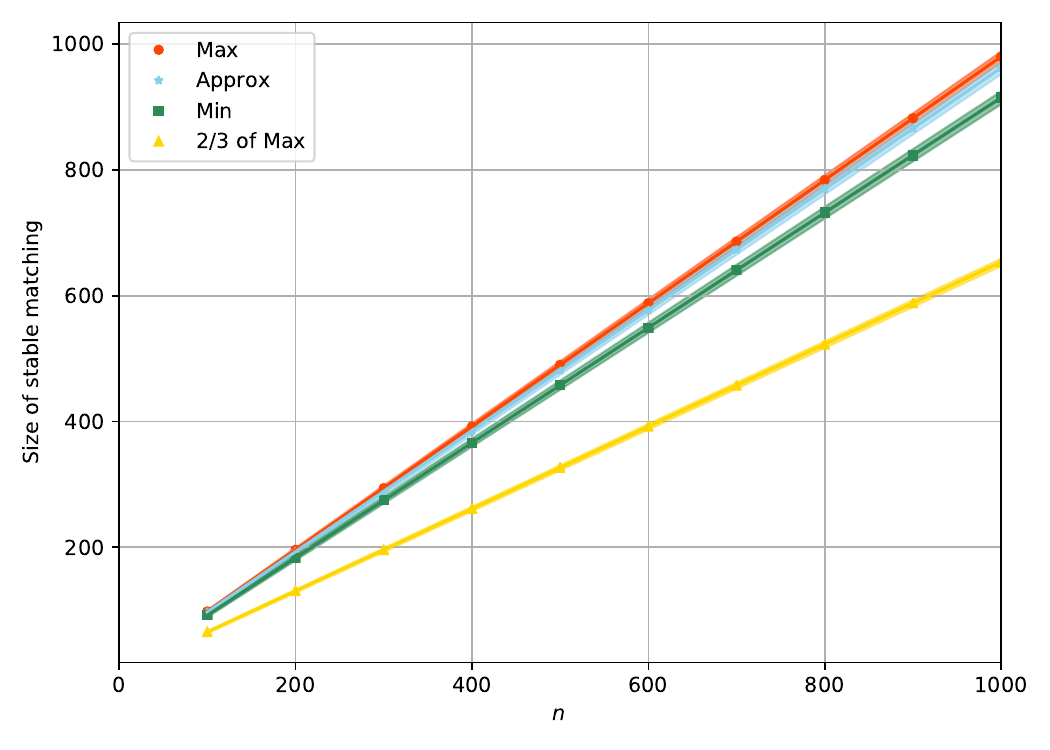}
    \caption[Plot of the size of stable matching returned by the approximation algorithm, and the minimum and maximum stable matching sizes]{Plot of the size of stable matching returned by the approximation algorithm, and the minimum and maximum stable matching sizes, with increasing $n$, where $n$ is the number of men or women. A second-order polynomial has been assumed for all best-fit lines.}
    \label{spa-st-exp-plots-a}
    \end{figure}

\begin{figure}	
\centering	
	\includegraphics[scale=0.7]{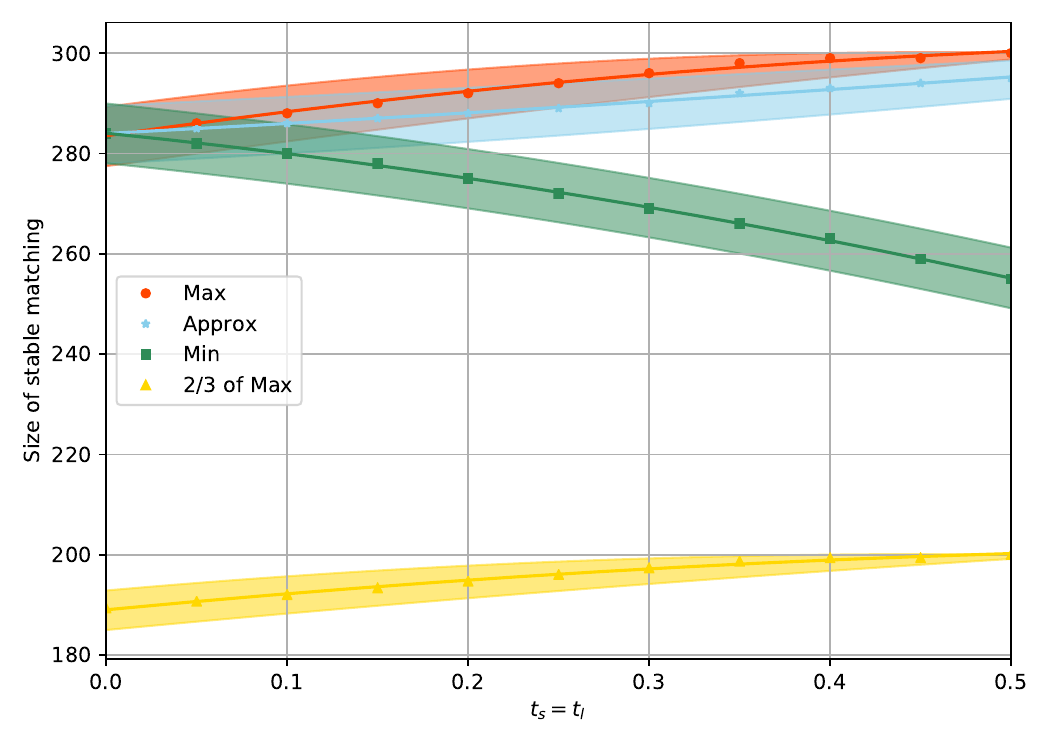}
    \caption[Plot of the size of stable matching returned by the approximation algorithm, and the minimum and maximum stable matching sizes, with increasing probability of ties.]{Plot of the size of stable matching returned by the approximation algorithm, and the minimum and maximum stable matching sizes, with increasing probability of ties. A second-order polynomial has been assumed for all best-fit lines.}
    \label{spa-st-exp-plots-b}
   \end{figure}

\begin{figure}	
	\centering	
	\includegraphics[scale=0.7]{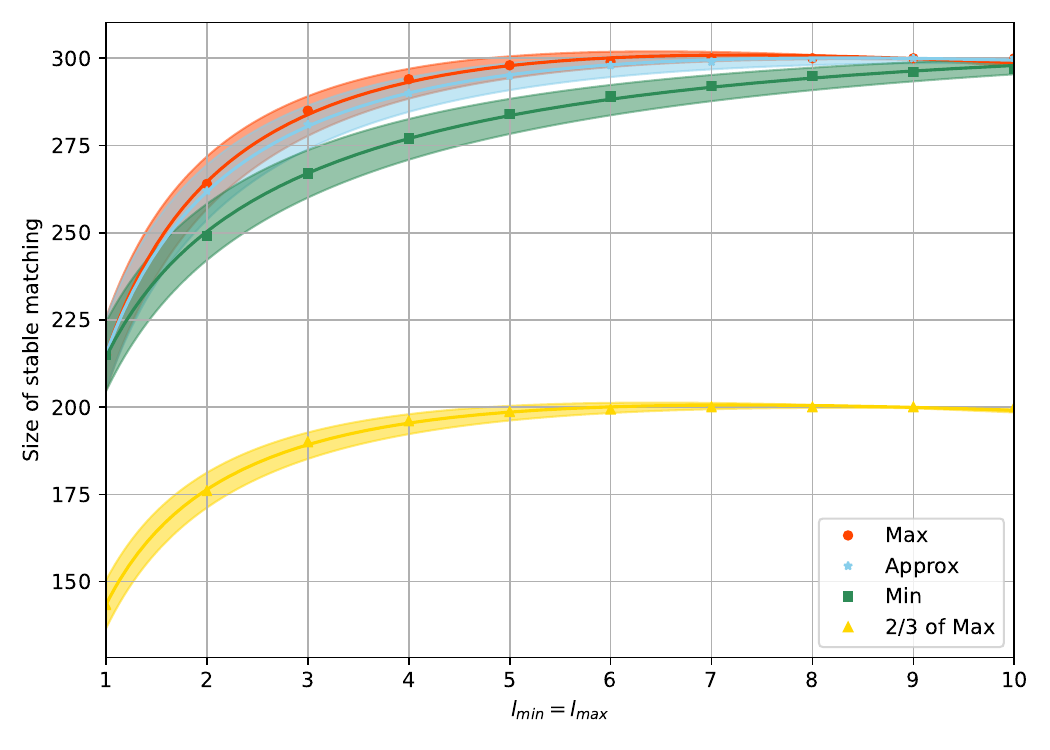}
    \caption[Plot of the size of stable matching returned by the approximation algorithm, and the minimum and maximum stable matching sizes, with increasing student preference list length.]{Plot of the size of stable matching returned by the approximation algorithm, and the minimum and maximum stable matching sizes, with increasing student preference list length. A second-order polynomial in $\log(l_{min})$ has been assumed for all best-fit lines.}
    \label{spa-st-exp-plots-c}
    \end{figure}

% \caption{Comparing the size of stable matching output by the approximation algorithm and the size of the minimum and maximum stable matching, as instance size, probability of ties and preference list lengths are increased.}
% \label{spa-st-exp-plots}
% \end{figure}    

\section{Future work}
\label{sec:conc}
This paper has described a $\frac{3}{2}$-approximation algorithm for \acrshort{max-spa-st}.  It remains open to describe an approximation algorithm that has a better performance guarantee, and/or to prove a stronger lower bound on the inapproximability of the problem than the current best bound of $\frac{33}{29}$ \cite{Yan07}. Further experiments could also measure the extent to which the order that students apply to projects in Algorithm \ref{Alg:APPROX_SPA-S_stable} affects the size of the stable matching generated.

The work in this paper has mainly focused on the size of stable matchings. However, it is possible for a stable matching to admit a \emph{blocking coalition}, where a permutation of student assignments could improve the allocations of the students and lecturers involved without harming anyone else. Since permutations of this kind cannot change the size of the matching they are not studied further here, but would be of interest for future work. 

\FloatBarrier
\bibliography{matching}

\appendix
% \vspace{5mm}
\noindent
\LARGE {\bf Appendix}
\normalsize
\section{Further discussion on Kir\'aly's $\frac{3}{2}$-approximation algorithm for {\sc smti}}
\label{app_spa-st_further_discussions}

Let $I$ be an instance of \acrshort{spa-st}. In Section \ref{cloning_spast_smti} we showed that converting $I$ to an instance of \acrshort{smti} and then using Kir\'aly's approximation algorithm does not necessarily result in a matching $M$ in $I$ that is a $\frac{3}{2}$-approximation to a maximum stable matching, even in the case where $M$ is stable. In this section we give some intuition as to why this is the case, by comparing the effect of conversion to an \acrshort{smti} instance and then use of Kir\'aly's approximation algorithm, when applied to instance $I_1$ from Section \ref{cloning_spast_smti} and when applied to a new instance $I_2$.

Figure \ref{spastapprox:ex2_before_clone} shows the \acrshort{spa-st} instance $I_2$ which is the same as the instance in Figure \ref{spastapprox:ex_before_clone} but with the capacities of projects $p_1$ and $p_2$ reduced to $1$. Figure \ref{spastapprox:ex2_after_clone2} shows $I_2$ converted into an \acrshort{smti} instance $I_2''$ using the same two-stage process as in Figure \ref{spastapprox:ex_beforeandafter_clone}. Finally, Table \ref{tab:tracekiraly2} shows the algorithm trace for instance $I_2''$ using Kir\'aly's $\frac{3}{2}$-approximation algorithm for \acrshort{smti} \cite{Kir12}. We can see from Figure \ref{tab:tracekiraly2} that, in contrast to the case for instance $I_1$, this process yields a stable matching $M=\{(s_1,p_3), (s_2,p_4), (s_4,p_2)\}$ in $I_2$ that is a $\frac{3}{2}$-approximation to a maximum stable matching (this must be the case since there are only $4$ students in the instance).

The main first difference in the traces can be seen on line $14$ of Table \ref{tab:tracekiraly} and line $11$ of Table \ref{tab:tracekiraly2}. On line $11$ of Table \ref{tab:tracekiraly2}, $m_2$ applies to $w_4$ as an advantaged man, giving him the ability to take $w_4$ from $m_3$. This shows the benefit of having a tie including $m_2$ and $m_3$ at the beginning of $w_4$'s list - either of these men being assigned to $w_4$ would be equally useful in a stable matching. Therefore allowing $m_2$ to take $w_4$ from $m_3$ gives $m_3$ a chance to get another partner, increasing the size of matching eventually attained.  On line $14$ of Table \ref{tab:tracekiraly}, $m_2$ becomes `stuck' on $w_5$, one of the women derived from a dummy resident. This stops $m_2$ being able to ever apply to $w_4$ as an advantaged man and the benefits of having $m_2$ and $m_3$ tied at the beginning of $w_4$'s preference list are not realised.

\begin{figure}[]
\centering
\begin{subfigure}[t]{0.3\textwidth}
Students preferences:\\
$s_1$: $p_3$\\
$s_2$: $p_4$ $p_1$ $p_2$\\
$s_3$: $p_3$\\
$s_4$: ($p_2$ $p_3$) $p_4$ $p_1$\\

Project details:\\
$p_1$: lecturer $l_1$, $c_1=1$\\
$p_2$: lecturer $l_1$, $c_2=1$\\
$p_3$: lecturer $l_2$, $c_3=2$\\
$p_4$: lecturer $l_2$, $c_4=1$\\
  
Lecturer preferences:\\
\begin{minipage}[b]{0.6\textwidth}
$l_1$: $s_2$ $s_4$\\
$l_2$: $s_4$ ($s_1$ $s_2$ $s_3$)\\
\end{minipage}
\begin{minipage}[b]{0.3\textwidth}
$d_1=2$\\
$d_2=2$\\
\end{minipage}

    \subcaption{\acrshort{spa-st} instance $I_2$. Same as instance $I$ in Figure \ref{spastapprox:ex_before_clone} except that projects $1$ and $2$ have an capacity of $1$.}
    \label{spastapprox:ex2_before_clone}
  \end{subfigure}
  \hspace*{\fill}
  \begin{subfigure}[t]{0.3\textwidth}
 
Resident preferences:\\
$r_1$: $h_3$ \\
$r_2$: $h_4$ $h_1$ $h_2$ \\
$r_3$: $h_3$ \\
$r_4$: ($h_2$ $h_3$) $h_4$ $h_1$ \\
$r_5$: ($h_3$ $h_4$) \\

Hospital preferences:\\
\begin{minipage}[b]{0.6\textwidth}
$h_1$: $r_2$ $r_4$ \\
$h_2$: $r_2$ $r_4$ \\
$h_3$: $r_5$ $r_4$ ($r_1$ $r_3$) \\
$h_4$: $r_5$ $r_4$ $r_2$\newline\newline\newline\newline
\end{minipage}
\begin{minipage}[b]{0.3\textwidth}
$c'_1=1$\\
$c'_2=1$\\
$c'_3=2$\\
$c'_4=1$\newline\newline\newline\newline
\end{minipage}

   \subcaption{\acrshort{hrt} instance $I_2'$ converted from the \acrshort{spa-st} instance in Figure \ref{spastapprox:ex2_before_clone}.}
    \label{spastapprox:ex2_after_clone1}
  \end{subfigure}
  \hspace*{\fill}
    \begin{subfigure}[t]{0.3\textwidth}

Women's preferences:\\
$w_1$: ($m_3$ $m_5$) \\
$w_2$: $m_4$ $m_1$ $m_2$ \\
$w_3$: ($m_3$ $m_5$) \\
$w_4$: ($m_2$ $m_3$ $m_5$) $m_4$ $m_1$ \\
$w_5$: ($m_3$ $m_5$ $m_4$) \\

Men's preferences:\\
$m_1$: $w_2$ $w_4$\\
$m_2$: $w_2$ $w_4$\\
$m_3$: $w_5$ $w_4$ ($w_1$ $w_3$)\\
$m_4$: $w_5$ $w_4$ $w_2$\\
$m_5$: $w_5$ $w_4$ ($w_1$ $w_3$)\newline\newline\newline

   \subcaption{\acrshort{smti} instance $I_2''$ converted from the \acrshort{hrt} instance in Figure \ref{spastapprox:ex2_after_clone1}.}
    \label{spastapprox:ex2_after_clone2}
  \end{subfigure}
  \caption[Conversion of an {\sc spa-st} instance to an {\sc smti}  instance.]{Conversion of an \acrshort{spa-st} instance to an \acrshort{smti}  instance.}
  \label{spastapprox:ex2_beforeandafter_clone}
    \end{figure}

\begin{table}[tbp] \centering\begin{tabular}{ p{0.5cm}p{8cm} | p{0.7cm} p{0.7cm}p{0.7cm}p{0.7cm}p{0.7cm}}\hline\hline 
& Action & $m_1$ & $m_2$ & $m_3$ & $m_4$ & $m_5$  \\ 
\hline 
$1$ & $m_5$ applies to $w_5$, accepted &  &  &  &  & $w_5$ \\
$2$ & $m_4$ applies to $w_5$, rejected, $m_4$ removes $w_5$ &  &  &  &  & $w_5$ \\
$3$ & $m_4$ applies to $w_4$, accepted &  &  &  & $w_4$ & $w_5$ \\
$4$ & $m_3$ applies to $w_5$, rejected, $m_3$ removes $w_5$ &  &  &  & $w_4$ & $w_5$ \\
$5$ & $m_3$ applies to $w_4$, accepted, $m_4$ removes $w_4$ &  &  & $w_4$ &  & $w_5$ \\
$6$ & $m_4$ applies to $w_2$, accepted &  &  & $w_4$ & $w_2$ & $w_5$ \\
$7$ & $m_2$ applies to $w_2$, rejected, $m_2$ removes $w_2$ &  &  & $w_4$ & $w_2$ & $w_5$ \\
$8$ & $m_2$ applies to $w_4$, rejected, $m_2$ removes $w_4$ &  &  & $w_4$ & $w_2$ & $w_5$ \\
$9$ & $m_2$ advantaged &  &  & $w_4$ & $w_2$ & $w_5$ \\

$10$ & $m_2$ applies to $w_2$, rejected, $m_2$ removes $w_2$ &  &  & $w_4$ & $w_2$ & $w_5$ \\
$11$ & $m_2$ applies to $w_4$, accepted, $m_3$ removes $w_4$ &  & $w_4$ &  & $w_2$ & $w_5$ \\
$12$ & $m_3$ applies $w_1$, accepted &  & $w_4$ & $w_1$ & $w_2$ & $w_5$ \\
$13$ & $m_1$ applies to $w_2$, rejected, $m_1$ removes $w_2$ &  & $w_4$ & $w_1$ & $w_2$ & $w_5$ \\
$14$ & $m_1$ applies to $w_4$, rejected, $m_1$ removes $w_4$ &  & $w_4$ & $w_1$ & $w_2$ & $w_5$ \\
$15$ & $m_1$ advantaged &  & $w_4$ & $w_1$ & $w_2$ & $w_5$ \\
$16$ & $m_1$ applies to $w_2$, rejected, $m_1$ removes $w_2$ &  & $w_4$ & $w_1$ & $w_2$ & $w_5$ \\
$17$ & $m_1$ applies to $w_4$, rejected, $m_1$ removes $w_4$ &  & $w_4$ & $w_1$ & $w_2$ & $w_5$ \\
$18$ & $m_1$ inactive & $-$ & $w_4$ & $w_1$ & $w_2$ & $w_5$ \\

   \hline\hline \end{tabular} \caption[Trace of running Kir\'aly's {\sc smti} $\frac{3}{2}$-approximation algorithm for instance $I_2''$]{Trace of running Kir\'aly's \acrshort{smti} $\frac{3}{2}$-approximation algorithm for instance $I_2''$ in Figure \ref{spastapprox:ex_after_clone2}. In this table, the phrase ``$m_i$ removes $w_j$'' indicates that man $m_i$ removes woman $w_j$ from their preference list.} \label{tab:tracekiraly2} \end{table} 

\newpage
\section{Experimental results tables}
\label{app_spa-st_more_exp_sec}

In this Section we present tables of results for the experiments conducted in Section \ref{sec_spa-st_more_sec_exp}. 
\FloatBarrier
%  \begin{landscape}
%  \thispagestyle{empty}
% \input{Appendix_SM_profile/tables/latex_table_RM.txt}
%  \end{landscape}

\begin{table}[h!]
 \centerline{\begin{tabular}{ R{1.3cm} | R{1.1cm} R{1.1cm} R{1.1cm} R{1.8cm} R{1.1cm} R{1.1cm} R{1.8cm} R{1.1cm} R{1.1cm} }\hline\hline & \multicolumn{3}{c}{Approx} &  \multicolumn{3}{c}{Minimum}  &  \multicolumn{3}{c}{Maximum} \\ 
Case & Median & $5$th & $95$th & Median & $5$th & $95$th & Median & $5$th & $95$th \\ 
\hline SIZE1 & $97.0$ & $93.0$ & $99.0$ & $92.0$ & $89.0$ & $95.0$ & $98.0$ & $95.0$ & $100.0$ \\ 
 SIZE2 & $193.0$ & $188.0$ & $197.0$ & $183.0$ & $179.0$ & $188.0$ & $196.0$ & $191.0$ & $199.0$ \\ 
 SIZE3 & $289.0$ & $283.0$ & $294.0$ & $275.0$ & $269.0$ & $280.0$ & $294.0$ & $288.0$ & $298.0$ \\ 
 SIZE4 & $385.0$ & $379.0$ & $391.0$ & $366.0$ & $360.0$ & $373.0$ & $392.0$ & $386.0$ & $397.0$ \\ 
 SIZE5 & $481.0$ & $474.0$ & $488.0$ & $458.0$ & $450.0$ & $465.0$ & $490.0$ & $483.0$ & $496.0$ \\ 
 SIZE6 & $577.0$ & $570.0$ & $585.0$ & $549.0$ & $541.0$ & $557.0$ & $588.0$ & $580.0$ & $594.0$ \\ 
 SIZE7 & $673.0$ & $665.0$ & $681.0$ & $641.0$ & $632.0$ & $649.0$ & $686.0$ & $678.0$ & $693.0$ \\ 
 SIZE8 & $770.0$ & $761.0$ & $778.0$ & $732.0$ & $723.0$ & $741.0$ & $784.0$ & $775.0$ & $791.0$ \\ 
 SIZE9 & $866.0$ & $856.0$ & $875.0$ & $823.0$ & $813.0$ & $833.0$ & $882.0$ & $873.0$ & $890.0$ \\ 
 SIZE10 & $962.0$ & $952.0$ & $971.0$ & $915.0$ & $904.0$ & $925.0$ & $980.0$ & $970.0$ & $988.0$ \\ 
 \hline\hline \end{tabular}} 
\caption{Comparison of the size of the stable matching returned by the approximation algorithm, and the minimum and maximum stable matching sizes, with increasing instance size.} \label{tab_res_size_sizes} \end{table}

\begin{table}[]
 \centerline{\begin{tabular}{ R{1.3cm} | R{1.1cm} R{1.1cm} R{1.1cm} R{1.8cm} R{1.1cm} R{1.1cm} R{1.8cm} R{1.1cm} R{1.1cm} }\hline\hline & \multicolumn{3}{c}{Approx} &  \multicolumn{3}{c}{Minimum}  &  \multicolumn{3}{c}{Maximum} \\ 
Case & Median & $5$th & $95$th & Median & $5$th & $95$th & Median & $5$th & $95$th \\ 
\hline TIES1 & $284.0$ & $278.0$ & $290.0$ & $284.0$ & $278.0$ & $290.0$ & $284.0$ & $278.0$ & $290.0$ \\ 
 TIES2 & $285.0$ & $279.0$ & $290.0$ & $282.0$ & $276.0$ & $288.0$ & $286.0$ & $280.0$ & $291.0$ \\ 
 TIES3 & $286.0$ & $280.0$ & $291.0$ & $280.0$ & $274.0$ & $286.0$ & $288.0$ & $282.0$ & $293.0$ \\ 
 TIES4 & $287.0$ & $281.0$ & $292.0$ & $278.0$ & $272.0$ & $283.0$ & $290.0$ & $284.0$ & $295.0$ \\ 
 TIES5 & $288.0$ & $282.0$ & $293.0$ & $275.0$ & $269.0$ & $281.0$ & $292.0$ & $287.0$ & $297.0$ \\ 
 TIES6 & $289.0$ & $284.0$ & $294.0$ & $272.0$ & $266.0$ & $278.0$ & $294.0$ & $289.0$ & $298.0$ \\ 
 TIES7 & $290.0$ & $285.0$ & $295.0$ & $269.0$ & $263.0$ & $275.0$ & $296.0$ & $291.0$ & $299.0$ \\ 
 TIES8 & $292.0$ & $286.0$ & $296.0$ & $266.0$ & $260.0$ & $272.0$ & $298.0$ & $294.0$ & $300.0$ \\ 
 TIES9 & $293.0$ & $288.0$ & $297.0$ & $263.0$ & $257.0$ & $269.0$ & $299.0$ & $296.0$ & $300.0$ \\ 
 TIES10 & $294.0$ & $289.0$ & $298.0$ & $259.0$ & $253.0$ & $265.0$ & $299.0$ & $297.0$ & $300.0$ \\ 
 TIES11 & $295.0$ & $291.0$ & $298.0$ & $255.0$ & $249.0$ & $261.0$ & $300.0$ & $298.0$ & $300.0$ \\ 
 \hline\hline \end{tabular}} 
\caption{Comparison of the size of the stable matching returned by the approximation algorithm, and the minimum and maximum stable matching sizes, with increasing probability of ties.} \label{tab_res_ties_sizes} \end{table}

\begin{table}[]
 \centerline{\begin{tabular}{ R{1.3cm} | R{1.1cm} R{1.1cm} R{1.1cm} R{1.8cm} R{1.1cm} R{1.1cm} R{1.8cm} R{1.1cm} R{1.1cm} }\hline\hline & \multicolumn{3}{c}{Approx} &  \multicolumn{3}{c}{Minimum}  &  \multicolumn{3}{c}{Maximum} \\ 
Case & Median & $5$th & $95$th & Median & $5$th & $95$th & Median & $5$th & $95$th \\ 
\hline PREF1 & $215.0$ & $205.0$ & $225.0$ & $215.0$ & $205.0$ & $225.0$ & $215.0$ & $205.0$ & $225.0$ \\ 
 PREF2 & $262.0$ & $254.0$ & $270.0$ & $249.0$ & $241.0$ & $257.0$ & $264.0$ & $256.0$ & $272.0$ \\ 
 PREF3 & $281.0$ & $274.0$ & $287.0$ & $267.0$ & $260.0$ & $273.0$ & $285.0$ & $278.0$ & $291.0$ \\ 
 PREF4 & $290.0$ & $285.0$ & $295.0$ & $277.0$ & $271.0$ & $283.0$ & $294.0$ & $289.0$ & $298.0$ \\ 
 PREF5 & $295.0$ & $291.0$ & $298.0$ & $284.0$ & $279.0$ & $289.0$ & $298.0$ & $295.0$ & $300.0$ \\ 
 PREF6 & $298.0$ & $294.0$ & $300.0$ & $289.0$ & $284.0$ & $293.0$ & $299.0$ & $297.0$ & $300.0$ \\ 
 PREF7 & $299.0$ & $296.0$ & $300.0$ & $292.0$ & $288.0$ & $296.0$ & $300.0$ & $299.0$ & $300.0$ \\ 
 PREF8 & $300.0$ & $298.0$ & $300.0$ & $295.0$ & $291.0$ & $297.0$ & $300.0$ & $299.0$ & $300.0$ \\ 
 PREF9 & $300.0$ & $299.0$ & $300.0$ & $296.0$ & $293.0$ & $299.0$ & $300.0$ & $299.0$ & $300.0$ \\ 
 PREF10 & $300.0$ & $299.0$ & $300.0$ & $297.0$ & $295.0$ & $299.0$ & $300.0$ & $300.0$ & $300.0$ \\ 
 \hline\hline \end{tabular}} 
\caption{Comparison of the size of the stable matching returned by the approximation algorithm, and the minimum and maximum stable matching sizes, with increasing student preference list length.} \label{tab_res_pref_sizes} \end{table}

% \newgeometry{bottom=2cm, left=1.8cm, right=1.8cm}

\begin{landscape}
\begin{table}[tbp]
 \centerline{\begin{tabular}{ R{1.3cm} | R{1.6cm} R{1.6cm} R{1.6cm} R{1.8cm} R{1.3cm} R{1.3cm} R{1.3cm} R{1.6cm} R{1.8cm} R{1.3cm} R{1.3cm} }\hline\hline & Minimum &  &  \%A$\geq$ & \multicolumn{5}{c}{Mean size} & \multicolumn{3}{c}{Mean time (ms)}\\ 
Case & A/Max  & \% A=Max  & $0.98 $Max  & A & Min & Max & A/Max & Min/Max & A & Min & Max \\ 
\hline SIZE1 & $0.9286$ & $17.8$ & $62.7$ & $96.4$ & $92.0$ & $97.8$ & $0.9859$ & $0.9408$ & $43.3$ & $147.6$ & $137.8$ \\ 
 SIZE2 & $0.9585$ & $1.6$ & $62.6$ & $192.6$ & $183.4$ & $195.7$ & $0.9840$ & $0.9373$ & $51.2$ & $230.6$ & $210.6$ \\ 
 SIZE3 & $0.9556$ & $0.1$ & $63.7$ & $288.7$ & $274.9$ & $293.7$ & $0.9831$ & $0.9361$ & $56.6$ & $346.4$ & $313.4$ \\ 
 SIZE4 & $0.9644$ & $0.0$ & $65.6$ & $384.9$ & $366.4$ & $391.7$ & $0.9827$ & $0.9354$ & $59.7$ & $488.7$ & $429.3$ \\ 
 SIZE5 & $0.9654$ & $0.0$ & $66.5$ & $481.0$ & $457.7$ & $489.6$ & $0.9824$ & $0.9349$ & $62.8$ & $660.3$ & $555.6$ \\ 
 SIZE6 & $0.9641$ & $0.0$ & $66.8$ & $577.2$ & $549.3$ & $587.7$ & $0.9821$ & $0.9346$ & $66.4$ & $862.3$ & $713.0$ \\ 
 SIZE7 & $0.9679$ & $0.0$ & $65.4$ & $673.3$ & $640.5$ & $685.7$ & $0.9819$ & $0.9341$ & $69.8$ & $1127.8$ & $900.6$ \\ 
 SIZE8 & $0.9684$ & $0.0$ & $67.4$ & $769.5$ & $732.0$ & $783.8$ & $0.9818$ & $0.9340$ & $73.0$ & $1437.3$ & $1098.2$ \\ 
 SIZE9 & $0.9653$ & $0.0$ & $68.6$ & $865.6$ & $823.4$ & $881.7$ & $0.9817$ & $0.9339$ & $76.5$ & $1784.3$ & $1343.9$ \\ 
 SIZE10 & $0.9701$ & $0.0$ & $68.0$ & $961.7$ & $914.7$ & $979.7$ & $0.9816$ & $0.9337$ & $86.6$ & $2281.2$ & $1651.0$ \\ 
 \hline\hline \end{tabular}} 
\caption{Comparisons of the matching output by the approximation algorithm, and IP model implementation outputs, with increasing instance size.} \label{tab_res_size} \end{table}

\begin{table}[tbp]
 \centerline{\begin{tabular}{ R{1.3cm} | R{1.6cm} R{1.6cm} R{1.6cm} R{1.8cm} R{1.3cm} R{1.3cm} R{1.3cm} R{1.6cm} R{1.8cm} R{1.3cm} R{1.3cm} }\hline\hline & Minimum &  &  \%A$\geq$ & \multicolumn{5}{c}{Mean size} & \multicolumn{3}{c}{Mean time (ms)}\\ 
Case & A/Max  & \% A=Max  & $0.98 $Max  & A & Min & Max & A/Max & Min/Max & A & Min & Max \\ 
\hline TIES1 & $1.0000$ & $100.0$ & $100.0$ & $284.0$ & $284.0$ & $284.0$ & $1.0000$ & $1.0000$ & $59.2$ & $184.0$ & $186.9$ \\ 
 TIES2 & $0.9792$ & $38.0$ & $100.0$ & $284.9$ & $282.0$ & $285.8$ & $0.9968$ & $0.9866$ & $61.2$ & $192.4$ & $194.7$ \\ 
 TIES3 & $0.9722$ & $12.1$ & $99.3$ & $285.9$ & $279.9$ & $287.9$ & $0.9933$ & $0.9722$ & $61.7$ & $201.0$ & $203.1$ \\ 
 TIES4 & $0.9655$ & $3.4$ & $95.2$ & $287.0$ & $277.6$ & $289.9$ & $0.9900$ & $0.9576$ & $62.3$ & $213.3$ & $214.5$ \\ 
 TIES5 & $0.9626$ & $1.0$ & $82.5$ & $288.0$ & $275.1$ & $291.9$ & $0.9865$ & $0.9423$ & $62.9$ & $234.3$ & $231.0$ \\ 
 TIES6 & $0.9558$ & $0.4$ & $66.7$ & $289.2$ & $272.4$ & $294.0$ & $0.9837$ & $0.9266$ & $64.2$ & $274.2$ & $260.6$ \\ 
 TIES7 & $0.9486$ & $0.2$ & $52.9$ & $290.3$ & $269.4$ & $295.7$ & $0.9816$ & $0.9111$ & $64.3$ & $358.3$ & $311.3$ \\ 
 TIES8 & $0.9527$ & $0.2$ & $46.4$ & $291.4$ & $266.2$ & $297.2$ & $0.9803$ & $0.8957$ & $64.2$ & $577.3$ & $380.7$ \\ 
 TIES9 & $0.9467$ & $0.2$ & $50.4$ & $292.5$ & $262.7$ & $298.3$ & $0.9805$ & $0.8804$ & $65.2$ & $1234.1$ & $427.5$ \\ 
 TIES10 & $0.9529$ & $0.5$ & $61.9$ & $293.7$ & $258.9$ & $299.1$ & $0.9821$ & $0.8656$ & $59.6$ & $2903.4$ & $409.1$ \\ 
 TIES11 & $0.9467$ & $1.0$ & $74.2$ & $294.8$ & $254.8$ & $299.5$ & $0.9842$ & $0.8506$ & $60.4$ & $5756.9$ & $377.4$ \\ 
 \hline\hline \end{tabular}} 
\caption{Comparisons of the matching output by the approximation algorithm, and IP model implementation outputs, with increasing probability of ties.} \label{tab_res_ties} \end{table}

\begin{table}[tbp]
 \centerline{\begin{tabular}{ R{1.3cm} | R{1.6cm} R{1.6cm} R{1.6cm} R{1.8cm} R{1.3cm} R{1.3cm} R{1.3cm} R{1.6cm} R{1.8cm} R{1.3cm} R{1.3cm} }\hline\hline & Minimum &  &  \%A$\geq$ & \multicolumn{5}{c}{Mean size} & \multicolumn{3}{c}{Mean time (ms)}\\ 
Case & A/Max  & \% A=Max  & $0.98 $Max  & A & Min & Max & A/Max & Min/Max & A & Min & Max \\ 
\hline PREF1 & $1.0000$ & $100.0$ & $100.0$ & $215.0$ & $215.0$ & $215.0$ & $1.0000$ & $1.0000$ & $74.3$ & $107.5$ & $105.1$ \\ 
 PREF2 & $0.9699$ & $12.3$ & $99.0$ & $262.1$ & $249.1$ & $264.1$ & $0.9926$ & $0.9433$ & $67.5$ & $133.8$ & $128.7$ \\ 
 PREF3 & $0.9617$ & $1.2$ & $84.0$ & $280.9$ & $266.4$ & $284.7$ & $0.9867$ & $0.9361$ & $68.1$ & $181.4$ & $174.0$ \\ 
 PREF4 & $0.9623$ & $1.0$ & $82.8$ & $290.0$ & $277.0$ & $293.9$ & $0.9866$ & $0.9426$ & $69.1$ & $249.7$ & $242.6$ \\ 
 PREF5 & $0.9661$ & $4.2$ & $95.1$ & $294.8$ & $283.9$ & $297.7$ & $0.9902$ & $0.9537$ & $68.3$ & $346.7$ & $340.3$ \\ 
 PREF6 & $0.9732$ & $15.7$ & $99.5$ & $297.3$ & $288.7$ & $299.1$ & $0.9940$ & $0.9653$ & $66.1$ & $472.4$ & $440.6$ \\ 
 PREF7 & $0.9767$ & $36.2$ & $100.0$ & $298.7$ & $292.1$ & $299.7$ & $0.9966$ & $0.9746$ & $64.5$ & $638.3$ & $550.9$ \\ 
 PREF8 & $0.9833$ & $58.2$ & $100.0$ & $299.3$ & $294.4$ & $299.9$ & $0.9982$ & $0.9819$ & $64.1$ & $811.9$ & $660.3$ \\ 
 PREF9 & $0.9866$ & $75.5$ & $100.0$ & $299.7$ & $296.1$ & $299.9$ & $0.9991$ & $0.9873$ & $63.4$ & $1032.2$ & $789.1$ \\ 
 PREF10 & $0.9900$ & $87.3$ & $100.0$ & $299.8$ & $297.4$ & $300.0$ & $0.9995$ & $0.9913$ & $104.3$ & $1239.4$ & $931.0$ \\ 
 \hline\hline \end{tabular}} 
\caption{Comparisons of the matching output by the approximation algorithm, and IP model outputs, with increasing student preference list length.} \label{tab_res_pref} \end{table}

\begin{table}[tbp] \centerline{\begin{tabular}{ R{1.5cm} | R{1.5cm} R{1.5cm} R{1.5cm} R{2.5cm} R{2cm} R{2cm} }\hline\hline & \multicolumn{3}{c}{Instances completed} & \multicolumn{3}{c}{Mean time (ms)}\\ 
Case & A & Min & Max & A & Min & Max \\ 
\hline SCALS1 & $10$ & $10$ & $10$ & $1393.8$ & $127980.3$ & $227764.3$ \\ 
 SCALS2 & $10$ & $10$ & $9$ & $5356.7$ & $353272.3$ & $1166441.0$ \\ 
 SCALS3 & $10$ & $10$ & $0$ & $13095.3$ & $785421.2$ & $1800000.0$ \\ 
 SCALS4 & $10$ & $7$ & $0$ & $18883.5$ & $1283453.5$ & $1800000.0$ \\ 
 SCALS5 & $10$ & $7$ & $0$ & $20993.0$ & $1455410.1$ & $1800000.0$ \\ 
\hline SCALP1 & $10$ & $0$ & $9$ & $193.3$ & $1800000.0$ & $264818.6$ \\ 
 SCALP2 & $10$ & $1$ & $10$ & $189.4$ & $1762884.4$ & $631225.2$ \\ 
 SCALP3 & $10$ & $0$ & $3$ & $196.6$ & $1800000.0$ & $1524675.3$ \\ 
 SCALP4 & $10$ & $0$ & $1$ & $248.5$ & $1800000.0$ & $1779420.1$ \\ 
 SCALP5 & $10$ & $0$ & $0$ & $283.7$ & $1800000.0$ & $1800000.0$ \\ 
 SCALP6 & $10$ & $0$ & $0$ & $288.4$ & $1800000.0$ & $1800000.0$ \\ 
 \hline\hline \end{tabular}} \caption{Scalability experiment results.} \label{tab_res_scal_tim} \end{table} 
\end{landscape}

\end{document}